\documentclass[journal,final]{IEEEtran}
\ifCLASSINFOpdf
\else
\fi

\usepackage{authblk}
\usepackage{algorithm}
\usepackage[noend]{algcompatible}

\usepackage{amsmath,amsfonts,bm}

\def\1{\bm{1}}

\def\vone{{\bm{1}}}

\def\va{{\bm{a}}}

\def\ve{{\bm{e}}}

\def\vh{{\bm{h}}}

\def\vl{{\bm{l}}}

\def\vr{{\bm{r}}}
\def\vs{{\bm{s}}}

\def\vv{{\bm{v}}}

\def\vx{{\bm{x}}}

\def\vz{{\bm{z}}}

\def\mA{{\bm{A}}}
\def\mB{{\bm{B}}}

\def\mE{{\bm{E}}}
\def\mF{{\bm{F}}}

\def\mH{{\bm{H}}}
\def\mI{{\bm{I}}}

\def\mR{{\bm{R}}}
\def\mS{{\bm{S}}}

\def\mX{{\bm{X}}}

\DeclareMathAlphabet{\mathsfit}{\encodingdefault}{\sfdefault}{m}{sl}
\SetMathAlphabet{\mathsfit}{bold}{\encodingdefault}{\sfdefault}{bx}{n}

\newcommand{\R}{\mathbb{R}}

\DeclareMathOperator*{\argmax}{arg\,max}
\DeclareMathOperator*{\argmin}{arg\,min}

\usepackage{pifont}
\usepackage{subfigure}
\usepackage{hyperref}
\usepackage{url}
\usepackage{enumitem,bbm,mathrsfs}
\usepackage{multirow,booktabs}
\usepackage{mathtools}
\usepackage[normalem]{ulem}

\usepackage{chngcntr}
\usepackage{apptools}
\AtAppendix{\counterwithin{lemma}{section}}

\usepackage{bookmark}
\usepackage{mathabx}

\newcommand{\subparagraph}{}    
\usepackage{titlesec}

\usepackage{cite}
\usepackage{amsmath}
\usepackage{amsthm}
\usepackage{amssymb}
\usepackage{tikz}
\usepackage{xcolor}
\usetikzlibrary{arrows}
\usepackage{hyperref}

\definecolor{cyan_a}{rgb}{0,0.7490196,0.7686275}

\newcommand{\remove}[1]{}
\newtheorem{lemma}{Lemma}
\newtheorem{theorem}{Theorem}

\newtheorem{definition}{Definition}

\def\R{\mathbb{R}}

\def\va{\mathbf{a}}

\newcommand{\norm}[1]{\left\|#1\right\|}

\newcommand{\set}[1]{\{#1\}}

\newenvironment{noiseless_case}[1][\ding{113} Noiseless case:\\]{\begin{trivlist}\item[\hskip \labelsep {\bfseries #1}]}{\end{trivlist}}
\newenvironment{noisy_case}[1][\ding{113} Noisy case:\\]{\begin{trivlist}\item[\hskip \labelsep {\bfseries #1}]}{\end{trivlist}}

\newtheorem{problem}{Problem}

\newtheoremstyle{solution}%
  {\topsep}{\topsep}{\normalfont}{}%
  {\itshape}{.}{5pt}{}

\begin{document}
%
\title{Reconstructing Point Sets from\\ Distance Distributions}
%
%
%

\author{Shuai~Huang
        and~Ivan~Dokmani\'c,~\IEEEmembership{Member,~IEEE}
\thanks{\copyright 2021 IEEE. Personal use of this material is permitted. Permission from IEEE must be obtained for all other uses, in any current or future media, including reprinting/republishing this material for advertising or promotional purposes, creating new collective  works,  for  resale  or  redistribution  to  servers  or  lists,  or  reuse  of  any  copyrighted  component  of  this  work  in  other works.}
\thanks{This work is supported by National Science Foundation under Grant CIF-1817577. The authors are affiliated with the Coordinated Science Laboratory, University of Illinois at Urbana-Champaign, Urbana, IL 61801 (e-mail: shuai.huang@emory.edu, dokmanic@illinois.edu).

}}

\maketitle

\begin{abstract}
We address the problem of reconstructing a set of points on a line or a loop from their unassigned noisy pairwise distances. When the points lie on a line, the problem is known as the turnpike; when they are on a loop, it is known as the beltway. We approximate the problem by discretizing the domain and representing the $N$ points via an $N$-hot encoding, which is a density supported on the discretized domain. We show how the distance distribution is then simply a collection of quadratic functionals of this density and propose to recover the point locations so that the estimated distance distribution matches the measured distance distribution. This can be cast as a constrained nonconvex optimization problem which we solve using projected gradient descent with a suitable spectral initializer. We derive conditions under which the proposed distance distribution matching approach locally converges to a global optimizer at a linear rate. Compared to the conventional backtracking approach, our method jointly reconstructs all the point locations and is robust to noise in the measurements. We substantiate these claims with state-of-the-art performance across a number of numerical experiments. Our method is the first practical approach to solve the large-scale noisy beltway problem where the points lie on a loop.
\end{abstract}


%
\IEEEpeerreviewmaketitle

\section{Introduction}
\label{sec:intro}
In this paper we address the problem of reconstructing the geometry of $N$ points from their unassigned pairwise distances in the one-dimensional case where the points lie on a line or a loop. In most distance geometry problems (DGP), one is given an indexed list of $\tbinom{N}{2}$ pairwise distances $\mathcal{D}=\big(d_k,\ 1\leq k\leq \tbinom{N}{2}\big)$, where $d_k$ is the distance between the $k$-th pair of points and could contain noise. In standard, \textit{assigned} problems, every distance $d_k$ is assigned to a pair of points $\set{ u_m,u_n }$ from $\mathcal{U}=\left( u_n, \ 1\leq n\leq N \right)$. Put differently, we know an assignment map $\mathscr{M}(k)=\set{m,n}$ such that $d_k=\norm{u_m-u_n}$. When the distances are exact, having the assignments allows us to construct the distance matrix, which in turn allows us to employ classical techniques based on eigendecomposition such as multidimensional scaling \cite{Torgerson1952} to estimate the relative point locations $\mathcal{U}$.

On the other hand, in the \textit{unassigned} distance geometry problem (uDGP) \cite{Duxbury:2016} addressed in this paper, the correspondences between the distances and pairs of points are unknown: the assignment $\mathscr{M}(k)$ is not available. Instead of a list, we only have the multiset\footnote{To allow for repeated distances.} $\mathcal{D}$ to work with. We must recover both the point locations and the assignments of the distances to pairs of points. Fig. \ref{fig:1d_uDGP} illustrates the two related reconstruction problems in 1D. When the $N$ points lie on a line, the problem is known in computer science as ``the turnpike problem'' \cite{Shamos:1978:CG:908431,Skiena:1990,Dakic:2000}. The multiset $\mathcal{D}$ contains $\tbinom{N}{2}$ distances from $u_m$ to $u_n$.

When the $N$ points lie on a loop, we have ``the beltway problem'' \cite{Skiena:1990,Lemke2003}. Assuming that the distances between pairs of points are measured in the clockwise direction and the length of the loop is $L$, the distance $d(u_m\rightarrow u_n)$ from $u_m$ to $u_n$ and the distance $d(u_n\rightarrow u_m)$ from $u_n$ to $u_m$ satisfy $d(u_m\rightarrow u_n)+d(u_n\rightarrow u_m)=L$. The multiset $\mathcal{H}$ then contains $N(N-1)$ distances.

\begin{figure}[tbp]
\centering
\vspace{-2em}
\subfigure{
\label{fig:1d_uDGP_line}
\includegraphics[height=1in]{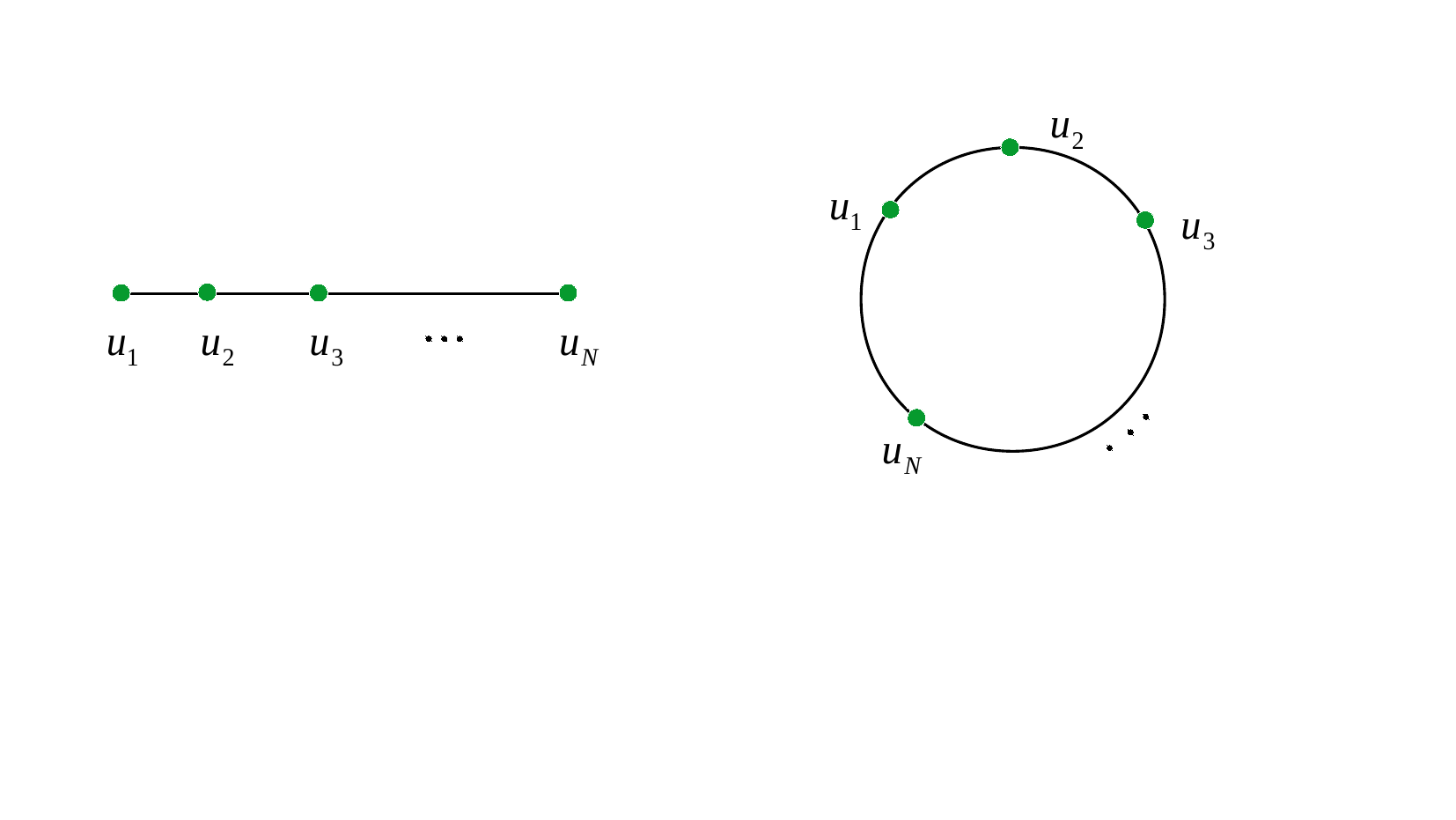}}
\subfigure{
\label{fig:1d_uDGP_loop}
\includegraphics[height=1in]{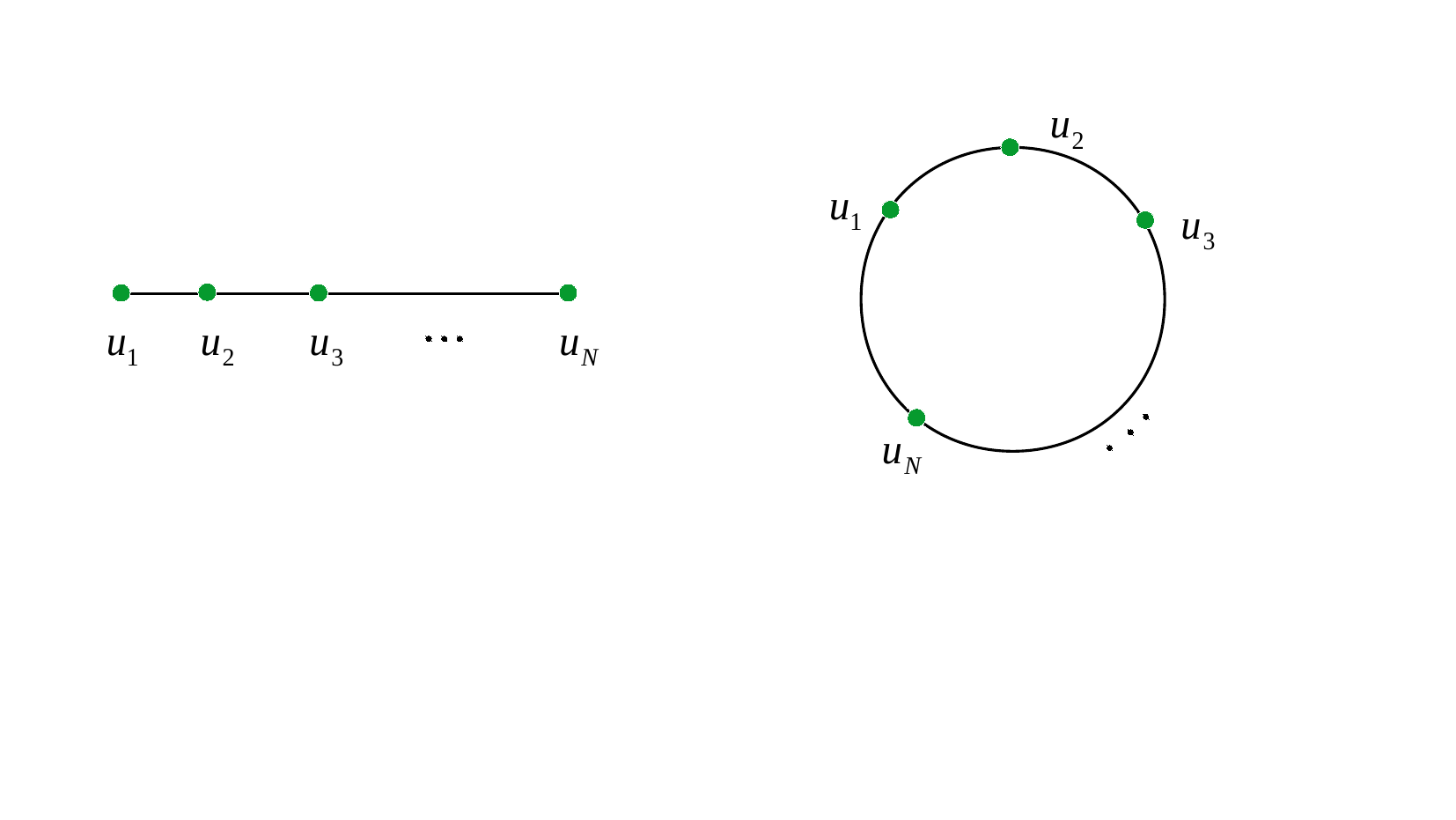}}
\vspace{-0.5em}
\caption{Reconstruction of the locations of $N$ points from their ``unassigned'' pairwise distances in the 1D case where the points could lie on a line or a loop. The correspondence between the distance $d_k$ and the pair of points $(u_m,u_n)$ is unknown.}
\vspace{-1.5em}
\label{fig:1d_uDGP}
\end{figure}

The uDGP is harder to solve than the usual assigned DGP \cite{EDG14} where the assignments are already known. If we want to apply the existing strategies developed for assigned DGP, we need to first find the correct assignments of the distances. The combinatorial nature of this task and the noise in the distance measurements make it challenging. Beyond general theoretical interest in solving the uDGP, its relevance in the 1D case stems from applications. We mention the following three:

\emph{a) Partial digestion.} One of the early methods for genome reconstruction uses partial digestion of DNA \cite{PD:1976}, though nowadays it has been replaced by the commercially available high-throughput sequencing platforms such as Illumina \cite{MOREY20133,REUTER2015586,NGS_review:2016}. In the experiment, an enzyme digests a DNA fragment at the so-called restriction sites $\{u_1<\cdots<u_N\}$. Since the digestion is random and partial, one is left with a collection of fragments whose lengths correspond to the distances between all pairs of restriction sites. The task is then to recover the $N$ site locations from the \emph{unassigned} fragment lengths, which is modeled as the turnpike problem. Hence sometimes the turnpike problem is also referred to as the partial digest problem \cite{Skiena:1990,Waterman:1995,Skiena:1994}.

\emph{b) De novo peptide sequencing.} In tandem mass spectrometry \cite{TMS:1986, TMS:2003}, a peptide is bombarded with electrons and broken down into smaller ionized peptide fragments. The mass-to-charge ratios of those fragments can be measured to produce the tandem mass spectrum of the peptide. In the experiment, the peptide backbone could break at any weak peptide bond and the fragment masses can be interpreted as  ``distances'' between pairs of broken peptide bonds. De novo peptide sequencing \cite{DatabasePeptide:1997,DenovoPeptide:1999} aims to reconstruct the amino acid sequence of a peptide from its mass spectrum. For cyclic peptides \cite{CyclicPeptide:2011,Fomin:2015}, the sequencing problem can be formulated as a beltway problem where the points lie on a loop. For non-cyclic peptides, it becomes the turnpike problem. 

We mention that the turnpike problem is also related to the problem of string reconstruction from substring compositions which arises in protein mass spectrometry \cite{Acharya2015StringRF,Bulteau:2014,LEE2013}. The advances presented here for the turnpike problem might inspire similar approaches to solve its string variant.
    
\emph{c) Spectral estimation.} Zintchenko and Wiebe \cite{SpecEst:2016} showed that randomized phase experiments allow one to infer the eigenvalue gaps in low-dimensional quantum systems. Reconstructing the eigenspectrum $\{\chi_1=0<\chi_2<\cdots\chi_N\}$ from pairwise eigenvalue gaps is then an instance of the turnpike problem. When the eigenvalue gaps between consecutive eigenvalues $(\chi_i, \chi_{i+1})$ are unique, this is known as reconstructing the Golomb ruler \cite{Sidon1932,Babcock:1953}. When $N\neq 6$, the recovered eigenspectrum is unique up to congruence \cite{Golomb:2007}.

\subsection{Related Work}

In the noiseless case, Lemke and Werman \cite{Lemke:1988} address the turnpike problem via polynomial factorization. Namely, the polynomial $Q_\mathcal{D}(a) = N+\sum_{k=1}^K(a^{d_k}+a^{-d_k})$ is invariant to permutations of pairwise distances. If one can factorize it as $Q_\mathcal{D}(a)=R(a)R(a^{-1})$ where $R(a)=\sum_{n=1}^N a^{u_n}$, then the point locations can be read off from the exponents. When the distances are all integers, the factorization runs in a time that is polynomial in the degree of $Q(a)$ \cite{Lenstra1982}, which is the largest pairwise distance. However, this approach quickly becomes impractical, and is brittle in the presence of noise.

The more practical backtracking algorithm by Skiena et al. \cite{Skiena:1990} produces a solution for typical noiseless instances in time $\mathcal{O}(N^2\log N)$ \cite{Lemke2003}. It progressively finds the assignment for the remaining largest unassigned distance in $\mathcal{D}$, and adopts the branch-and-bound search strategy to recover the point locations in a depth-first manner. However, there exist examples with exponential runtime \cite{Skiena:1990,ZhangExp:1994}. Abbas and Bahig \cite{Abbas2016} later demonstrated that some of the worst-case scenarios can be avoided by performing a breadth-first search instead. An alternative to clever combinatorial search is to formulate the problem as a binary integer program \cite{Miller:1960,IBARAKI197639,Papadimitriou:1982}, and then relax it to obtain a convex semidefinite program \cite{Dakic:2000}. One drawback of this scheme is that it is computationally infeasible for large-scale problems. In this paper we propose to relax the integer program to a constrained nonconvex optimization problem that can be solved efficiently using projected gradient descent with a spectral initializer.

To address the noisy case where the turnpike problem becomes NP-hard \cite{Cieliebak:2004}, Skiena and Sundaram proposed a modification of the backtracking algorithm where an interval is associated with each recovered point to account for the uncertainty \cite{Skiena:1994}. As a consequence, the number of backtracking paths could grow exponentially large. Pruning can be performed on the paths when the relative errors in the distances are small; however, it requires careful adaptive tuning and could sometimes lead to no solution. Our approach naturally incorporates noise into the problem formulation, thus exhibiting better performance compared to the current state-of-the-art backtracking approach.

The combinatorial turnpike problem can be formulated as an assignment problem \cite{Assignment:1957,Burkard:2009} or a general integer program (when the domain is discrete) \cite{Miller:1960,IBARAKI197639,Papadimitriou:1982}. Most of the prior approaches try to first find the correct assignments of the distances to pairs of points $\mathscr{M}(k)$, and then recover the point locations $u_n$. On the other hand, the approach by Daki\'c \cite{Dakic:2000} adopts the integer programming formulation where the point locations are represented by a binary vector in the noiseless case, and directly recovered via a semidefinte relaxation. However, the resulting problem size becomes formidable for large-scale problems, and there is no guarantee that the semidefinite relaxation would produce a rank-1 solution. Additionally, a quadratic integer programming formulation was also proposed by Fontoura et al. \cite{Fontoura:superset:2018} to solve the ``minimum distance superset problem'' where some distances are missing.

The beltway problem is more difficult than the turnpike problem \cite{Skiena:1990, Lemke2003}. Due to the loop structure, it can no longer be formulated as a polynomial factorization problem. It is also impossible for the backtracking approach to rely on the remaining largest unassigned distance to find the point locations progressively. Lemke et al. \cite{Lemke2003} showed that the computational complexity of the search in the beltway problem is $\mathcal{O}(N^N\log N)$. For small problems, Fomin \cite{Fomin:2016:1,Fomin:2016:2} proposed to avoid an exhaustive search in the noiseless case by further removing the redundant distances from $\mathcal{H}$ sequentially, and later extended it to handle noisy measurements\cite{Fomin:2019:3}. To the best of our knowledge, our work in this paper offers an alternative by providing the first practical approach to solve the large-scale noisy beltway problem.

\subsection{Uniqueness}

One complication with the turnpike problem is that the solution is not necessarily unique (up to a relabeling of the points and up to a congruence). Fortunately, the solution to the uDGP in any dimension is known to be \textit{generically} unique, in the sense made precise in the form of the reconstructability for the point configurations by Boutin and Kemper in \cite[Theorem 2.6 and Proposition 2.11]{BOUTIN2004709}. For example, if the points are sampled i.i.d. from an absolutely continuous probability distribution, then almost surely the distance distribution specifies their geometry uniquely (up to relabeling and congruence). 

Boutin and Kemper worked with complete distance measurements. Gortler et al. \cite{GUGR2018} later relaxed the completeness assumption and only required the underlying graph to be generically globally rigid \cite{CGGR2010}. Under this sufficient condition, they proved that the reconstruction of a generic point configuration is unique.

Importantly, beyond uniqueness, Boutin and Kemper \cite{BOUTIN2004709} showed that when the multiset $\mathcal{D}$ in the turnpike problem contains only distinct distances, there is a suitably defined neighborhood around each uniquely reconstructable point configuration such that all configurations within the neighborhood are also uniquely reconstructable, and the forward and backward mappings between the different distance multisets are continuous. We emphasize that this result does not depend on any particular algorithm, but is rather a fundamental statement about the well-posedness of the inverse problem of recovering the geometry from distance distribution. To the best of our knowledge, there has not been much work on the uniqueness of beltway reconstructions. In the remainder of this paper, we assume that the measured distances correspond to a uniquely reconstructable configuration.

\subsection{Our Approach and Paper Outline}

We proceed along the line of integer programming to solve the turnpike and beltway problems in Section \ref{sec:main_turnpike} and \ref{sec:main_beltway} respectively. Instead of relaxing the integer program to a convex SDP as Daki\'c \cite{Dakic:2000}, we relax it to a constrained minimization of a nonconvex objective, which is computationally much more efficient and suitable for large-scale problems. To this end we also develop an efficient projection onto the relaxed constraint set. It can be initialized with a suitably constructed initializer inspired by the spectral initialization strategy \cite{Netrapalli2015:RPAM,WF:2015} or a random initializer. The measurement noise is naturally incorporated into our formulation by smoothing the target distance distribution. We complement these results with a convergence analysis of the proposed method in the neighbourhood of a global optimum, and an analysis of the difficulty of recovery using the mutual information between the point and distance.

Starting with the easier turnpike problem, we present the proposed distance distribution matching approach in Section \ref{sec:main_turnpike}, and then demonstrate how it can be adapted to solve the beltway problem in Section \ref{sec:main_beltway}. Convergence analysis of the proposed approach and the accompanying analysis of the difficulty of recovery are given in Section \ref{sec:cvg}. Numerical experiments in Section \ref{sec:exp} show that our method achieves state-of-the-art performances for the turnpike recovery, and is the first practical approach to solve the large-scale noisy beltway problem. We conclude this paper with a discussion of our results in Section \ref{sec:con}. The proofs of the formal results can be found in the Appendix. Detailed derivations of the proofs are given in the accompanying Supplementary Material.

\section{The Noisy Turnpike Problem}
\label{sec:main_turnpike}

We begin by addressing the following problem:

\begin{problem}[Noisy Turnpike]
Reconstruct the relative positions of $N$ points on a line $\{u_1,u_2,\cdots,u_N\}$ from a multiset $\mathcal{D}$ of $\binom{N}{2}$ unassigned noisy pairwise distances,
\[
\mathcal{D}=\left\{d_k = b_k + w_k,\ 1\leq k\leq \tbinom{N}{2}\right\}\,,
\]
where $b_k = \norm{u_{i} - u_{j}}$ with $\set{i, j} = \mathscr{M}^{-1}(k)$, $i < j$, is the noiseless distance, $d_k$ is the measured noisy distance, and $w_k$ is the noise.
\end{problem}

For notational convenience, from now on we will augment $\mathcal{D}$ with $N$ zero self-distances, that is, the distances from every point $u_n$ to itself. The total number of distances considered in the turnpike problem is then $K=\tbinom{N}{2}+N$.

\begin{figure}[tbp]
\centering
\includegraphics[width=3in]{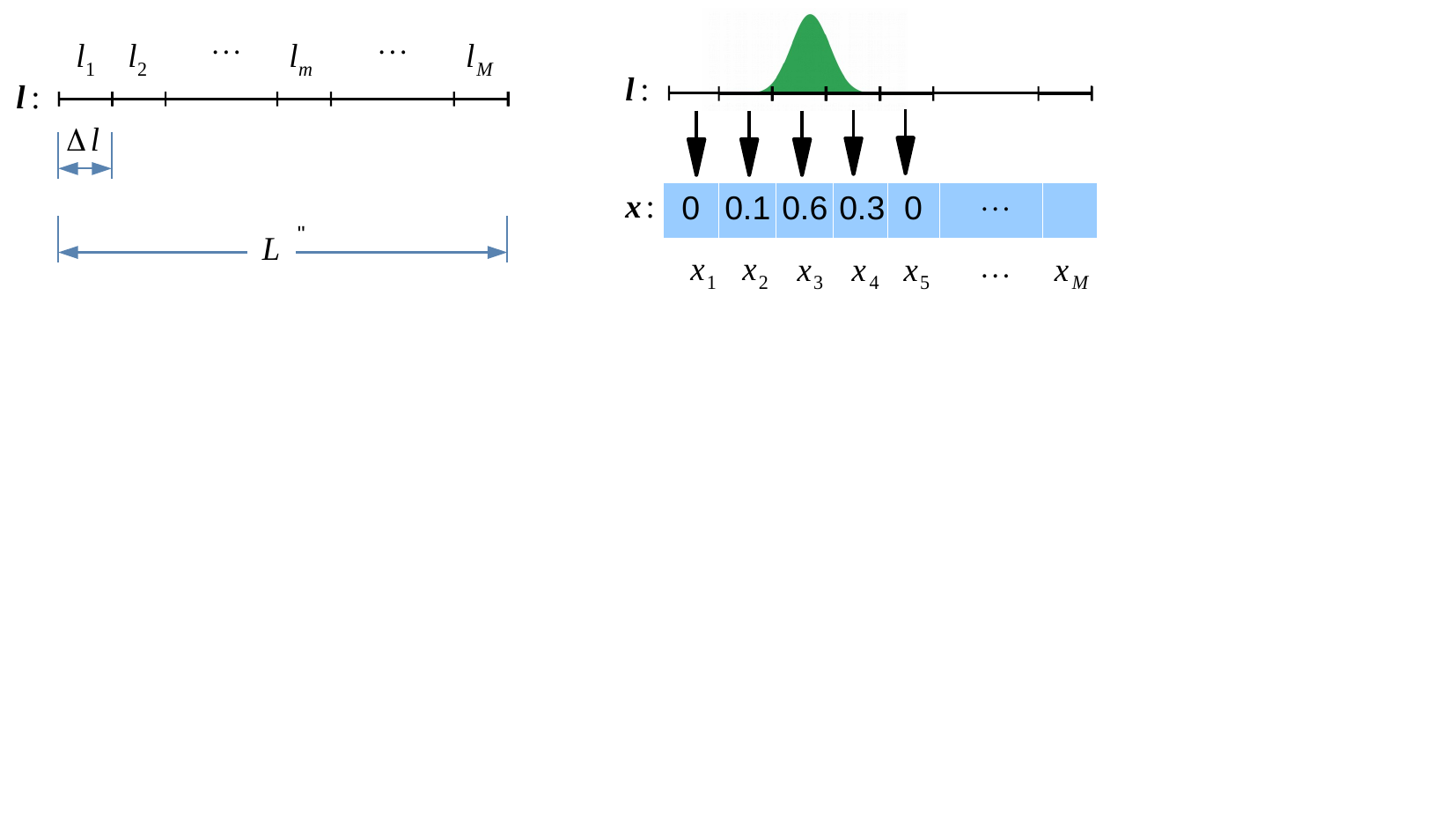}
\caption{In the turnpike problem, the 1D domain $\vl$ is discretized into $M$ segments $\set{l_1,\cdots,l_M}$. The point locations are represented by the vector $\vx$: the $m$-th entry $x_m$ is the probability that a point is located at $l_m$.}
\label{fig:discretization}
\end{figure}

As shown in Fig. \ref{fig:discretization}, suppose the $N$ points lie on a line segment $\boldsymbol l$ of length $L$. We discretize the 1D domain by dividing $\boldsymbol l$ into $M$ segments $\set{l_1,\cdots,l_M}$ of equal length $\lambda$. As a result, the point location $u_n$ and the distance $d_k$ are quantized to $v_n$ and $y_k$ respectively:
\begin{align}
v_n=\left\lfloor\frac{u_n}{\lambda}\right\rceil \quad \textnormal{ and } \quad y_k=\left\lfloor\frac{d_k}{\lambda}\right\rceil\,,
\end{align}
where $\left\lfloor\cdot\right\rceil$ is the nearest integer function. In order to avoid confusion in quantized locations, we need to choose a $\lambda$ at least smaller than the minimum distance between two points. Conversely, this can be interpreted as a minimum separation criterion given a fixed discretization. We will henceforth assume this criterion is satisfied.

We now represent the point set by a vector $\vx = (x_m)_{m=1}^M \in \R^M$, with $x_m = 1$ if the $m$-th segment contains a point and $x_m = 0$ otherwise. However, instead of insisting that each discretization cell contain an integral number of points, we relax the $0$-$1$ integer constraints on $\vx$ as
\begin{align}
    \label{eq:box_0_1}
    &0\leq x_m\leq 1,\ \forall\ m\in\{1,\cdots,M\}\\
    \label{eq:l1_sum}
    &\sum_{m=1}^Mx_m=N\,.
\end{align}
By doing so, we can interpret the rescaled $x_m$ as the probability that a point is located at $l_m$ on the discretized domain. Beyond mathematical convenience of having a convex domain for $\vx$, this is a natural way to handle noise and represent uncertainty in point locations.

The noise in the quantized distance $y_k$ comes from both the measurement noise that is already contained in $d_k$ and the quantization error due to the finite-resolution grid. Letting $y\in\set{0,1,\cdots,M-1}$ denote the quantized distance, we can compute the distance distribution $p(y)$ using $\vx$ as follows:
\begin{align}
    \label{eq:quad_form_p}
    p(y)=\frac{1}{K}\sum_{i=1}^M\sum_{j=i}^Mx_ix_j\cdot\delta\big(y_{ij}-y\big)=\frac{1}{K}\cdot\vx^T\mA_y\vx\,,
\end{align}
where $y_{ij}$ is the quantized distance between the segments $l_i$ and $l_j$, $\delta(\cdot)$ is the Kronecker delta function, and $\mA_y\in\set{0,1}^{M\times M}$ is the measurement matrix whose $(i,j)$-th entry is given by
\begin{align}
\label{eq:measurement_matrix}
A_y(i,j)=\left\{
\begin{array}{l}
1\\
0
\end{array}
\quad
\begin{array}{l}
\textnormal{if } j-i=y,\textnormal{ and } i\leq j\\
\textnormal{otherwise}\,.
\end{array}
\right.
\end{align}
We can see that $\mA_y$ is a Toeplitz matrix. The normalization by $\frac{1}{K}$ in \eqref{eq:quad_form_p} serves to justify interpreting $p(y)$ as a probability mass function, or a distribution. 

Take as an example the case with $N=3$ points $\set{u_1=1,\ u_2=3,\ u_3=5}$ where the distance multiset $\mathcal{D}$ is $\set{0,0,0,2,2,4}$. We have $\vx=[1\ 0\ 1\ 0\ 1]^T$. Apart from counting the frequencies of the distances in $\mathcal{D}$, we can compute $p(y=2)=\vx^T\mA_2\vx$ as follows
\[
p(y=2) =\frac{1}{6}\cdot\vx^T\left[\begin{array}{ccccc} 0 &0 &\colorbox{cyan_a!30}{1} &0 &0\\ 0 &0 &0 &\colorbox{cyan_a!30}{1} &0\\ 0 &0 &0 &0 &\colorbox{cyan_a!30}{1}\\ 0 &0 &0 &0 &0\\ 0 &0 &0 &0 &0 \end{array}\right]\vx = \frac{1}{3}.
\]
Indeed, two pairwise distances, which is a third of $K=N + \binom{N}{2} = 6$ pairwise distances in $\mathcal{D}$, equal $2$.

\subsection{Distance Distribution Matching}
\label{subsec:ddm}

\begin{figure}[tbp]
\centering
\subfigure{
\label{fig:prob_approx_cont}
\includegraphics[width=1.25in]{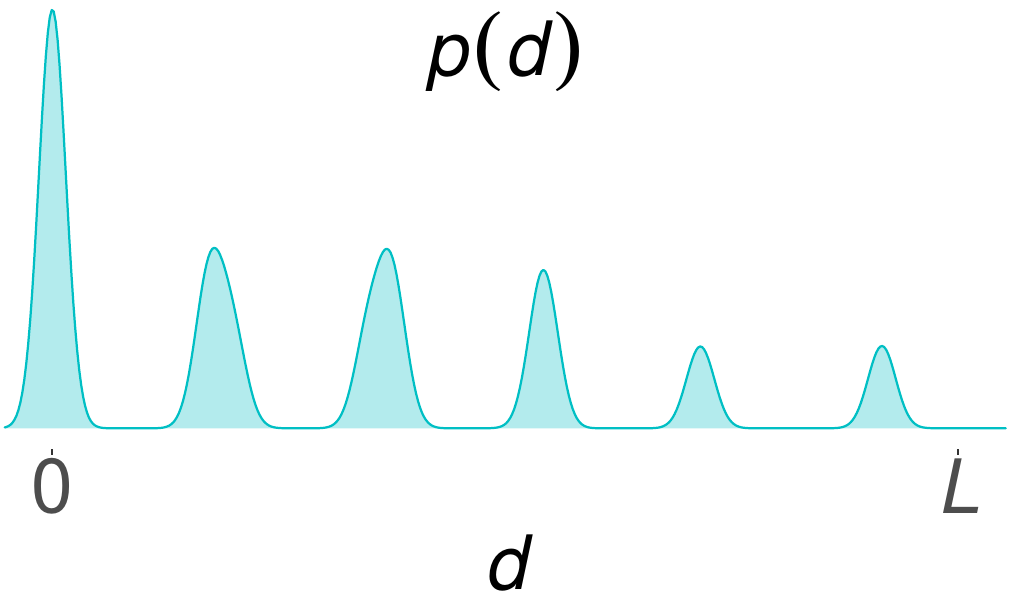}}
\subfigure{
\label{fig:prob_approx_disc}
\includegraphics[width=1.25in]{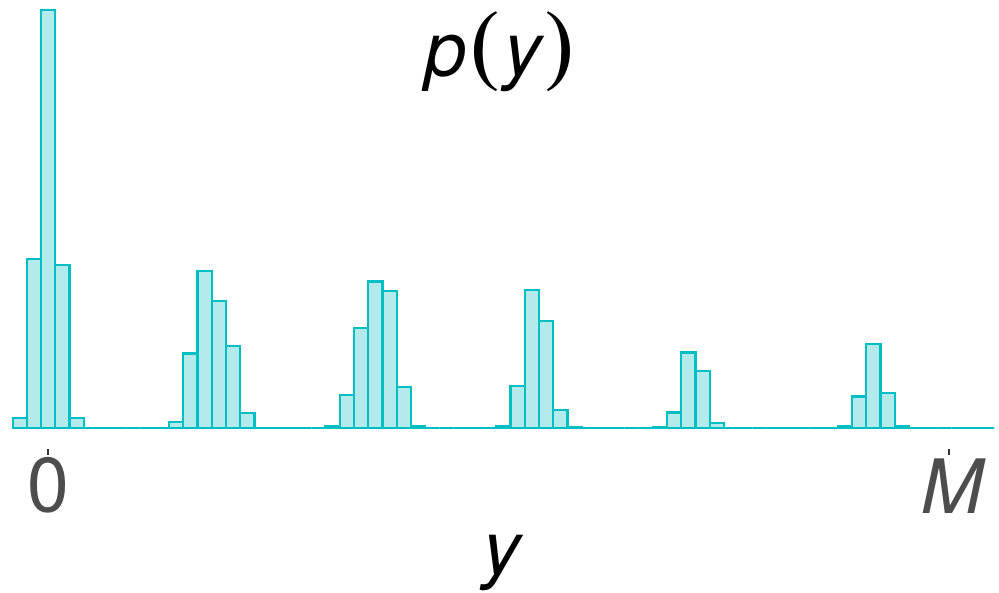}}
\caption{Left: The approximated distribution $p(d)$ based on the distance multiset $\mathcal{D}$; Right: The discretized distance distribution $p(y)$ from $p(d)$.}
\label{fig:prob_dist_discretization}
\end{figure}

Depending on how the distance $d_k$ is measured in various applications, a variety of noise models for the noise $w_k$ may be appropriate \cite{Oppenheim:2009,Vaseghi:2006}. Here we model $w_k$ as i.i.d. zero-mean Gaussian noise with unknown variance $\xi^2$: $w_k\sim\mathcal{N}(0,\xi^2)$. The \emph{oracle} distance distribution $g(d)$ is then
\begin{align}
\label{eq:oracle_dist}
g(d)=\frac{1}{K}\cdot\sum_{k=1}^{K}\mathcal{N}\left(d\ \left|\ b_k,\xi^2\right.\right)\,,
\end{align}
where $b_k$ is the noiseless distance, $\mathcal{N}(d\ |\ b_k,\xi^2)$ is the Gaussian probability density function with $b_k$ as the mean and $\xi^2$ as the variance.

Let $\vz\in[0,1]^M$ denote the solution to the turnpike problem where $z_m$ is the estimated (unnormalized) probability that a point is located at $l_m$. Similar to \eqref{eq:quad_form_p}, the \emph{estimated} distance distribution $q_{\vz}(y)$ can also be expressed in terms of $\vz$:
\begin{align}
\label{eq:quad_form_q}
q_{\vz}(y)=\frac{1}{K}\cdot\vz^T\mA_y\vz\,.
\end{align}
Ideally, we find a solution so that the estimated distribution $q_{\vz}(y)$ matches the oracle distance distribution $g(d)$. However, since $b_k$ and $w_k$ are unknown in practice, we are going to approximate $g(d)$ using the distance measurements in the multiset $\mathcal{D}$. The \emph{approximated} distance distribution $p(d)$ is:
\begin{align}
\label{eq:approx_dist}
p(d)=\frac{1}{K}\cdot\sum_{k=1}^{K}\mathcal{N}\left(d\ \left|\ d_k, \sigma^2\right.\right)\approx g(d)\,,
\end{align}
where $d_k\in\mathcal{D}$, the noise variance $\sigma^2$ should be tuned according to an a priori estimate of the noise level in the data. As shown in Fig. \ref{fig:prob_dist_discretization}, the distribution $p(d)$ is further discretized to the following $p(y)$ in order to perform distribution matching with respect to the quantized distance $y$. 
\begin{align}
\label{eq:obs_dist}
p(y)=\int_{(y-0.5)\lambda}^{(y+0.5)\lambda}p(d)\ \textnormal{d}d\,.
\end{align}
We solve for $\vz$ by minimizing the mean-squared error between the distributions $q_{\vz}(y)$ and $p(y)$ subject to suitable constraints. The resulting optimization problem, which we refer to as ``distance distribution matching'' (DDM), is given as
\begin{equation}
\tag{DDM-T}
\label{eq:constrained_nonconvex} 
\begin{aligned}
\min_{\vz}\quad &f(\vz)=\frac{1}{M}\sum_{y=0}^{M-1}\big(q_{\vz}(y)-p(y)\big)^2
\\
\textnormal{subject to}\quad &0\leq z_m\leq 1,\ \forall\ m\in\set{1,\cdots,M}\\
&\sum_{m=1}^Mz_m = N\,.
\end{aligned}
\end{equation}

The same name ``DDM'' also goes for the beltway optimization problem \eqref{eq:constrained_nonconvex_bw} introduced later in Section \ref{sec:main_beltway} when it is clear from the context.

\subsection{Extracting Point Locations from the Estimated Distribution}

In general, the recovered vector $\vz$ will not be supported on exactly $N$ indices. In the following we discuss how to extract the $N$ point location estimates from $\vz$ when this is the case.

\begin{noiseless_case}
If we assume that there are no measurement noise in $d_k$ and no quantization error in the quantized distance $y_k=\left\lfloor\frac{d_k}{\lambda}\right\rceil$, the vector $\vx$ is then binary: $\vx\in\set{0,1}^M$. Suppose $\vz^\dagger$ is one of the global optimizers of \eqref{eq:constrained_nonconvex} that is different from $\vx$ and $f(\vz^\dagger)=0$. We have from $q(y=0)=p(y=0)$ that 
\[
\|\vz^\dagger\|_2^2={\vz^\dagger}^T\mA_0\vz^\dagger = \vx^T\mA_0\vx=\|\vx\|_2^2=N\,.
\]
From the constraints in \eqref{eq:constrained_nonconvex}, we can get $\|\vz^\dagger\|_2^2=\sum_{m=1}^Mz_m^\dagger=N$. However, if the $m$-th entry $z^\dagger_m\in(0,1)$, then $\|\vz^\dagger\|_2^2<\sum_{m=1}^Mz_m^\dagger$ which contradicts $\|\vz^\dagger\|_2^2=\sum_{m=1}^Mz_m^\dagger$. Hence $z^\dagger_m\notin(0,1)$, and the global optimizer is integer-valued, $\vz^\dagger\in\set{0,1}^M$. The points are at the segments that correspond to the $1$-entries in $\vz^\dagger$.

If the solution $\vz$ is not a global optimizer, then $\vz\in[0,1]^M$. The point locations can be extracted in the same way as in the noisy case which we describe next.
\end{noiseless_case}

\begin{noisy_case}
In the noisy case we have $\vx\in[0,1]^M$. Since the distribution $p(d)$ in \eqref{eq:approx_dist} is an approximation to the oracle distribution $g(d)$, the global optimizer $\vx$ of \eqref{eq:constrained_nonconvex} would be a \emph{perturbed} version of the ground truth noisy signal. For the solution $\vz$, we interpret its $m$-th entry $z_m$ as the \emph{estimated} probability that a point is located at the $m$-th segment $l_m$. 
Extracting $N$ point locations from $\vz$ can be posed as a clustering problem. Each line segment $l_m$ is viewed as a cluster with the weight $z_m$. We cluster the $M$ segments using the agglomerative clustering approach  \cite{Rokach2005}; the pseudocode and illustrations are given in the Supplementary Material. The centroids of the $N$ clusters with the largest weights are taken as the estimated point locations.
\end{noisy_case}

\begin{algorithm}[tbp]
\caption{Projected gradient descent}
\label{alg:pdp_pgd}
\begin{algorithmic}[1]
\REQUIRE adaptive rate $\phi\in(0,1)$, convergence threshold $\epsilon$
\STATE Compute the distribution $q_{\vz}(y)$ and the initializer $\vz_0$
\FOR{$t=\{0,1,\cdots,T\}$}
    \WHILE{true} 
    \STATE Compute the update $\vz_{t+1}=\mathscr{P}_\mathcal{S}\big(\vz_t-\eta\cdot\nabla f(\vz_t)\big)$ 
    \IF {$f(\vz_{t+1})\leq f(\vz_t)$}
        \STATE Increase the step size $\eta=\frac{1}{\phi}\cdot\eta$ and \textbf{break}
    \ELSE
        \STATE Decrease the step size $\eta=\phi\cdot\eta$
    \ENDIF
    \ENDWHILE
    \IF {$\frac{\|\vz_{t+1}-\vz_t\|_2}{\|\vz_t\|_2}<\epsilon$}
        \STATE Convergence is reached, set $\vz=\vz_{t+1}$ and \textbf{break}
    \ENDIF
\ENDFOR

\STATE {\bfseries Return} $\vz$
\end{algorithmic}
\end{algorithm}

\subsection{Projected Gradient Descent}
Let $\mathcal{S}$ denote the convex set defined by the constraints in \eqref{eq:constrained_nonconvex}:
\begin{align}
\label{eq:convex_set}
    \mathcal{S} = \left\{\vz\ |\ 0\leq z_m\leq 1\textnormal{ and }\sum_{m=1}^Mz_m=N\right\}\,.
\end{align}
Given a proper initialization $\vz_0$, we propose to solve \eqref{eq:constrained_nonconvex} via the projected gradient descent method:
\begin{align}
\label{eq:pgd_update}
\vz_{t+1} = \mathscr{P}_\mathcal{S}\big(\vz_t-\eta\cdot\nabla f(\vz_t)\big)\,,
\end{align}
where $\eta>0$ is the step size, $\mathscr{P}_\mathcal{S}(\cdot)$ is the projection of the gradient descent update onto $\mathcal{S}$, and $\nabla f(\vz_t)$ is the gradient
\[
\nabla f(\vz_t) = \frac{2}{MK^2}\sum_{y=0}^{M-1}\left(\vz_t^T\mA_y\vz_t-\vx^T\mA_y\vx\right)\cdot\left(\mA_y+\mA_y^T\right)\vz_t\,,
\]
where both $q_{\vz}(y)$ and $p(y)$ are replaced with their quadratic forms in \eqref{eq:quad_form_p} and \eqref{eq:quad_form_q}. An adaptive strategy can be used to determine some suitable step size $\eta>0$ to minimize the objective function. The distance distribution matching approach is finally summarized by Algorithm \ref{alg:pdp_pgd}.

A suitable initialization is needed to solve the constrained noncovex problem in \eqref{eq:constrained_nonconvex} via projected gradient descent. We explore two initialization strategies: a simple random initialization, and a spectral initialization derived from the quadratic structure of the turnpike problem.

\subsubsection{Spectral Initialization}
\label{subsec:spec_init}
Here we can borrow an idea from another problem with quadratic measurements, the phase retrieval problem \cite{Gerchberg:72,Fienup:82}. In phase retrieval, the task is to compute a complex signal $\vx_{\mathbb{C}}\in\mathbb{C}^M$ from its quadratic measurements of the form $\psi_i = |\langle\vx_{\mathbb{C}}, \va_i\rangle|^2$ for $1 \leq i \leq I$. Since $\psi_i = \vx_{\mathbb{C}}^* \va_i \va_i^* \vx_{\mathbb{C}}$, spectral initialization for phase retrieval is based on a weighted sum of the rank-1 measurement matrices $\va_i \va_i^*$. Namely, using matrix concentration results, Netrapalli et al. \cite{Netrapalli2015:RPAM} showed that the leading eigenvector of $\sum_{i=1}^I \psi_i \va_i \va_i^*$ is close to the true $\vx_{\mathbb{C}}$. Similar arguments can be used for quadratic systems of full-rank random matrices \cite{QuadFR:2019}. 

In our formulation of the turnpike problem \eqref{eq:quad_form_p}, the rank-1 matrices $\va_i \va_i^*$ are replaced by $\mA_y$ which are not necessarily PSD nor rank-1; they are also deterministic. Notwithstanding, we can use the spectral initialization strategy. As we shall see from the numerical experiments in Section \ref{sec:exp}, this strategy works well empirically, although a rigorous proof remains an open question.

One way to interpret the spectral initialization in our case is via a subspace projection. Let $\mH_y=\frac{\mA_y}{\|\mA_y\|_F}$, with $\norm{\,\cdot\,}_F$ being the Frobenius norm. We can rewrite \eqref{eq:quad_form_p} as 
\begin{align}
\psi_y=\frac{p(y)\cdot K}{\|\mA_y\|_F}=\vx^T\mH_y\vx=\langle\mH_y,\ \vx\vx^T\rangle =: \langle \mH_y,\ \mX\rangle\,.
\end{align}

The set $\left\{\mH_y,\ 0\leq y\leq M-1\right\}$ can be viewed as an orthonormal basis for the matrix subspace $\mathrm{span} \, \set{\mH_1, \ldots \mH_{M-1}}$,
\begin{align}
    \langle \mH_i,\ \mH_j \rangle=\left\{\begin{array}{l}
    1\\
    0
    \end{array}
    \quad
    \begin{array}{l}
    \textnormal{if }i=j\\
    \textnormal{if }i\neq j\,
    \end{array}.
    \right.
\end{align}
With this interpretation, $\psi_y$ becomes an expansion coefficient of $\mX$ in the direction of $\mH_y$. Least-squares estimate of $\mX$ is
\begin{align}
\label{eq:spec_init_lsq}
    \widehat{\mX}=\sum_{y=0}^{M-1}\psi_y\cdot\mH_y\,,
\end{align}
which is nothing but the orthogonal projection of $\mX$ on the subspace spanned by the $\mH_y$. Finally, we find the spectral initializer $\vz_0$ so that $\vz_0 \vz_0^T$ is close to $\widehat{\mX}$ in Frobenius norm subject to the constraint that $\|\vz_0\|^2_2=N$. Let the spectral initializer $\vz_0=\sqrt{N}\ve_{\max}$, where $\|\ve_{\max}\|_2=1$. We have
\begin{align}
\begin{split}
\ve_{\max}&=\argmin_{\ve:\|\ve\|_2=1}\ \|\widehat{\mX}-N\ve\ve^T\|_F^2=\argmax_{\ve:\|\ve\|_2=1}\ \ve^T\widehat{\mX}\ve\,.
\end{split}
\end{align}
\begin{itemize}
\item When $\widehat{\mX}$ is symmetric, $\ve_{\max}$ is the leading singular vector of $\widehat{\mX}$ that corresponds to the largest singular value. 
\item When $\widehat{\mX}$ is not symmetric, we use the method of Lagrange multipliers and find the stationary points of the Lagrangian $\mathcal{L}(\ve,\mu)=\ve^T\widehat{\mX}\ve-\mu(\ve^T\ve-1)$. Setting the gradients to $0$, we have
\begin{align}
    (\widehat{\mX}+\widehat{\mX}^T)\ve&=2\mu\ve \quad \textnormal{ and } \quad \ve^T\ve=1\,.
\end{align}
The stationary points are given by the eigenvectors of $\widehat{\mX}+\widehat{\mX}^T$, with $\ve_{\max}$ being the one that corresponds to the largest eigenvalue. We can find it via the power iteration.
\end{itemize}

\begin{figure}[tbp]
\centering
\includegraphics[width=3in]{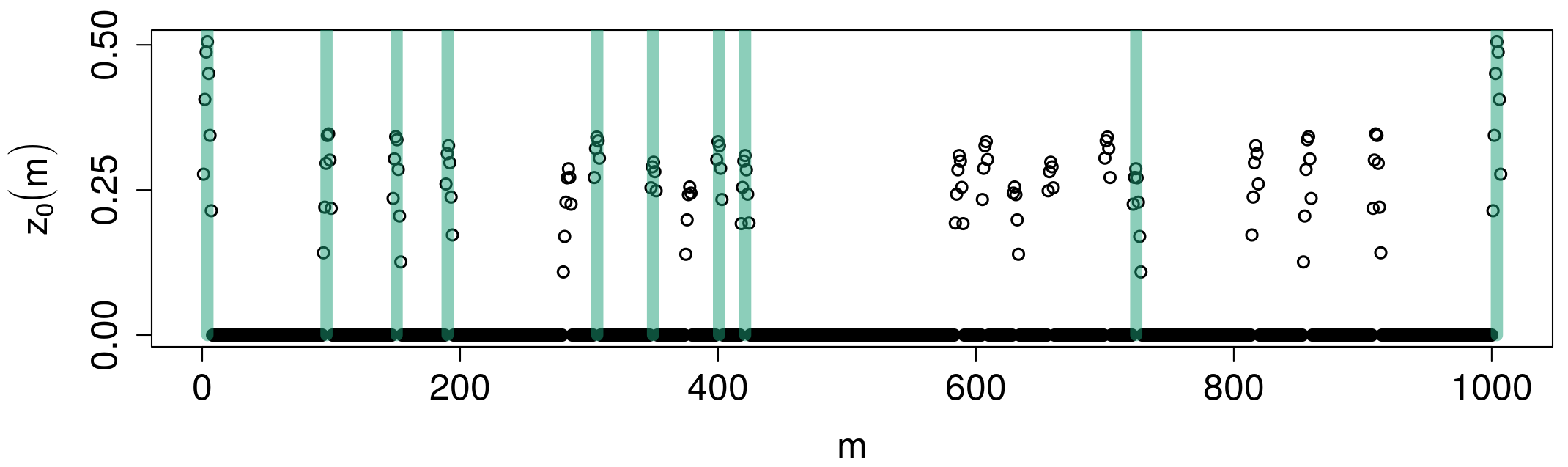}
\caption{The spectral initializer $\vz_0$ for a configuration with $N=10$ points.}
\label{fig:spec_init}
\end{figure}

Fig. \ref{fig:spec_init} shows the spectral initializer $\vz_0$ for a point configuration with $N=10$ points uniformly sampled from $[0,1]$. The 1D domain is discretized with a quantization step $\lambda=1e^{-3}$, producing $M=1e^3$ possible point locations. The true point locations are illustrated by vertical lines. We can see that the entries corresponding to the neighbourhood of the true point locations have larger values, indicating that those locations have higher probabilities in being the actual point locations.

\subsubsection{Efficient Projection onto the Simplex with Box Constraints}
As shown in Fig. \ref{fig:projection}, the gradient descent update $\overline{\vz}=\vz-\eta\nabla f(\vz)$ is projected back onto the convex set $\mathcal{S}$ in \eqref{eq:convex_set}, which is a simplex with box constraints. The projection is the solution to the following convex problem
\begin{align}
\label{eq:projection_box_constraints}
\begin{split}
\min_\vs\quad&\frac{1}{2}\|\vs-\overline{\vz}\|_2^2\\
\textnormal{subject to}\quad&0\leq s_m\leq 1,\ \forall\ m\in\{1,\cdots,M\}\\
&\sum_{m=1}^Ms_m=N\,.
\end{split}
\end{align}
Duchi et al. \cite{projection06,Duchi:2008} proposed an efficient algorithm to compute the projection onto the $l_1$-ball when $s_m$ is only lower-bounded by $0$. Gupta et al. \cite{l1_box10, l1_box12} later extended it to handle projections with box constraints, when $s_m$ is both lower-bounded and upper-bounded. However, their approach is based on a sequential search for an optimal threshold $\kappa$, which is inefficient and cannot be parallelized for large-scale problems. Building on the work of \cite{projection06}, we address these issues by deriving a closed-form expression for the optimal $\kappa$ in \eqref{eq:kappa_compute} in terms of the entry index $r$ of a sorted $\vs$.

\begin{figure}[tbp]
\centering
\includegraphics[width=1.6in]{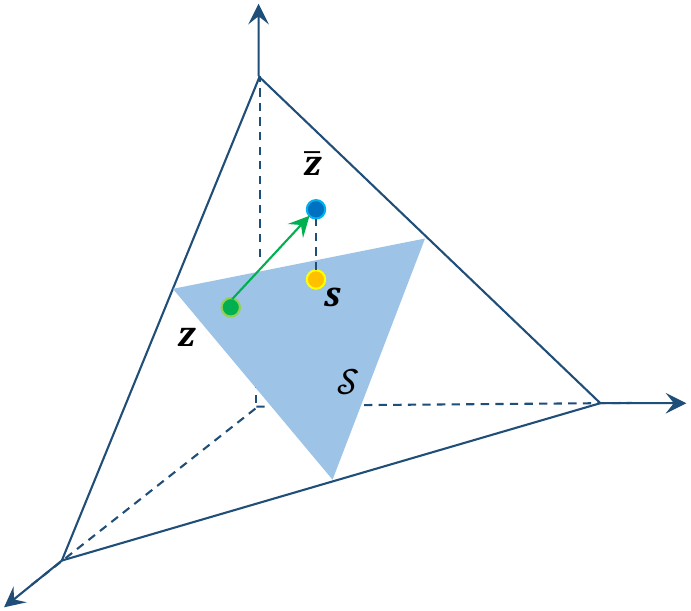}
\caption{ The gradient descent update $\overline{\vz}=\vz-\eta \nabla f(\vz)$ is projected back to the convex set $\mathcal{S}$.}
\label{fig:projection}
\end{figure}

Specifically, the Lagrangian of \eqref{eq:projection_box_constraints} is
\[
\mathcal{L}=\frac{1}{2}\|\vs-\overline{\vz}\|_2^2+\kappa\left(\sum_{m=1}^Ms_m-N\right)-{\boldsymbol\zeta}^T\vs+{\boldsymbol\xi}^T(\vs-\boldsymbol 1)\,,
\]
where $\kappa\in\mathbb{R}$ is a real Lagrange multiplier, ${\boldsymbol\zeta}\in\mathbb{R}_+^M$, ${\boldsymbol\xi}\in\mathbb{R}_+^M$ are the nonnegative Lagrange multipliers. Taking the subgradient of $\mathcal{L}$ w.r.t. $\vs$, and setting it to $0$, we have
\begin{align}
\nabla_{s_m}\mathcal{L}=s_m-\overline{z}_m+\kappa-\zeta_m+\xi_m=0\,.
\end{align}
Since $\mathcal{S}$ is a closed convex set, the projection solution $\vs$ exists and is unique. We need to consider the following two cases.

\emph{a)} If the solution $\vs$ contains only zero or one entries, there are $N$ entries in $\vs$ that equal $1$ and their indices correspond to the top $N$ entries of $\overline{\vz}$. 

\emph{b)} If at least one entry of $\vs$ is in $(0,1)$, the complementary slackness KKT condition indicates that when $0< s_m < 1$, the Lagrange multipliers are $\zeta_m=\xi_m=0$. We then have
\begin{align}
\label{eq:project_one}
s_m=\overline{z}_m-\kappa\quad\quad\textnormal{if } 0< s_m < 1\,.
\end{align}
The above \eqref{eq:project_one} gives us an efficient way to compute $s_m$ if it happens to be between $0$ and $1$: simply subtract the threshold $\kappa$ from $\overline{z}_m$. In order to find the optimal solution $\vs$, we need to compute $\kappa$ and identify the three types of entries of $\vs$: those that equal $0$, those that equal $1$, and those that are between $0$ and $1$. We will make use of the following lemma from \cite{projection06} about the entries of $\vs$ that equal $0$:
\begin{lemma}[Lemma 2, \cite{projection06}]
\label{lemma:lower_bound}
Let $\vs$ be the optimal solution to the minimization problem in (\ref{eq:projection_box_constraints}). Let $i$ and $j$ be two indices such that $\overline{z}_i>\overline{z}_j$. If $s_i=0$ then $s_j$ must be $0$ as well.
\end{lemma}
Similarly, we can prove the following lemma about the entries of $\vs$ that equal $1$ (proved in Appendix \ref{proof:lemma:upper_bound}).
\begin{lemma}
\label{lemma:upper_bound}
Let $\vs$ be the optimal solution to the minimization problem in (\ref{eq:projection_box_constraints}). Let $i$ and $j$ be two indices such that $\overline{z}_i>\overline{z}_j$. If $s_j=1$ then $s_i$ must be $1$ as well.
\end{lemma}

Since reordering of the entries of $\overline{\vz}$ does not change the value of (\ref{eq:projection_box_constraints}), and adding some constant to $\overline{\vz}$ does not change the solution of (\ref{eq:projection_box_constraints}), without loss of generality we can assume that the entries of $\overline{\vz}$ are all positive in a non-increasing order: $\overline{z}_1\geq \overline{z}_2\geq\cdots\geq \overline{z}_M\geq N$. Lemma \ref{lemma:lower_bound} and \ref{lemma:upper_bound} imply that for the optimal solution $\vs$:
\begin{itemize}
\item The entries of $\vs$ are in a non-increasing order.
\item The first $\rho$ entries of $\vs$ satisfy $0<s_m\leq 1$; the rest of the entries are $0$s.
\end{itemize}
Since there exists $s_m\in(0,1)$, we have $\rho>N$ so that at most $N-1$ entries of $\vs$ could equal $1$. Suppose the first $r-1$ entries of $\vs$ are all $1$s. The following must hold for $1\leq r\leq N<\rho$
\begin{align}
\label{eq:cst_r}
&0<z_r-\kappa< 1\\
\label{eq:cst_rm1}
&1\leq z_{r-1}-\kappa,\quad\textnormal{if }2\leq r\leq N<\rho\,.
\end{align}

\begin{algorithm}[tbp]
\caption{Projection onto the simplex with box constraints}
\label{alg:ep_box}
\begin{algorithmic}[1]
\STATE Shift $\overline{\vz}$ s.t. $\overline{z}_m\geq N$, $\forall\ m\in\set{1,\cdots,M}$; and sort $\overline{\vz}$ in a non-increasing order.
\FOR{$r=1:N$}
    \STATE Construct $\vv$ out of $\overline{\vz}$ by removing the first $r-1$ entries
    \STATE Let $N_r=N-r+1$. Compute $\rho_v$ according to \cite{Duchi:2008}: $\rho_v=\max\left\{l\in[N_r]\,:\, v_l-\textstyle\frac{1}{l}\big(\textstyle\sum_{m=1}^lv_m-N_r\big)>0\right\} $

    \STATE Compute $\kappa_v=\frac{1}{\rho_v}\left(\sum_{m=1}^{\rho_v}v_m-(N-r+1)\right)$
    \STATE Check if $(\rho_v, \kappa_v)$ satisfy \eqref{eq:cst_r} via $\widehat{s}_r=z_r-\kappa_v$
    \IF{$0<\widehat{s}_r< 1$}
        \IF{$r=1$}
            \STATE Set $\kappa=\kappa_v$, $\rho=\rho_v+r-1$ and \textbf{break}
        \ELSE
            \STATE Check if $(\rho_v, \kappa_v)$ satisfy \eqref{eq:cst_rm1} via $\widehat{s}_{r-1}=z_{r-1}-\kappa_v$
            \IF{$\widehat{s}_{r-1}\geq 1$}
                \STATE Set $\kappa=\kappa_v$, $\rho=\rho_v+r-1$ and \textbf{break}
            \ENDIF
        \ENDIF
    \ELSE
        \STATE \textbf{continue}
    \ENDIF  
\ENDFOR
\IF{$(r,\rho,\kappa)$ can be found}
\STATE Compute $\vs=\max\{\overline{\vz}-\kappa, 0\}$ and $\vs=\min\{\vs,1\}$
\ELSE
\STATE Compute $\vs$ by setting the top $N$ entries of $\overline{\vz}$ to $1$ and the rest entries to $0$
\ENDIF
\STATE {\bfseries Return} $\vs$
\end{algorithmic}
\end{algorithm}

We can write the sum of $\vs$ as 
$
\sum_{m=1}^Ms_m=\sum_{m=1}^\rho s_m=(r-1)+\sum_{m=r}^\rho (\overline{z}_m-\kappa)=N
$,
which gives
\begin{align}
\label{eq:kappa_compute}
\kappa=\frac{1}{\rho-r+1}\left(\sum_{m=r}^\rho \overline{z}_m-(N-r+1)\right)\,.
\end{align}
Finally, we can write the minimizing $\vs$ as
\begin{align}
\label{eq:min_s}
\vs=\left\{
\begin{array}{l}
1, \\
\overline{z}_m-\kappa, \\
0,
\end{array} \quad 
\begin{array}{l}
\textnormal{if }m\leq r-1\\
\textnormal{if }r\leq m\leq \rho\\
\textnormal{if }\rho+1\leq m\leq M\,.
\end{array}
\right.
\end{align}

If $r$ is known, we can find the value of $\rho$ efficiently using the approach in \cite{projection06,Duchi:2008}, and thus identify the three types of entries in $\vs$. The threshold $\kappa$ and the solution $\vs$ can be computed using \eqref{eq:kappa_compute} and \eqref{eq:min_s}. According to the following lemma (proved in Appendix \ref{proof:lemma:unique_r}), we can find $r$ by checking the integers in $\set{1,\ldots,N}$ one by one or in parallel until the computed $(\rho, \kappa)$ satisfy the two constraints \eqref{eq:cst_r} and \eqref{eq:cst_rm1}.
\begin{lemma}
\label{LEMMA:UNIQUE_R}
If the solution $\vs$ has at least one entry $s_m\in(0,1)$, there is one and only one $r\in\set{1,\ldots,N}$ that produces the $(\rho,\kappa)$ satisfying \eqref{eq:cst_r} and \eqref{eq:cst_rm1}.
\end{lemma}

In practice we do not know beforehand what the solution $\vs$ is like. Given the uniqueness of the solution, we could look for the right $(r,\rho,\kappa)$-values. If they can be found, $\vs$ can then be computed using \eqref{eq:min_s}. Otherwise, it must be that $\vs$ contains only $0-1$ entries and can be obtained straightforwardly. The proposed projection method is summarized in Algorithm \ref{alg:ep_box}, and it requires $\mathcal{O}(N^2)$ operations.
\vspace{0.5em}

After the two anchor points $u_1,u_N$ corresponding to the largest pairwise distance are located at the two ends of the 1D discretized domain, not all the $M$ possible locations are consistent with the distance multiset $\mathcal{D}$. For example, if $l_i$ is a candidate location in the noiseless case, the two distances from $l_i$ to $u_1$ and $u_N$ must belong to $\mathcal{D}$, i.e., $\{d_{i1},d_{iN}\}\in\mathcal{D}$. We can prune the 1D domain by removing the locations that do not satisfy this condition. The pruning process in the noisy case can be performed in a similar fashion. Let $\overline{M}$ denote the number of candidate locations after pruning. The total complexity of the distance distribution matching is at most $\mathcal{O}(M^2)$, and it can be further reduced to $\mathcal{O}(\overline{M}^3)$ when $\tbinom{\overline{M}}{2}<M$.

\section{The Noisy Beltway Problem}
\label{sec:main_beltway}

We now demonstrate how the approach introduced in Section \ref{sec:main_turnpike} can be adapted to solve the noisy beltway problem:
\begin{problem}[Noisy Beltway]
Reconstruct the relative positions of $N$ points on a loop $\{u_1,u_2,\cdots,u_N\}$ from a multiset $\mathcal{B}$ of $N(N-1)$ unassigned noisy pairwise distances,
\[
\mathcal{B}=\left\{d_k=b_k+w_k,\ 1\leq k\leq N(N-1)\right\}\,,
\]
where $b_k = \norm{u_{i} - u_{j}}$ with $\set{i, j} = \mathscr{M}^{-1}(k)$, $i < j$, is the noiseless distance, $d_k$ is the measured noisy distance, and $w_k$ is the noise.
\end{problem}

\begin{figure}[tbp]
\centering
\includegraphics[height=1.2in]{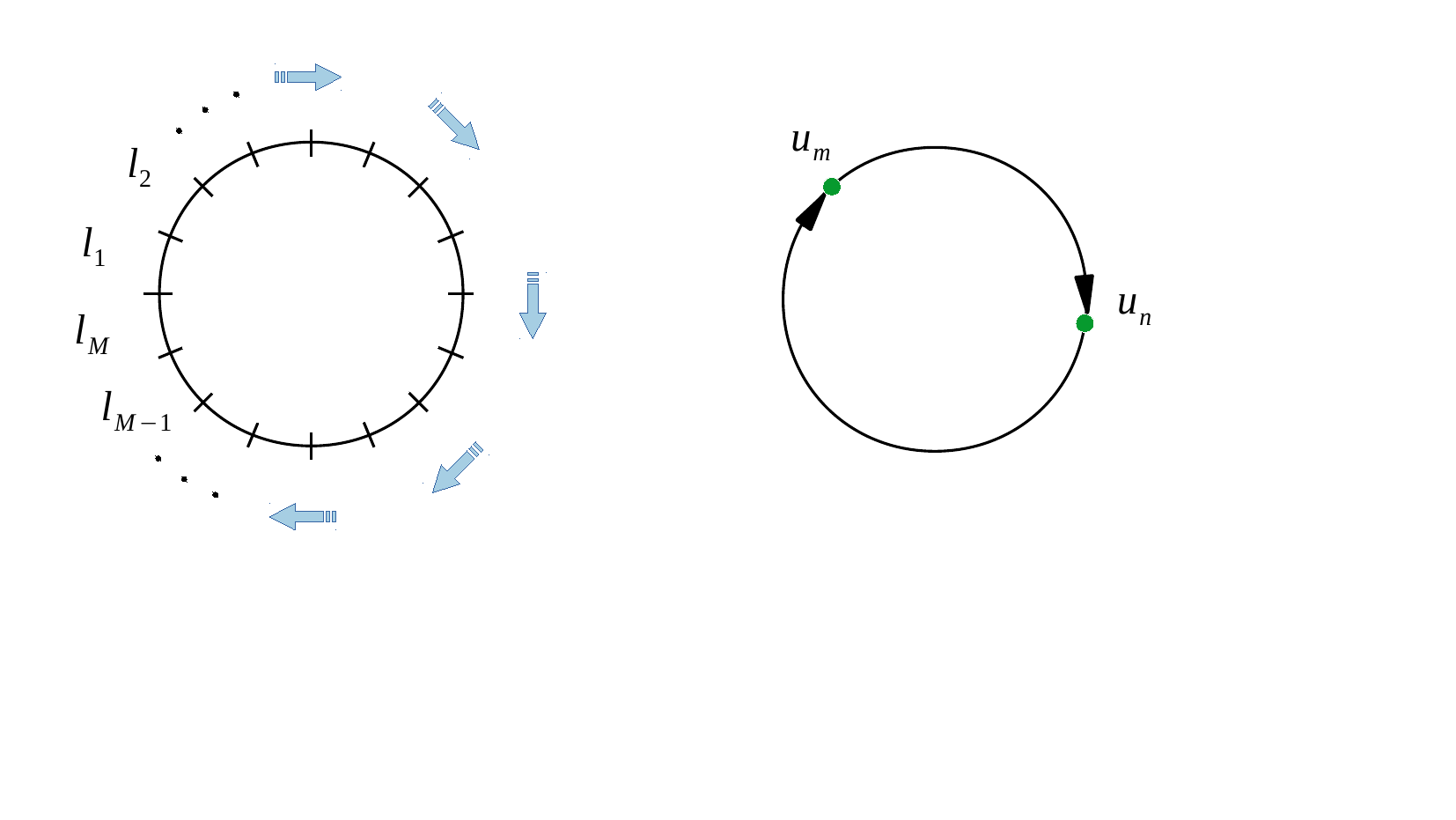}
\caption{In the beltway problem, the 1D domain is also discretized into $M$ segments $\set{l_1,\cdots,l_M}$. The distances are measured in the clockwise direction, and there are two distances associated with a pair of points $(u_m\neq u_n)$: $d(u_n\rightarrow u_m)$ and $d(u_m\rightarrow u_n)$.}
\label{fig:bw_discretization}
\end{figure}

Similarly as for the turnpike problem in Section \ref{sec:main_turnpike}, we augment $\mathcal{B}$ with $N$ zero self-distances. The total number of distances considered in the beltway problem is then $Z=N^2$. As shown in Fig. \ref{fig:bw_discretization}, the loop of length $L$ is also discretized into $M$ line segments $\set{l_1,\ldots,l_M}$. The point locations can be represented by a vector $\vx\in[0,1]^M$ where the $m$-th entry $x_m$ is the probability that a point is located at $l_m$. Compared to the turnpike problem, there are two distances measured in the clockwise direction associated with every pair of points $(u_m\neq u_n)$: the distance $d(u_m\rightarrow u_n)$ from $u_m$ to $u_n$ and the distance $d(u_n\rightarrow u_m)$ from $u_n$ to $u_m$ satisfy: $d(u_m\rightarrow u_n)+d(u_n\rightarrow u_m)=L$.

The quantized distance distribution $r(y)$ can again be written as a quadratic form in terms of $\vx$,
\begin{align}
\label{eq:quad_form_p_bw}
    r(y)=\frac{1}{Z}\sum_{i=1}^M\sum_{j=1}^M x_ix_j \delta\big(y_{i\rightarrow j}-y\big)=\frac{1}{Z}\vx^T\mR_y\vx\,,
\end{align}
where $y_{i\rightarrow j}$ is the quantized distance from $l_i$ to $l_j$, $\delta(\cdot)$ is the delta function, and $\mR_y\in\set{0,1}^{M\times M}$ is the measurement matrix whose $(i,j)$-th entry is given by
\begin{align}
\label{eq:measurement_matrix_bw}
    R_y(i,j)=\left\{
\begin{array}{l}
1\\
1\\
0
\end{array}
\quad
\begin{array}{l}
\textnormal{if } j-i=y,\textnormal{ and } i\leq j\\
\textnormal{if } M-(i-j)=y,\textnormal{ and } i>j\\
\textnormal{otherwise}\,.
\end{array}
\right.
\end{align}
We can see that $\mR_2$ is a circulant matrix. Note the differences between the formulations \eqref{eq:quad_form_p}-\eqref{eq:measurement_matrix} in the turnpike problem and the above \eqref{eq:quad_form_p_bw}-\eqref{eq:measurement_matrix_bw}. In the turnpike problem there is only one distance associated with a pair of segments $(l_i\neq l_j)$. The summation with respect to $j$ goes from $i$ to $M$ in \eqref{eq:quad_form_p}, producing the ``Toeplitz'' matrix $\mA_y$ defined by \eqref{eq:measurement_matrix}. On the other hand, in the beltway problem there are two distances associated with a pair of segments $(\l_i\neq l_j)$. The summation with respect to $j$ goes from $1$ to $M$ in \eqref{eq:quad_form_p_bw}, producing a different ``circulant'' matrix $\mR_y$ defined by \eqref{eq:measurement_matrix_bw}. The distance distributions $p(y)$, $r(y)$ are thus different in the two problems. 

Take as an example the case from Section \ref{sec:main_turnpike} with $N=3$ points $\set{u_1=1,\ u_2=3,\ u_3=5}$ and $\vx=[1\ 0\ 1\ 0\ 1]^T$. Suppose that the 3 points now lie on a loop. We can compute $g(y=2)=\vx^T\mR_2\vx$ as follows:
\[
g(y=2)=\frac{1}{9}\cdot\vx^T\left[\begin{array}{ccccc} 0 &0 &\colorbox{cyan_a!30}{1} &0 &0\\ 0 &0 &0 &\colorbox{cyan_a!30}{1} &0\\ 0 &0 &0 &0 &\colorbox{cyan_a!30}{1}\\ \colorbox{cyan_a!30}{1} &0 &0 &0 &0\\ 0 &\colorbox{cyan_a!30}{1} &0 &0 &0 \end{array}\right]\vx = \frac{2}{9}.
\]

For the noisy beltway problem, we propose to compute an estimate $\vz$ of the true $\vx$ by solving the following optimization problem analogous to the previous \eqref{eq:constrained_nonconvex}:
\begin{equation}
\tag{DDM-B}
\label{eq:constrained_nonconvex_bw}
\begin{aligned}
\min_{\vz}\quad &f(\vz)=\frac{1}{M}\sum_{y=0}^{M-1}\big(h_{\vz}(y)-r(y)\big)^2\\
\textnormal{subject to}\quad &0\leq z_m\leq 1,\ \forall\ m\in\set{1,\cdots,M}\\
&\sum_{m=1}^Mz_m = N\,.
\end{aligned}
\end{equation}

The spectral initialization can be adapted by computing the $\psi_y,\mH_y$ in \eqref{eq:spec_init_lsq} as $\psi_y=\frac{r(y)\cdot Z}{\|\mR_y\|_F}$ and $\mH_y=\frac{\mR_y}{\|\mR_y\|_F}$. We can also prune the locations that are not consistent with the distance multiset $\mathcal{B}$ from the 1D domain and get $\overline{M}$ candidate locations. The computational complexity to solve the above \eqref{eq:constrained_nonconvex_bw} is at most $\mathcal{O}(M^2)$, and it can be further reduced to $\mathcal{O}(\overline{M}^3)$ when $\overline{M}^2<M$.

\section{Analysis on Convergence and Difficulty of Recovery}
\label{sec:cvg}

In this section we first study convergence of our distance distribution matching approach in the neighbourhood $\mathcal{E}(\tau)$ of a global optimizer $\vx$ (see Fig. \ref{fig:cvg_nbhd}),
\begin{align}
\label{eq:cvg_neighbourhood}
    \mathcal{E}(\tau)=\set{\vz\ |\ \|\vz-\vx\|_2<\tau\,,\ \vz\in\mathcal{S}}\,.
\end{align}
We then evaluate the difficulty of recovery using the mutual information between the point and distance.

Unlike the phase retrieval problem whose optimization landscape could be made benign through the design of suitable measurement matrices or increasing the number of measurements, the turnpike and beltway problems are in a disadvantageous situation due to their deterministic measurement models: both the measurement matrices and the number of measurements are fixed. On the other hand, the deterministic nature of the turnpike and beltway problems allows us to study them using the Monte Carlo method.

\subsection{Analysis of the Turnpike Problem}
\label{subsec:turnpike_cvg}

\subsubsection{Convergence Analysis}

\begin{figure}[tbp]
\centering
\includegraphics[width=2in]{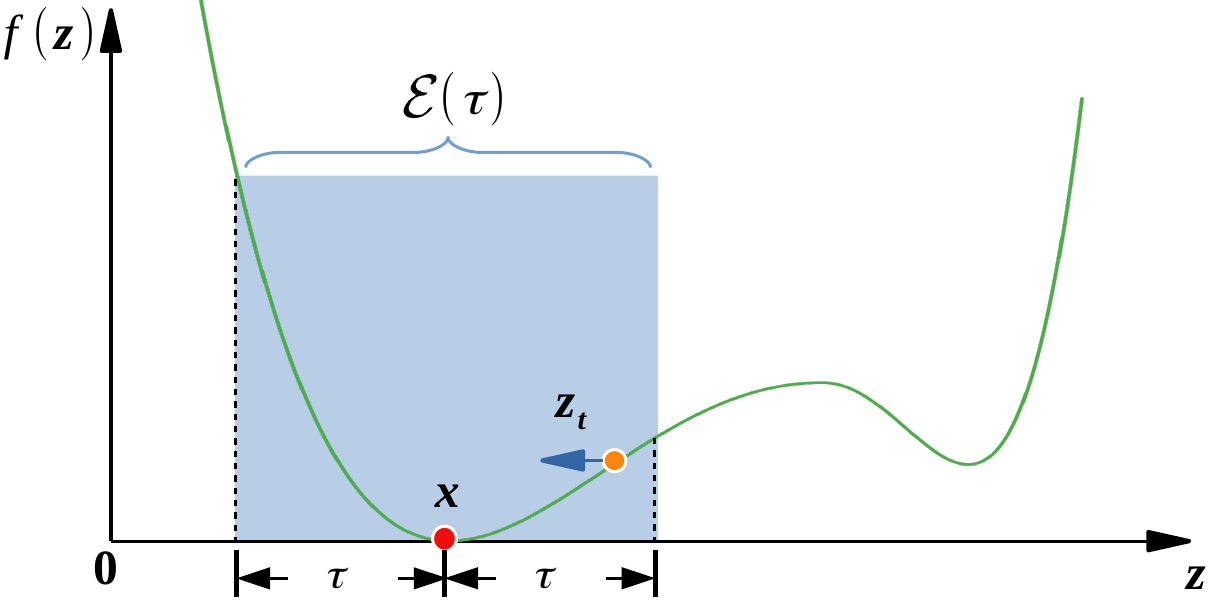}
\caption{When the solution $\vz_t$ reaches the convergence neighbourhood $\mathcal{E}(\tau)$ around a global optimum $\vx$, the projected gradient descent update \eqref{eq:pgd_update} converges linearly to $\vx$.}
\label{fig:cvg_nbhd}
\end{figure}

\begin{figure*}[tbp]
\centering
\subfigure{
\label{fig:dist_rep_1}
\includegraphics[height=0.6in]{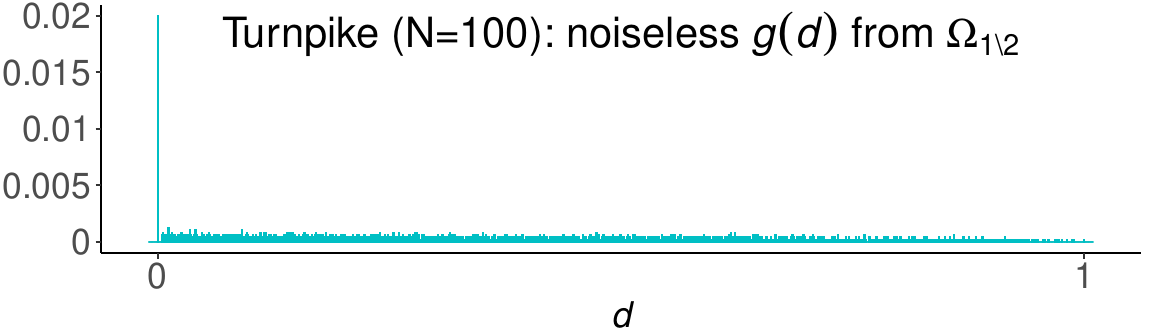}}
\subfigure{
\label{fig:dist_rep_10}
\includegraphics[height=0.6in]{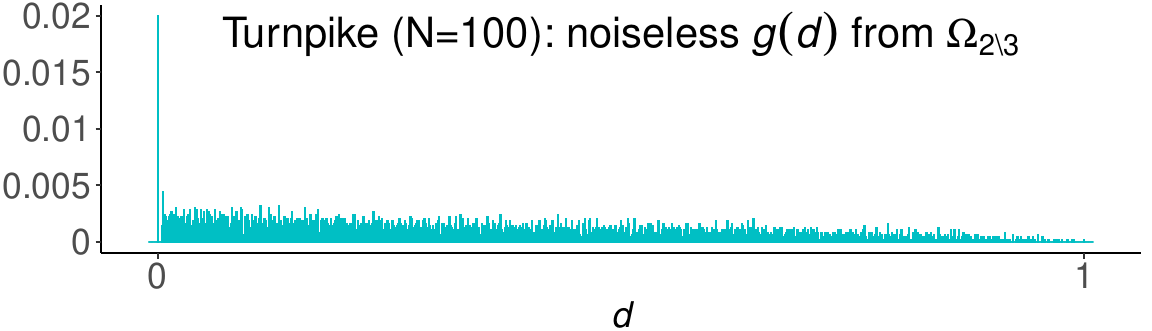}}
\subfigure{
\label{fig:dist_rep_50}
\includegraphics[height=0.6in]{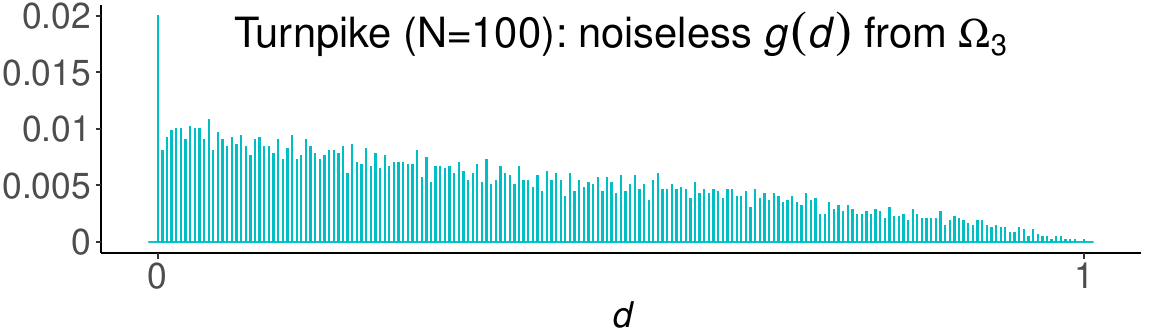}}
\caption{Given the quantization step $\lambda\in\Lambda=\{1e^{-4},1e^{-3},5e^{-3}\}$, the turnpike point configurations drawn from different sets $\Omega^{(N)}_{1\setminus 2},\Omega^{(N)}_{2\setminus 3},\Omega^{(N)}_3$ defined in \eqref{eq:omega_setminus} have different number of repeated distances in the noiseless measurements.}
\label{fig:dist_rep_noiseless}
\end{figure*}

\begin{figure*}[tbp]
\centering
\subfigure{
\label{fig:dist_1_0}
\includegraphics[height=1.7in]{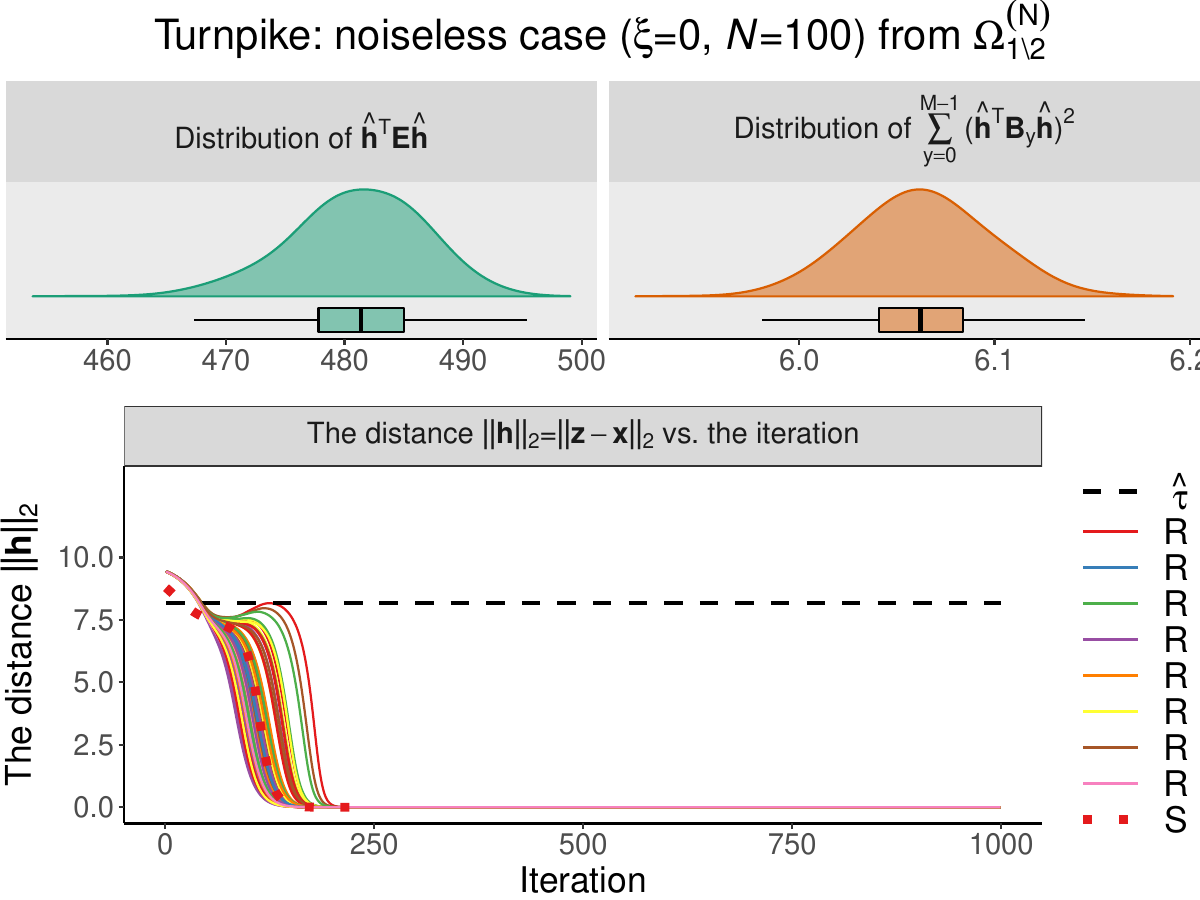}}
\subfigure{
\label{fig:dist_10_0}
\includegraphics[height=1.7in]{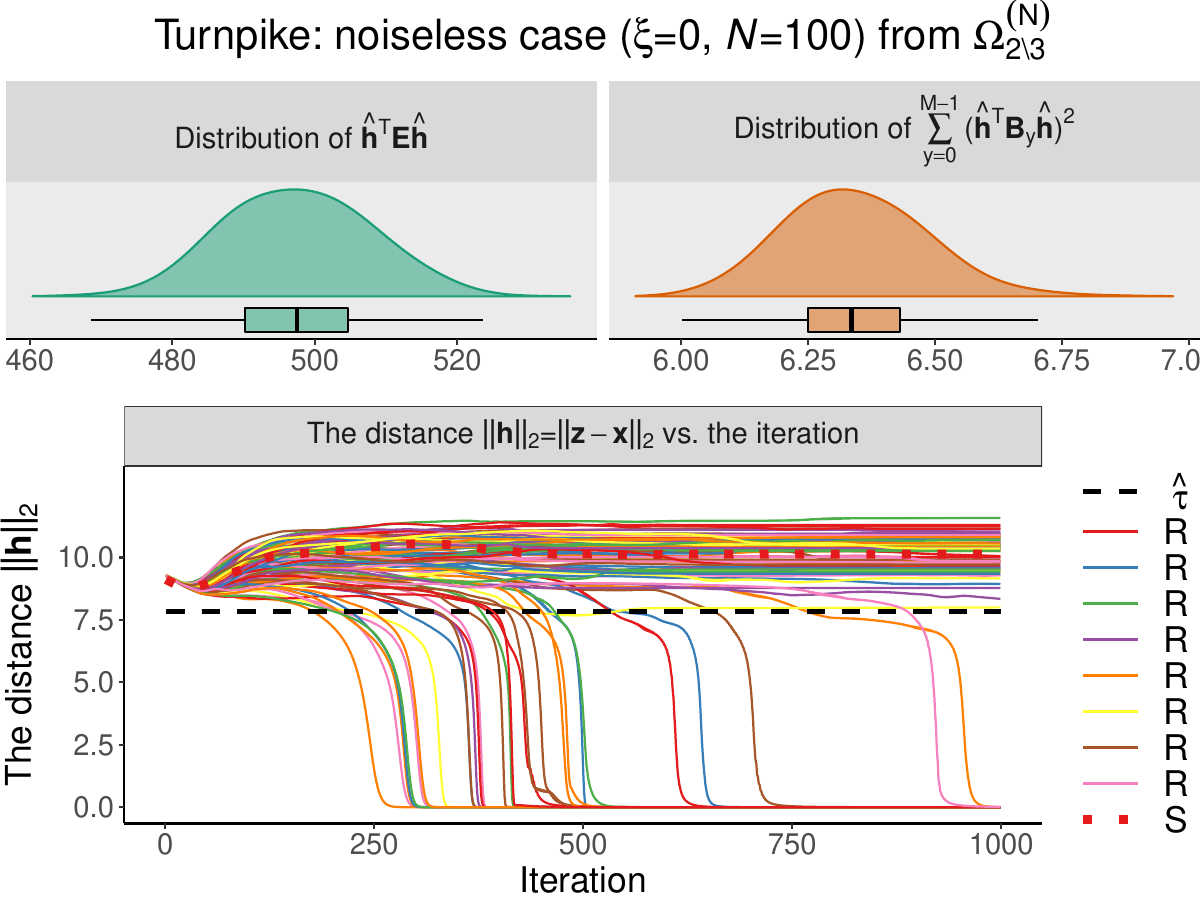}}
\subfigure{
\label{fig:dist_50_0}
\includegraphics[height=1.7in]{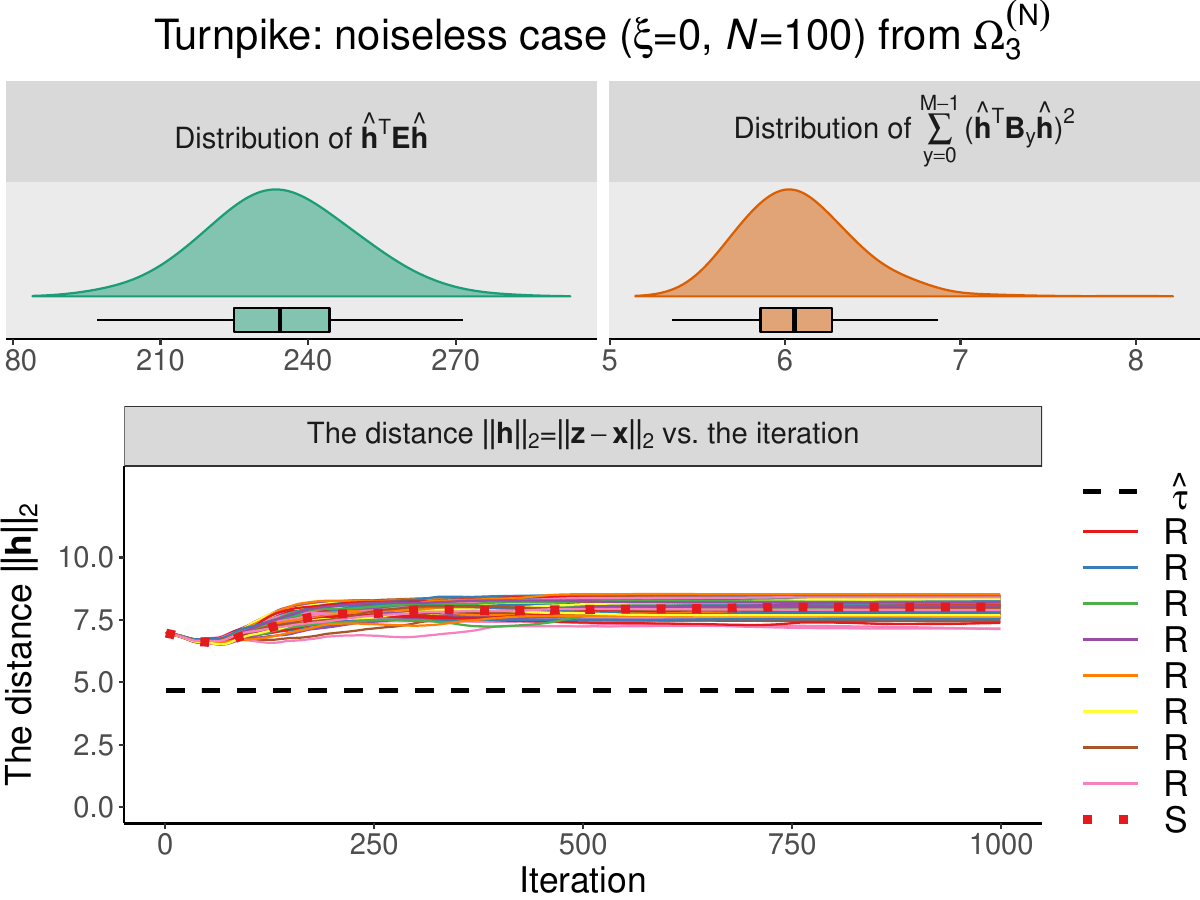}}
\caption{In the noiseless turnpike problem, the number of repeated distances in the measurements affects the performance of projected gradient descent from two aspects: 1) the radius of the empirical convergence neighbourhood $\mathcal{E}(\widehat{\tau})$; 2) the difficulty to reach $\mathcal{E}(\widehat{\tau})$ from outside.}
\label{fig:dist_monte_carlo_noiseless}
\end{figure*}

As illustrated in Fig. \ref{fig:cvg_nbhd}, when the distance between the solution $\vz_t$ and a global optimum $\vx$ is less than some $\tau>0$, i.e. $\|\vz_t-\vx\|_2<\tau$, we would like to show that the projected gradient descent update in \eqref{eq:pgd_update} converges linearly to a global optimizer $\vx$. The convergence neighbourhood $\mathcal{E}(\tau)$ in \eqref{eq:cvg_neighbourhood} is characterized by the regularity condition $RC(\alpha,\beta,\tau)$ of the objective function $f(\vz)$ \cite{WF:2015}: For all $\vz\in\mathcal{E}(\tau)$,
\begin{align}
\label{eq:rc_condition}
    \langle\nabla f(\vz), \vz-\vx\rangle\geq \frac{1}{\alpha}\|\vz-\vx\|_2^2+\frac{1}{\beta}\|\nabla f(\vz)\|_2^2\,,
\end{align}
where $\alpha>0,\beta>0$ are some chosen constants. Let $\overline{\vz}_{t+k}$ denote the gradient descent update. The $RC(\alpha,\beta,\tau)$ in \eqref{eq:rc_condition} ensures $\overline{\vz}_{t+k}$ with a step size $\eta\in(0,\frac{2}{\beta}]$ converges linearly to $\vx$ once $\vz_t$ reaches $\mathcal{E}(\tau)$ \cite[Lemma 7.10]{WF:2015}:
\begin{align}
\label{eq:gd_closer}
    \left\|\overline{\vz}_{t+k}-\vx \right\|^2_2\leq\left(1-\frac{2\eta}{\alpha}\right)^k\cdot\|\vz_t-\vx\|^2_2.
\end{align}
We shall further extend the above \eqref{eq:gd_closer} to the projected gradient descent update $\vz_{t+k}=\mathscr{P}_\mathcal{S}(\overline{\vz}_{t+k})$.

\begin{noiseless_case}
The global optimizer $\vx\in\{0,1\}^M$. Here we make use of the following theorem:
\begin{theorem}
\label{THM:CVG_SED}
In the noiseless case, let $\vh=\vz-\vx$ and $\mB_y=\mA_y+\mA_y^T$. If $\vz$ satisfies
\begin{align}
\label{eq:absolute_radius}
\|\vh\|_2=\|\vz-\vx\|_2<\tau=\left(2-\frac{1}{\theta}\right)\cdot\sqrt{\frac{\mu_\mE}{4}}\,,
\end{align}
where $\theta\in\big(\frac{1}{2},1\big)$ is some fixed constant and $\mu_\mE>0$ depends on the matrix $\mE=\sum_{y=0}^{M-1}\mB_y\vx\vx^T\mB_y^T$, 
\begin{enumerate}
\item There exists a choice of $\{\alpha>0,\beta>0\}$ such that the regularity condition $RC(\alpha,\beta,\tau)$ holds.
\item Under this choice of parameters $\{\alpha,\beta,\tau\}$, if $\|\vh_t\|_2=\|\vz_t-\vx\|_2<\tau$ and the step size $\eta\in(0,\frac{2}{\beta}]$, the projected gradient descent update in \eqref{eq:pgd_update} converges linearly to $\vx$:
\begin{equation}
\label{eq:converge}
\|\vz_{t+k}-\vx\|_2^2<\left(1-\frac{2\eta}{\alpha}\right)^k\cdot\|\vz_t-\vx\|_2^2\,.
\end{equation}
\end{enumerate}
\end{theorem}

The proof is given in Appendix \ref{proof:thm:cvg_sed}. The radius $\tau$ of the convergence neighbourhood varies for different signals $\vx$. 

Let $\vh=\vz-\vx$ and $\overline{\vh}=\vh/\|\vh\|_1$. According to Lemma \ref{LEMMA:E_MIN} in Appendix \ref{proof:lemma:e_min}, $\mu_\mE$ can be computed via the convex program:
\begin{align}
\label{eq:convex-mu}
\mu_\mE = \min_{\vz\in\mathcal{S}, \vz\neq\vx}\,\frac{(\vz-\vx)^T\mE(\vz-\vx)}{\|\vz-\vx\|_1^2} = \min\limits_{\overline{\vh}\in\mathcal{G}}\,\sum_{y=0}^{M-1}\left(\overline{\vh}^T\mB_y\vx\right)^2\,,
\end{align}
where $\vz\in\mathcal{S}$, $\vz\neq\vx$, and $\mathcal{G}$ is the convex set defined in Lemma \ref{LEMMA:E_MIN}. Note that $\mu_\mE>0$ in the noiseless case. To see this, let us assume that $\mu_\mE=0$. We then have
\begin{align}
\label{eq:noiseless_lambda_E_zero}
(\vz-\vx)^T\mB_y\vx=0,\ \forall\ y\in\{0,\cdots,M-1\}.
\end{align}
Using $\mB_0=2\mI$, where $\mI$ is the identity matrix, we can get $\vz^T\vx=\vx^T\vx=N$. Since $\vz\in\mathcal{S}$ and $\vx$ is a binary vector containing exactly $N$ ones, the vector $\vz$ must equal $\vx$ to ensure $\vz^T\vx=N$. This is in contradiction with the assumption that $\vz\neq\vx$.\footnote{If $\vz=\vx$, then we already have a global optimal solution.} Hence $\mu_\mE\neq 0$. Since $\mE$ is a positive semidefinite matrix, we can get that $\mu_\mE>0$ and $\tau>0$ for all $\vh$. 

From \eqref{eq:second_term_lb1_two} in the proof of Theorem \ref{THM:CVG_SED}, we can see that finding an upper bound on $\sum_y(\vh^T\mB_y\vh)^2$ and a lower bound on $\vh^T\mE\vh$ is the key to verify that the regularity condition $RC(\alpha,\beta,\tau)$ holds. The absolute upper and lower bounds are given in \eqref{eq:ub_hDh} and \eqref{eq:lb_hdx}. However, together with \eqref{eq:convex-mu} these bounds lead to a pessimistic estimate of the convergence radius $\tau$ in \eqref{eq:absolute_radius}. For example, in the simulation experiments where there are $N=100$ points, we can get $\tau\approx0.5$. An empirical radius $\widehat{\tau}$ that does not rely on absolute bounds would be practically more useful to describe the algorithm's convergence behavior around $\vx$. To this end, we complement Theorem \ref{THM:CVG_SED} by estimates of the convergence radius computed numerically.

Let $\widehat{\vh}=\vh/\|\vh\|_2$. Note that $\widehat{\vh}$ is obtained by normalizing $\vh$ with the $l_2$-norm $\|\vh\|_2$. This is different from the previously defined $\overline{\vh}$ which is computed by normalizing $\vh$ with the $l_1$-norm $\|\vh\|_1$. Using the Monte Carlo simulations detailed in Appendix \ref{app:monte_carlo}, we can estimate the empirical upper and lower bounds $\nu_1>0,\nu_2>0$ so that the following two inequalities
\begin{align}
    \label{eq:upper_bd_hbh}
    \sum_{y=0}^{M-1}\big(\widehat{\vh}^T\mB_y\widehat{\vh}\big)^2&\leq\nu_1\\
    \label{eq:lower_bd_heh}
    \widehat{\vh}^T\mE\widehat{\vh}&\geq\nu_2\,,
\end{align}
hold with probability $P(\nu_1)$ and $P(\nu_2)$ respectively. The left- and right-hand sides of \eqref{eq:second_term_lb1_two} can then be bounded as
\begin{align}
    \sqrt{\sum_y\left(\vh^\textrm{T}\mB_y\vh\right)^2}&\leq\|\vh\|_2^2\cdot \sqrt{\nu_1}\\
    (2-\frac{1}{\theta})\sqrt{\sum_y\left(\vh^T\mB_y\vx\right)^2}&\geq\left(2-\frac{1}{\theta}\right)\|\vh\|_2\cdot \sqrt{\nu_2}\,,
\end{align}
with probability $P(\nu_1)$ and $P(\nu_2)$ respectively.
In order for \eqref{eq:second_term_lb1_two} to hold empirically, $\|\vh\|_2$ should satisfy
\begin{align}
    \|\vh\|_2<\widehat{\tau} = \left(2-\frac{1}{\theta}\right)\cdot\sqrt{\frac{\nu_2}{\nu_1}}\,,
\end{align}
where $\widehat{\tau}$ is the resulting empirical radius. As shown in Fig. \ref{fig:dist_monte_carlo_noiseless}, the distributions of $\widehat{\vh}^T\mE\widehat{\vh}$ and $\sum_{y=0}^{M-1}\big(\widehat{\vh}^T\mB_y\widehat{\vh}\big)^2$ can be used to find $\{\nu_1,\nu_2\}$ for different types of signals. We observe that \emph{the number of repeated distances} in noiseless measurements (see Fig. \ref{fig:dist_rep_noiseless}) affects the recovery performance from two aspects:
\begin{enumerate}
    \item The radius of the empirical convergence neighbourhood $\mathcal{E}(\widehat{\tau})$. The empirical radius $\widehat{\tau}$ decreases when there are more repeated distances in the noiseless measurements.
    \item The difficulty to reach $\mathcal{E}(\widehat{\tau})$ from outside. When the noiseless measurements contain many repeated distances, it is much more difficult to reach $\mathcal{E}(\widehat{\tau})$ compared to the case where there are few repeated distances.
\end{enumerate}

In the Monte Carlo experiments, we generate point configurations with varying numbers of repeated distances as follows:
\begin{itemize}[leftmargin=*]
    \item A continuous point configuration $\mathcal{U}_{\textnormal{con}}$ consists of $N=100$ points uniformly sampled in $[0,1]$ with the minimum pairwise distance $d_{\min}\geq 5e^{-3}$ and the maximum pairwise distance $d_{\max}=1$. $\mathcal{U}_{\textnormal{con}}$ has unique pairwise distances\footnote{Since the point locations are sampled continuously in $[0,1]$, the possibility that the point configuration produces repeated distances goes to zero.}.
    \item Using different quantization steps $\lambda\in\Lambda$:
    \begin{align}
    \label{eq:discretization_set}
    \Lambda=\{\lambda_1=1e^{-4},\lambda_2=1e^{-3},\lambda_3=5e^{-3}\}\,,
    \end{align}
    We can generate discrete point configuration $\mathcal{U}_{\textnormal{dis}}$ out of $\mathcal{U}_{\textnormal{con}}$. The set of discrete point configurations can be defined as:
    \begin{definition}
    $\Omega^{(N)}(\lambda)$ is the finite set that contains all possible $N$-point configurations in the discretized 1D domain with the quantization step $\lambda$, the minimum pairwise distance $d_{\min}\geq 5e^{-3}$ and the maximum pairwise distance $d_{\max}$=1.
    \end{definition}
    Using the $\Lambda$ in \eqref{eq:discretization_set}, we can get discrete point configurations belonging to three sets $\Omega^{(N)}(\lambda_1)\supset\Omega^{(N)}(\lambda_2)\supset\Omega^{(N)}(\lambda_3)$. We study how the proposed approach performs with respect to configurations drawn from the following disjoint sets:
    \begin{subequations}
    \label{eq:omega_setminus}
    \begin{align}
    \Omega^{(N)}_{1\setminus2}&=\Omega^{(N)}(\lambda_1)\setminus\Omega^{(N)}(\lambda_2)\\
    \Omega^{(N)}_{2\setminus3}&=\Omega^{(N)}(\lambda_2)\setminus\Omega^{(N)}(\lambda_3)\\
    \Omega^{(N)}_3&=\Omega^{(N)}(\lambda_3)
    \end{align}
    \end{subequations}
    where $\Omega^{(N)}(\lambda_i)\setminus\Omega^{(N)}(\lambda_{i+1})$ is the operation that removes all the elements of $\Omega^{(N)}(\lambda_{i+1})$ from its superset $\Omega^{(N)}(\lambda_i)$. When the quantization step $\lambda$ increases, the more likely we can observe repeated distances in the measurements. This way the discrete point configuration $\mathcal{U}_{\textnormal{dis}}$ drawn from disjoint sets in \eqref{eq:omega_setminus} will produce different numbers of repeated distances.
\end{itemize}
Fig. \ref{fig:dist_rep_noiseless} shows the noiseless distance distributions produced by turnpike configurations from different sets: the $\mathcal{U}$ from $\Omega^{(N)}_{1\setminus2}$ has the least number of repeated distances, and the $\mathcal{U}$ from $\Omega^{(N)}_3$ has the highest number of repeated distances.

For each point configuration, we randomly sample $1e^4$ vectors $\vz\in\mathcal{S}$, where $\mathcal{S}$ is the convex set defined by the constraints in \eqref{eq:constrained_nonconvex}. We approximate the distributions of $\widehat{\vh}^T\mE\widehat{\vh}$ and $\sum_{y=0}^{M-1}\big(\widehat{\vh}^T\mB_y\widehat{\vh}\big)^2$ via Monte Carlo simulations. The thresholds $\{\nu_1,\nu_2\}$ are chosen so that \eqref{eq:lower_bd_heh} and \eqref{eq:upper_bd_hbh} hold with high probability $P(\nu_1)=0.999$, $P(\nu_2)=0.999$. The empirical convergence radius $\widehat{\tau}$ can then be computed with the constant $\theta$ set to $0.95$. 

Fig. \ref{fig:dist_monte_carlo_noiseless} also shows how the distance distribution matching approach recovers a point configuration under different initializations. The projected gradient descent is initialized with 100 different random initializations ({\bfseries R}) and the spectral initialization ({\bfseries S}). From Fig. \ref{fig:dist_monte_carlo_noiseless}, we can see that both the random and spectral initializations are not inside the empirical convergence neighbourhood $\mathcal{E}(\widehat{\tau})$ at the beginning. It becomes increasingly difficult for the iterate $\vz_t$ to reach $\mathcal{E}(\widehat{\tau})$ as more repeated distances start to appear in the noiseless measurements. In fact, it is also more difficult for the backtracking approach by Skiena et al. to backtrack to the right path when there are many repeated distances, since their algorithm only backtracks when it has used up suitable distances to build its current path. 

Additionally, since $\widehat{\tau}$ is an empirical radius that relies on \eqref{eq:upper_bd_hbh} and \eqref{eq:lower_bd_heh} to hold with high probabilities, even if the iterates $\vz_t$ reaches $\mathcal{E}(\widehat{\tau})$, there is still a small chance for the next iterate $\vz_{t+1}$ to pull away from $\vx$, as illustrated by the noiseless case from $\Omega^{(N)}_{1\setminus2}$ in Fig. \ref{fig:dist_monte_carlo_noiseless}.

\end{noiseless_case}

\begin{noisy_case}
In the noisy case the global optimizer is no more binary; we have $\vx\in[0,1]^M$. There is no guarantee that a convergence neighborhood exists around the ground truth signal. Since we only have access to noisy distances and need to approximate the oracle distance distribution $g(d)$ using $p(d)$ in \eqref{eq:approx_dist}, the global optimizer $\vx$ would be a perturbation of the ground truth signal which depends on the realization of noise and is unknown in practice. While this precludes numerical estimates of $\widehat{\tau}$ as in the noiseless case, the algorithm empirically still converges to a solution that, after agglomerative clustering, leads to successful recovery of the point locations under moderate levels of noise. This is confirmed by extensive numerical experiments in Section \ref{sec:exp:turnpike}. We thus leave a theoretical analysis of convergence in the presence of noise an open question.
\end{noisy_case}

\subsubsection{Analysis of Difficulty of Recovery}
\label{subsubsec:turnpike_mutual_info}
We now quantitatively characterize the difficulty of recovering different types of point configurations from an information-theoretic perspective. We define the following two random variables:
\begin{definition}
The point $X\in\{1,\cdots,M\}$ is a random variable with distribution $P(X=m)=\frac{1}{N}x_m$, where $x_m$ is the $m$-th entry of the ground truth signal $\vx$, and $P(X=m)$ corresponds to the normalized point density at the $m$-th segment $l_m$ in the 1D discrete domain.
\end{definition}
\begin{definition}
The distance $Y\in\{0,\cdots,M-1\}$ is a random variable with conditional distribution $P(Y=y|X=m)=\sum_{k=1}^MP(X=k)\delta(y=|k-m|)$. Its marginal distribution is
\begin{align}
\label{eq:it_analysis_distance_dist}
\begin{split}
P(Y=y)&=\sum_{m=1}^MP(X=m)P(Y=y|X=m)\\
&=\frac{1}{N^2} \sum_{m=1}^M\sum_{k=1}^M x_mx_k\delta(y=|k-m|)\,.
\end{split}
\end{align}
\end{definition}
We note that the above \eqref{eq:it_analysis_distance_dist} is equivalent to the distance distribution $p(y)$ in \eqref{eq:quad_form_p} in the sense that they both correspond to the same distance multiset $\mathcal{D}$.

The mutual information $I(X;Y)$ between the point $X$ and the distance $Y$ measures the information shared between $X$ and $Y$. We have:
\begin{align}
\label{eq:mutual_info}
\begin{split}
    I(X;Y)&=H(Y)-H(Y|X)\,,
\end{split}
\end{align}
where $H(Y)$ is the entropy of $Y$ and $H(Y|X)$ is the conditional entropy of $Y$ given $X$. A higher $I(X;Y)$ suggests it is easier to recover $X$ (and, by extension, $\mathcal{U}$) from $Y$. For every point configuration $\mathcal{U}\in\Omega^{(N)}_{i\setminus(i+1)}$, we can then compute $I(X;Y)$ according to the above \eqref{eq:mutual_info} exactly, and use it to quantitatively characterize the difficulty of recovery. The maximum mutual information is $\max (I(X;Y))=H(X)$, which is obtained when $H(X|Y)=0$. We emphasize that the mutual information estimate is a heuristic proxy for the difficulty of recovery. Whether one can rigorously connect mutual information to the optimization landscape remains an open question.

Since we do not know the true signal $\vx$ beforehand, the mutual information could not be computed in practice. In this case, the number of repeated distances in noiseless measurements is what we could observe. Performing Monte Carlo simulations under the same setting as before, we show the mutual information values (mean$\pm$standard deviation) in Table \ref{tab:turnpike_mutual_information} based on $100$ point configurations in each $\Omega^{(N)}_{i\setminus(i+1)}$. We can see that $I(X;Y)$ decreases when there are more repeated distances in noiseless measurements. This empirical correlation between the mutual information and the number of repeated distances in noiseless measurement suggests that the latter could serve as a coarse proxy to mutual information.

\begin{table}[tbp]
\caption{Turnpike ($N=100$): mutual information.}
\label{tab:turnpike_mutual_information}
\centering
\begin{tabular}{@{}lccc@{}}
\toprule
 & $\Omega^{(N)}_{1\setminus2}$ & $\Omega^{(N)}_{2\setminus3}$ & $\Omega^{(N)}_3$ \\ \cmidrule{2-4} 
$I(X;Y)$ & 3.877$\pm$0.014 & 2.695$\pm$0.020 & 1.194$\pm$0.019 \\ 
\bottomrule
\end{tabular}
\end{table}

\subsection{Analysis of the Beltway Problem}

\begin{figure*}[tbp]
\centering
\subfigure{
\label{fig:dist_bw_rep_1}
\includegraphics[height=0.6in]{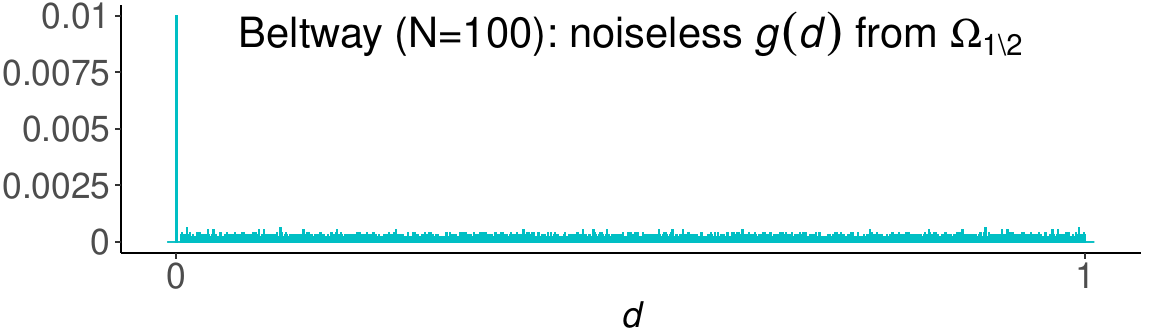}}
\subfigure{
\label{fig:dist_bw_rep_10}
\includegraphics[height=0.6in]{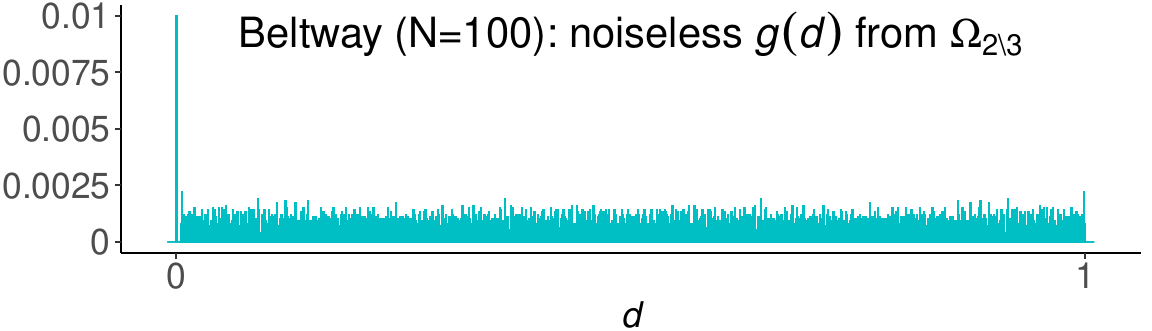}}
\subfigure{
\label{fig:dist_bw_rep_50}
\includegraphics[height=0.6in]{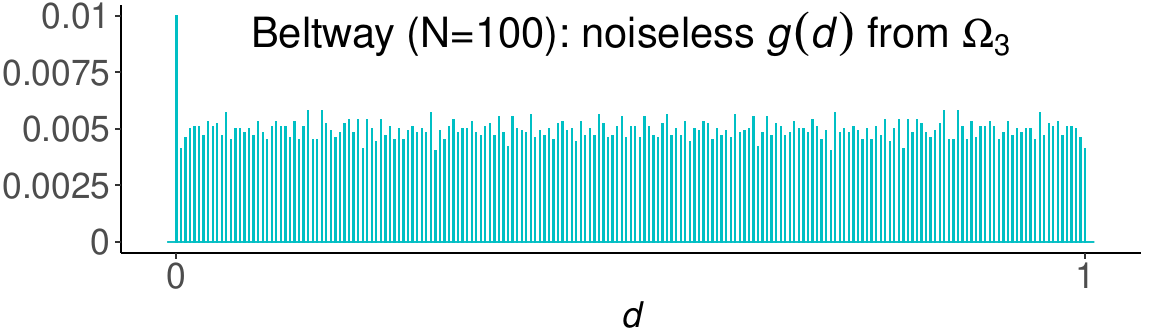}}
\caption{Given the quantization step $\lambda\in\Lambda=\{1e^{-4},1e^{-3},5e^{-3}\}$, the beltway point configurations drawn from different sets $\Omega^{(N)}_{1\setminus 2},\Omega^{(N)}_{2\setminus 3},\Omega^{(N)}_3$ defined in \eqref{eq:omega_setminus} have different number of repeated distances in the noiseless measurements.}
\label{fig:dist_bw_rep_noiseless}
\end{figure*}

\begin{figure*}[tbp]
\centering
\subfigure{
\label{fig:dist_bw_1_0}
\includegraphics[height=1.7in]{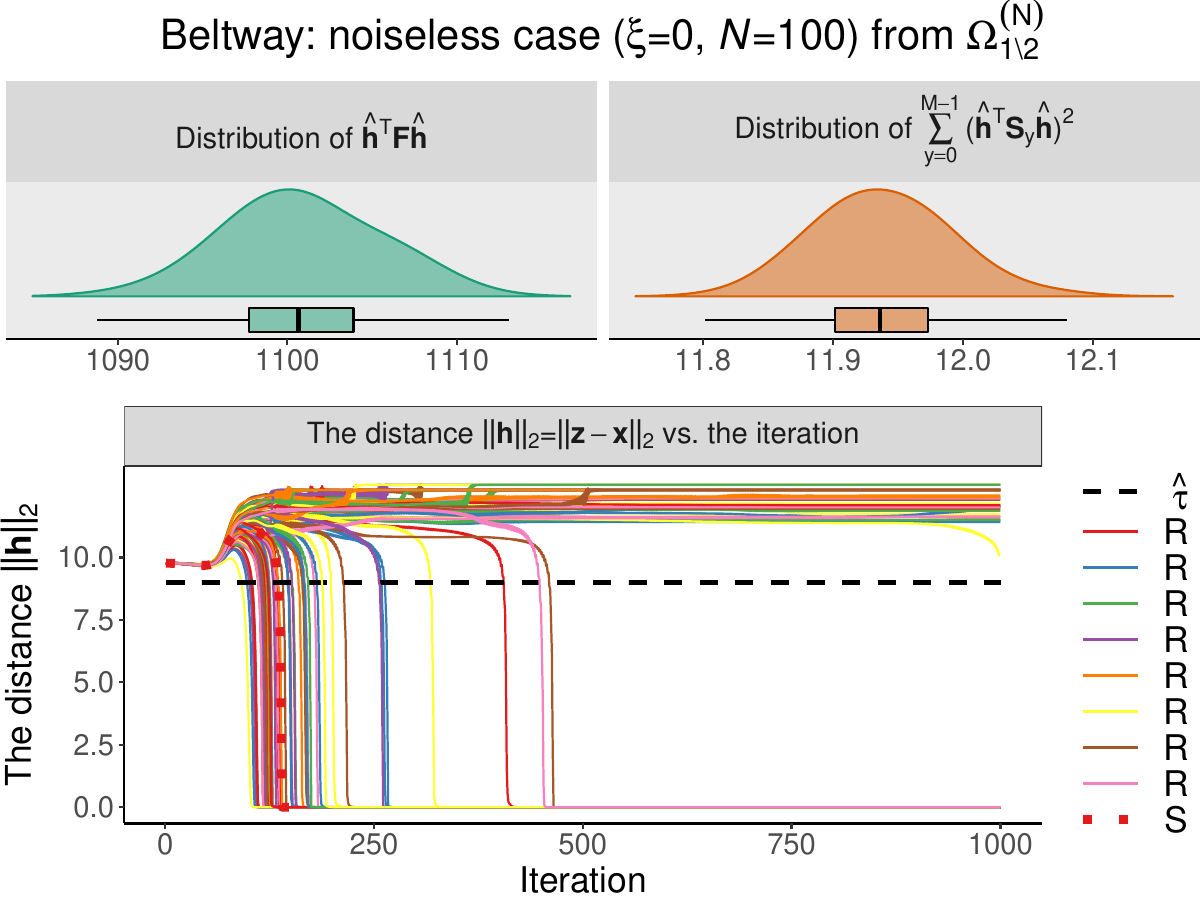}}
\subfigure{
\label{fig:dist_bw_10_0}
\includegraphics[height=1.7in]{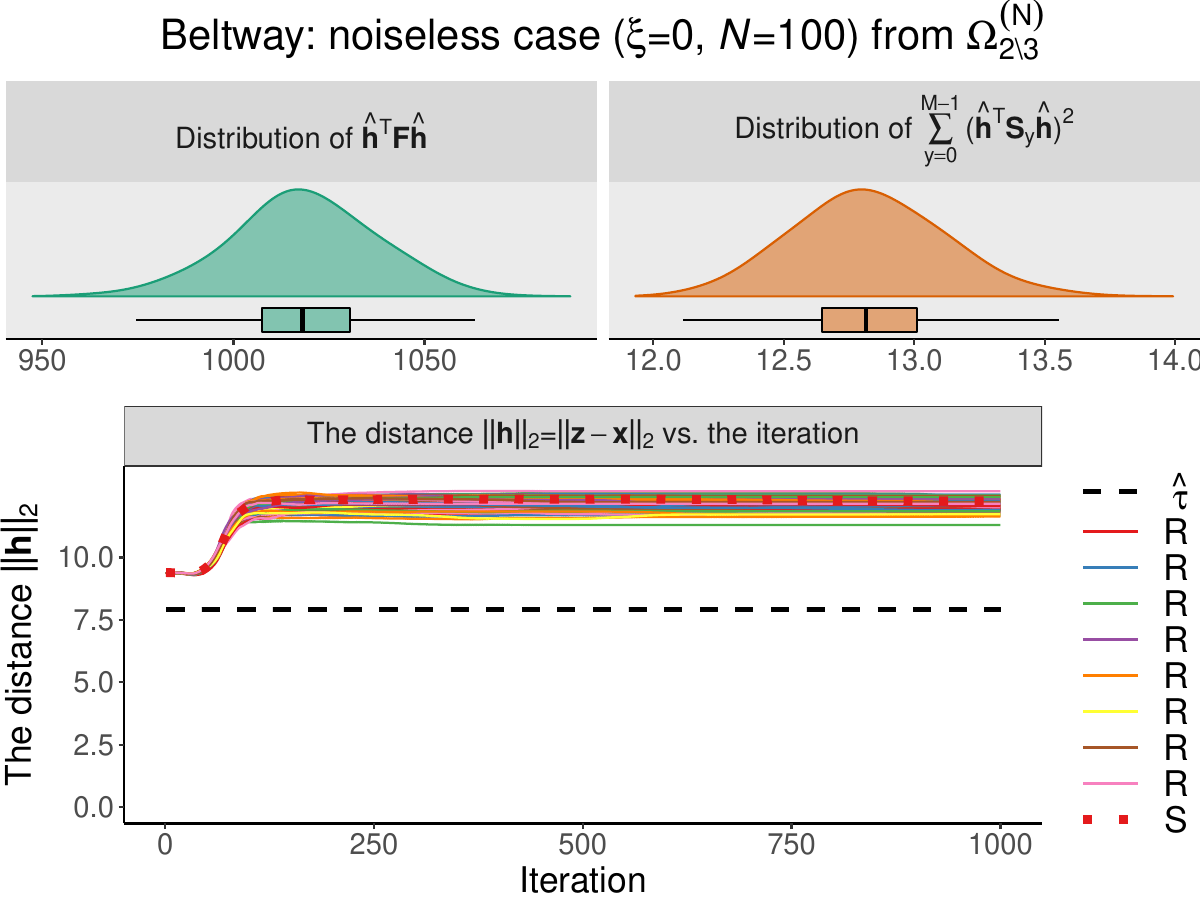}}
\subfigure{
\label{fig:dist_bw_50_0}
\includegraphics[height=1.7in]{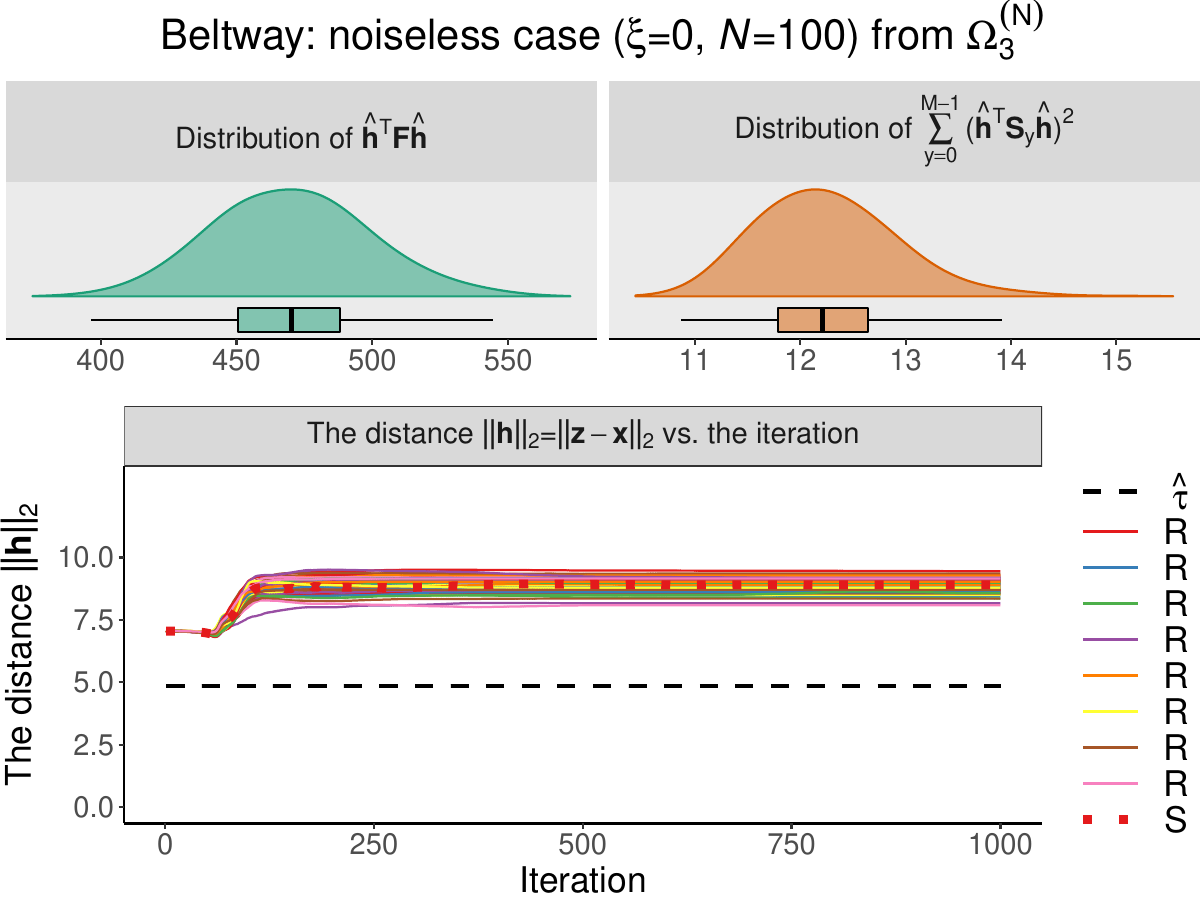}}
\caption{In the noiseless beltway problem, the number of repeated distances in the measurements affects the performance of projected gradient descent from two aspects: 1) the radius of the empirical convergence neighbourhood $\mathcal{E}(\widehat{\tau})$; 2) the difficulty to reach $\mathcal{E}(\widehat{\tau})$ from outside.}
\label{fig:dist_bw_monte_carlo_noiseless}
\end{figure*}

For the beltway problem, we have similar results given in the Supplementary Material. We also run Monte Carlo simulations to show how \emph{the number of repeated distances} in the noiseless measurements affects the reconstruction performance. The point locations and the experimental setup are the same as those of the turnpike problem in Section \ref{subsec:turnpike_cvg}, except that the points are now migrated onto a loop with length $L=1.005$ where $d(u_1\rightarrow u_N)=1$ and $d(u_N\rightarrow u_1)=0.005$. 

\begin{table}[tbp]
\caption{Beltway ($N=100$): mutual information.}
\label{tab:beltway_mutual_information}
\centering
\begin{tabular}{@{}lccc@{}}
\toprule
 & $\Omega^{(N)}_{1\setminus2}$ & $\Omega^{(N)}_{2\setminus3}$ & $\Omega^{(N)}_3$ \\ \cmidrule{2-4} 
$I(X;Y)$ & 3.605$\pm$0.008 & 2.239$\pm$0.004 & 0.693$\pm$0.001 \\ 
\bottomrule
\end{tabular}
\end{table}

Fig. \ref{fig:dist_bw_rep_noiseless} shows the noiseless distance distributions from different types of point configurations. As shown in Fig. \ref{fig:dist_bw_monte_carlo_noiseless}, when there are more repeated distances in the noiseless measurements, the empirical radius $\widehat{\tau}$ decreases and it is more difficult for the spectral and random initializations to reach $\mathcal{E}(\widehat{\tau})$ from outside. By comparing the beltway recoveries in Fig. \ref{fig:dist_bw_monte_carlo_noiseless} with the turnpike recoveries in Fig. \ref{fig:dist_monte_carlo_noiseless}, we can see that the empirical convergence radii $\widehat{\tau}$ are similar in both problems. However, it is more difficult for the iterate $\vz_t$ to reach $\mathcal{E}(\widehat{\tau})$ from outside in the beltway problem. This is corroborated by the analysis on difficulty of recovery where we use the mutual information $I(X;Y)$ as a proxy to the hardness of reconstruction. Table \ref{tab:beltway_mutual_information} shows that the estimated $I(X;Y)$ in each case of the beltway recovery is lower than the corresponding turnpike recovery, indicating it is more difficult to recover beltway point configurations.

\section{Experimental Results}
\label{sec:exp}
In this section we compare the proposed distance distribution matching approach with the current state-of-the-art backtracking approach for the turnpike recovery, and show that our approach can solve large-scale noisy beltway recovery (to the best of our knowledge this is the first such algorithm). Reproducible code and data are available online at \urlstyle{tt}\url{https://github.com/swing-research/turnpike-beltway}

\subsection{Noiseless Partial Digestion on Real Genome Data}
\label{subsec:pd_real_genome}

\begin{figure}[tbp]
\centering
\includegraphics[height=1.5in]{figures/partial_digestion}
\caption{Partial digestion of a DNA with the restriction enzyme when $N=5$.}
\label{fig:partial_digestion}
\end{figure}

As illustrated in Fig. \ref{fig:partial_digestion}, suppose the target sequence appears in the restriction sites $\{u_2,u_3,u_4\}$ with $\{u_1,u_5\}$ being the two ends of the DNA. By carefully controlling the amount of restriction enzyme and the time of the digestion, the DNA is partially digested. This means that the enzyme cleaves the molecule randomly, producing fragments of different lengths that correspond to the pairwise distances between pairs of restriction sites \cite{PD:1976,Waterman:1995}. We experiments with \emph{E. Coli} K12 MG1655 genome data from the GenBank$^{\tiny{\text{\textregistered}}}$ assembly \cite{E:Coli:K12:MG1655}, which is a nucleotide sequence of length$=4,641,652$. Four letters \texttt{A}, \texttt{C}, \texttt{G}, \texttt{T} are used to represent the four nucleotide bases of the DNA molecule \cite{NucleotidesSymbol:1985}. The list of restriction enzymes used in the experiments and the number of restriction sites (including the two ends \texttt{5'} and \texttt{3'} as dummy restriction sites) are shown in Table \ref{tab:restriction_enzyme}. Note that the recognition sequence could also be in a reverse order depending on which way the nucleotide sequence is read. 

Since there are four nucleotide bases that cannot be further digested, the unlabeled pairwise distances are all integers in this case, and the DNA sequence has a total of $M=4,641,653$ equally spaced possible locations for the restriction sites. Note that the matrix $\mA_y$ has a simple structure and thus needs not be stored during computation. Using our distance distribution matching approach, we correctly recover all site locations (Table \ref{tab:restriction_enzyme}). The runtimes were tested on a Quad-core processor machine (Intel Xeon X5355) with 24 GB RAM. Each processor has 2.66 GHz speed with 8MB of cache, and all 4 cores were used. It took approximately 25 and 32 minutes to reconstruct the locations of SmaI and BamHI respectively.

\begin{table}[tbp]
\caption{Restriction enzymes used in the partial digestion.}
\label{tab:restriction_enzyme}
\centering
\begin{tabular}{llllll}
\toprule
Enzyme &Recognition sequence &$N$ \\ \midrule
SmaI &\texttt{5'---CCC \textcolor{magenta}{|} GGG---3'} &$495$\\
BamHI &\texttt{5'---G \textcolor{magenta}{|} GATCC---3'} & $512$\\ \bottomrule
\end{tabular}
\end{table}

\begin{figure*}[tbp]
\centering
\subfigure{
\label{fig:n10_compare}
\includegraphics[height=1.25in]{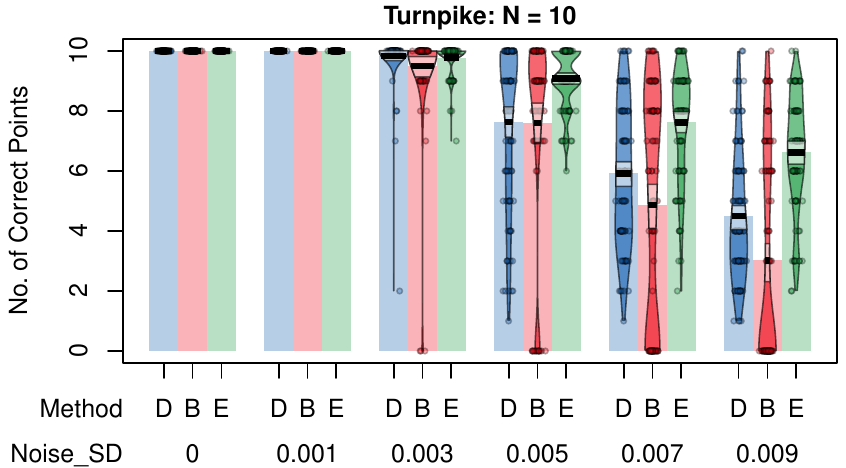}}
\subfigure{
\label{fig:n100_compare}
\includegraphics[height=1.25in]{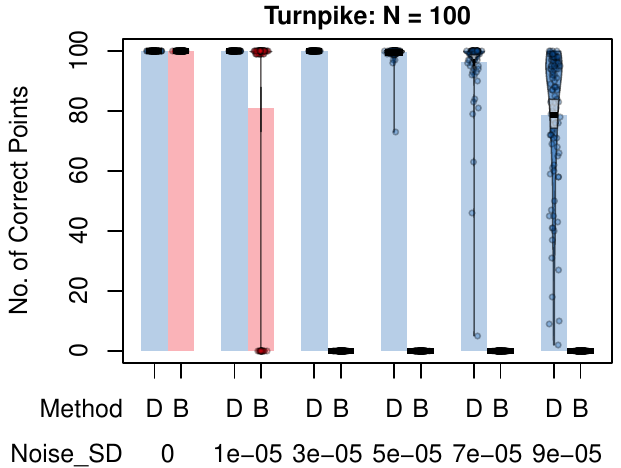}}
\caption{The distribution and the mean of the number of correctly recovered points across $100$ random trials in the \emph{``turnpike''} recovery experiments using the distance distribution matching approach ({\bfseries D}), the backtracking approach ({\bfseries B}) and the exhaustive search ({\bfseries E}). In each random trial, $N$ points are uniformly sampled from the interval $[0,1]$. When $N=10$, the smallest distance between two different points is set to $d_{\min}=1e^{-2}$. When $N=100$, we set $d_{\min}=1e^{-4}$. The distances are further corrupted with white Gaussian noise $w\sim\mathcal{N}(0,\xi^2)$, where we control $\xi<d_{\min}$.}
\label{fig:tp_compare}
\end{figure*}

\subsection{Turnpike Recovery on Simulated Data}
\label{sec:exp:turnpike}

In the turnpike recovery experiments where the points are located on a line, we compare the distance distribution matching approach and the state-of-the-art backtracking approach by \cite{Skiena:1994} through simulated noisy recovery experiments. We first uniformly sample $N=10$ points from the interval $[0,1]$ with the minimum pairwise distance between two different points set to $d_{\min}=1e^{-2}$ and the maximum pairwise distance set to $d_{\max}=1$. The length $L$ of the line $\boldsymbol l$ thus equals $d_{\max}$. The quantization step is set to $\lambda=1e^{-3}$ to balance the trade-off between reducing the quantization error and computational complexity, creating $M=\frac{L}{\lambda}=1e^3$ possible locations for the $10$ points. The distance measurement $d_k$ is corrupted with white Gaussian noise, $w\sim\mathcal{N}(0,\xi^2)$. We control the noise level by varying the standard deviation of the noise: $\xi\in\{0,1e^{-3}, 3e^{-3}, 5e^{-3}, 7e^{-3}, 9e^{-3}\}$. The results obtained when $\xi=0$ correspond to the case where there is only quantization error and no measurement noise. 

For the distance distribution matching approach, the unlabeled pairwise distance measurements are collected and extended to form the multiset $\mathcal{D}$. As discussed in Section \ref{subsec:ddm}, the parameter $\sigma$ in the approximated distribution $p(d)$ is unknown, and can be tuned in practice to obtain best performance. In the experiments, $\sigma$ is tuned in the interval $(0, d_{\min}=1e^{-2})$, producing multiple solutions corresponding to each $\sigma$. We shall choose the solution whose distance distribution is closest to the observed distance distribution in terms of the earth mover's distance \cite{EMD:2001}. The exact recovered point locations $\{\widehat{u}_1,\widehat{u}_2,\cdots,\widehat{u}_N\}$ are obtained using the aforementioned agglomerative clustering method in Section \ref{subsec:ddm}. For each noise level specified by $\xi$, $100$ random trials are performed and the number of correctly recovered points is recorded for each random trial. The runtime is about 1s for each noise level in a random trial.

For the backtracking approach, the search path for every distance $d_k$ is performed in an interval $[d_k-\delta_d,\ d_k+\delta_d]$. In order to make a fair comparison, we need to ensure that both approaches are evaluating the distance $d_k$ within roughly the same range. Here we choose $\delta_d=5\sigma_{\max}=5e^{-2}$, where $\sigma_{\max}$ is the largest $\sigma$ tuned by the distance distribution matching approach. The runtime is less than 0.01s for each noise level in a random trial. We should note that the best results are obtained by choosing $\delta_d=1$, i.e. the maximum pairwise distance. However, this essentially becomes performing an exhaustive search over all possible paths, the complexity grows exponentially. It is simply impractical when the number of points $N$ and the number of possible locations $M$ are large. Since there are only $10$ points to be recovered in this case, we also compute the solution obtained via the exhaustive search as a comparison, which corresponds to the best solution one can hope to achieve given noisy measurements. 

The recovered point locations can be matched to the true locations efficiently using the Hungarian algorithm \cite{Hungarian:1955}. If the distance between a recovered location $\hat{u}_n$ and the true location $u_n$ is less than $\frac{1}{2}d_{\min}$, the recovery of the $n$-th point is considered to be a success. The recovery results across $100$ random trials when $N=10$ are shown in Fig. \ref{fig:tp_compare}: we would like to show the average number of correct points through bar plots and the distribution of the number of correct points through violin plots. Every dot in the violin plot corresponds to the number of correct points in a random trial. The width of the violin plot corresponds to the density of the dots. The shape of the violin thus shows the distribution of the number of correct points across 100 random trials. We can see that our approach is significantly more robust to noise compared to the backtracking approach, and offers a competitive alternative to the exhaustive search approach.

In order to test how the two approaches are holding up against large-scale problems, we then uniformly sample $N=100$ points from the interval $[0,1]$ as before, with the minimum pairwise distance set to $d_{\min}=1e^{-4}$ and the maximum pairwise distance set to $d_{\max}=1$. The distance measurement $d_k$ is also corrupted with white Gaussian noise $w\sim\mathcal{N}(0,\xi^2)$, where $\xi\in\{0,1e^{-5}, 3e^{-5}, 5e^{-5}, 7e^{-5}, 9e^{-5}\}$. The quantization step is set to $\lambda=1e^{-5}$, creating $M=\frac{L}{\lambda}=1e^5$ possible locations for the $100$ points.

For the distance distribution matching approach, the standard deviation $\sigma$ in the noise model is tuned in the interval $\sigma\in[0,d_{\min}=1e^{-4}]$, the runtime is about 30 minutes for each noise level in a random trial. For the backtracking approach, the tolerance threshold $\tau_d$ is chosen to be $\tau_d=5\sigma_{\max}=5e^{-4}$, the runtime is about 0.2s for each noise level in a random trial. Since $N$ and especially $M$ are much larger in this case, we are not able to perform an exhaustive search for comparison here. The recovery results across $100$ random trials when $N=100$ are shown in Fig. \ref{fig:tp_compare}. We can see that the proposed approach is more robust and has greater advantage over the backtracking approach for large-scale problems. When the noise level is high, the backtracking approach is not able to produce solutions, the proposed approach does not break down completely and recovers some of the points correctly.

\subsection{Beltway Recovery on Simulated Data}
\label{sec:exp:beltway}

\begin{figure}[tbp]
\centering
\subfigure{
\label{fig:n10_compare_bw}
\includegraphics[width=1.6in]{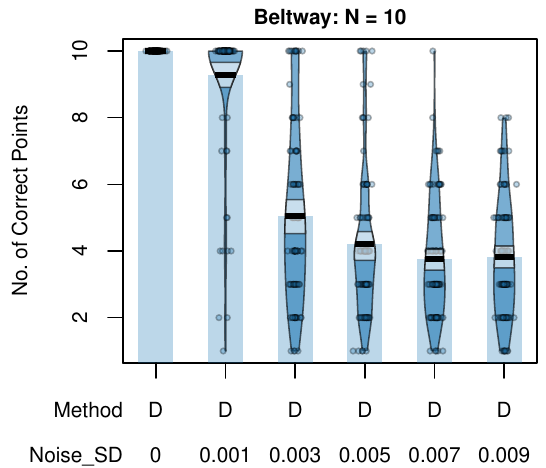}}
\subfigure{
\label{fig:n100_compare_bw}
\includegraphics[width=1.6in]{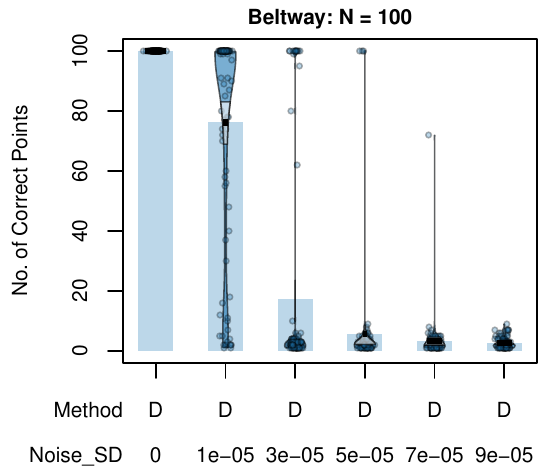}}
\caption{The distribution and the mean of the number of correctly recovered points across $100$ random trials in the \emph{``beltway''} recovery experiments using the distance distribution matching approach ({\bfseries D}). In each random trial, $N$ points are uniformly sampled from a loop of length $L=d_{\min}+d_{\max}$, where the largest pairwise distance $d_{\max}$ is set to $1$, the smallest distance $d_{\min}$ between two different points is set to $1e^{-2}$ when $N=10$ and $1e^{-4}$ when $N=100$. The distances are further corrupted with white Gaussian noise $w\sim\mathcal{N}(0,\xi^2)$, where we control $\xi<d_{\min}$.}
\label{fig:bw_compare}
\end{figure}

We next use the proposed distance distribution matching approach to perform the beltway recovery experiments where the points lie on a loop. To the best of our knowledge, our approach is the first practical approach that can solve the large-scale beltway problem efficiently. Note that the exhaustive search is impractical even when $N$ is small but $M$ is large \cite{Lemke2003}. Hence we only present the recovery results obtained using the proposed approach here. We uniformly sample $N$ points from a loop of length $L=d_{\min}+d_{\max}$, where $d_{\min}$ is the minimum distance between two different points and $d_{\max}$ is the maximum pairwise distance. When $N=10$, we set $d_{\min}=1e^{-2}$ and $d_{\max}=1$. The distance $d_k$ is also corrupted with a white Gaussian noise: $w_k\sim\mathcal{N}(0,\xi^2)$, where $\xi\in\set{0,1e^{-3}, 3e^{-3}, 5e^{-3},7e^{-3},9e^{-3}}$. The quantization step is set to $\lambda=1e^{-3}$, creating $M=\frac{L}{\lambda}=1.01e^3$ possible locations for the $10$-points case. The runtime is about 1s for each noise level in a random trial. When $N=100$, we set $d_{\min}=1e^{-4}$ and $d_{\max}=1$. The standard deviation of the white Gaussian noise is chosen from $\xi\in\set{0,1e^{-5},3e^{-5},5e^{-5},7e^{-5},9e^{-5}}$ as before, and the quantization step is set to $\lambda=1e^{-5}$, creating $M=\frac{L}{\lambda}=1.0001e^5$ possible locations for the $100$-points case. The runtime is about 1$\sim$1.5 hours for each noise level in a random trial. The recovery results across $100$ random trials are shown in Fig. \ref{fig:bw_compare}. We can see that the proposed approach is able to reconstruct all the point locations correctly when there is only quantization error and no measurement noise, i.e. $\xi=0$. When measurement noise is added, the proposed approach could still recover some of the points correctly.

\subsection{Comparison of Initialization Schemes}

\begin{figure*}[tbp]
\centering
\subfigure{
\label{fig:init_tp_n100_compare}
\includegraphics[width=.3\textwidth]{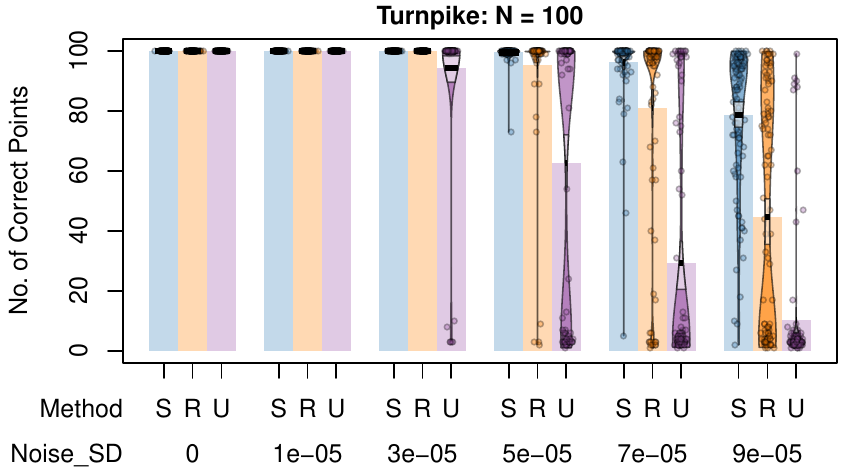}}
\subfigure{
\label{fig:init_bw_n100_compare}
\includegraphics[width=.3\textwidth]{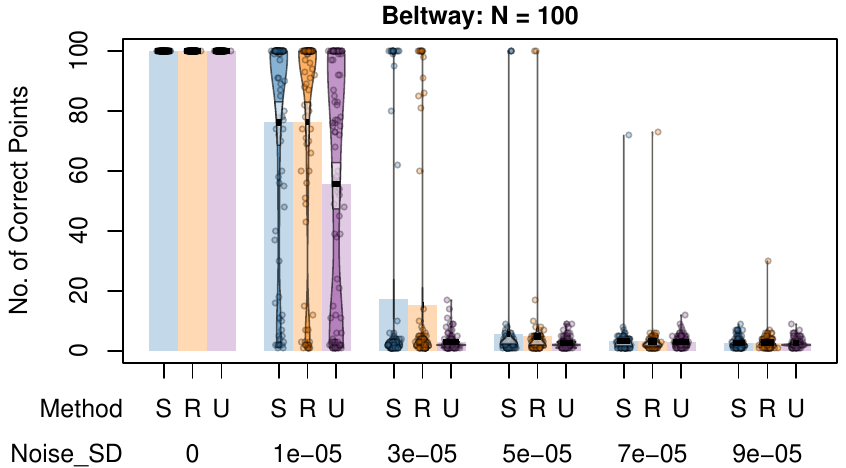}}
\caption{The distribution and the mean of the number of correctly recovered points across $100$ random trials in the turnpike and beltway recoveries comparing the \emph{``three initialization schemes''}: the spectral initialization ({\bfseries S}), the random initialization ({\bfseries R}), and the uniform initialization ({\bfseries U}). In each random trial, $N=100$ points are uniformly sampled from the interval $[0,1]$, with the smallest distance between two different points set to $d_{\min}=1e^{-4}$. The distances are further corrupted with white Gaussian noise $w\sim\mathcal{N}(0,\xi^2)$, where we control $\xi<d_{\min}$.}
\label{fig:init_compare}
\end{figure*}

A spectral initialization scheme is adopted in the distance distribution matching approach to solve the nonconvex turnpike and beltway recoveries. It is meant to provide a good initializer that highlights the possible point locations. Here we put it to test and compare it with the other two initialization schemes, i.e. the ``random'' initialization and the ``uniform'' initialization. In the random initialization scheme, the entries of the initializer $\vz_0$ are generated independently according to the white Gaussian distribution $\mathcal{N}(0,0.01)$. In the uniform initialization scheme, the entries of $\vz_0$ are set to all ones. We should note that the initializers from all three schemes are projected to the convex set $\mathcal{S}$ defined by the constraints in \eqref{eq:constrained_nonconvex} before they can be used with the projected gradient descent. Following the same settings when $N=100$ as in Sections \ref{sec:exp:turnpike} and \ref{sec:exp:beltway}, we perform simulated noisy turnpike and beltway recoveries using the three initialization schemes. The recovery results are shown in Fig. \ref{fig:init_compare}. For the turnpike recovery, the spectral initialization is more robust than the other two schemes. For the beltway recovery, the spectral and random initializations perform almost equally well, and they both perform better than the uniform initialization.

\section{Conclusion}
\label{sec:con}
We introduced a new method to solve two important unlabeled distance geometry problems in 1D: the turnpike and the beltway. Our aim was to find an approach that is computationally efficient and that can deal with imprecise, noisy data. While some earlier methods are efficient on typical runs with perfect data, they become inoperable or impractical when faced with noise. This is not surprising as these approaches are either based on factoring polynomials or on clever variants of exhaustive search. In the latter case, the extra branching due to noise quickly explodes, especially for large-scale problems.

We propose an alternative based on nonconvex programming. The key ingredient is a suitable global objective function which involves all the measured distances and all the unknown points, so that the method looks for all the points at once. By first modeling the distance distribution as a collection of quadratic functionals of the unknown point and then using recent ideas in non-convex optimization, the proposed distance distribution matching approach achieves both stated goals. Numerical experiments with real and synthetic data show that it significantly outperforms the state-of-the-art backtracking approach for the turnpike problem. To the best of our knowledge, it is also the first practical and computationally efficient method for the large-scale noisy beltway problem. 

One drawback comes from using a gradient-based optimization method: we lose the ability to list all solutions when uniqueness does not hold, unlike some of the search-based methods which naturally produce the desired list \cite{Lemke2003,Abbas2016}. We were also not able to provide theoretical guarantees that the introduced spectral initializer converges to a global optimum. Due to the hardness of the noisy problem, we expect this to hold with high probability over probabilistic point set models that contain mostly distinct distances; empirically, this is indeed the case. Another drawback comes from the discretization of the domain space: the fine discretization needed for the recovery would increase the problem size inevitably. We tried to bypass the discretization by directly optimizing with respect to the point locations. However, the optimization easily got stuck in some local optimum. It would only work if we had an initialization that was very close to the true solution. Currently we prune the domain space to reduce the problem size by removing locations that are not consistent with distance measurements. This allows us to extend the proposed approach to solve a related 3D unknown view tomography problem \cite{Zehni:3D:2020}. Notwithstanding these drawbacks, our method can be used to solve large-scale unassigned problems with noise. It thus opens up avenues for new biological applications similar to the recent de novo cyclic peptide sequencing via mass spectrometry \cite{CyclicPeptide:2011,Kavan2013}.

\begin{appendices}
\section{Proofs for Projected Gradient Descent}

\subsection{Proof of Lemma \ref{lemma:upper_bound}}
\label{proof:lemma:upper_bound}
\begin{proof}
Suppose that $s_i<1$. We construct a vector $\widetilde{\vs}\in\mathbb{R}^M$ out of $\vs$ by swapping the positions of $s_i$ and $s_j$ in $\vs$, i.e. $\widetilde{s}_i=s_j$ and $\widetilde{s}_j=s_i$. $\widetilde{\vs}$ also satisfies the constraints in (\ref{eq:projection_box_constraints}). Since $\overline{z}_i>\overline{z}_j$ and $s_j=1$, we then have:
\begin{align}
\begin{split}
\|\vs-\overline{\vz}\|_2^2-\|\tilde{\vs}-\overline{\vz}\|_2^2=2(1-s_i)(\overline{z}_i-\overline{z}_j)> 0\,.
\end{split}
\end{align}
This is in contradiction with the fact that $\vs$ is the minimizer of (\ref{eq:projection_box_constraints}). Hence $s_i$ must be $1$.
\end{proof}

\subsection{Proof of Lemma \ref{LEMMA:UNIQUE_R}}
\label{proof:lemma:unique_r}
\begin{proof}
Let $\mathcal{S}$ denote the convex set defined by the constraints $0\leq s_m\leq 1$, $\forall\ 1\leq m\leq M$ and $\sum_{m=1}^Ms_m=N$. Note that the entries of $\overline{\vz}$ and $\vx$ are in a non-increasing order. We will proceed in the following two steps:

\emph{\bfseries Step 1)} Since $\mathcal{S}$ is non-empty, the projection onto it exists, i.e. there is one $r\in\set{1,\cdots,N}$ that produces the $(\rho,\kappa)$ that satisfy $1\leq \overline{z}_{r-1}-\kappa$ if $2\leq r\leq N<\rho$ and $0<\overline{z}_r-\kappa<1$. In fact, $\mathcal{S}$ is a closed convex set, the projection is also unique.

\emph{\bfseries Step 2)} Suppose that there are two different $r_1<r_2\in\set{1,\cdots,N}$ that produce the two pairs $(\rho_1,\kappa_1|r_1)$ and $(\rho_2,\kappa_2|r_2)$ that satisfy the constraints (\ref{eq:cst_r}) and (\ref{eq:cst_rm1}). As detailed in the Supplementary Material, we can prove that
\[r_1-1+\sum_{m=r_1}^{\rho_1}(\overline{z}_m-\kappa_1)\,<\,r_2-1+\sum_{m=r_2}^{\rho_2}(\overline{z}_m-\kappa_2)\,.\]

Let $\vs_1,\vs_2$ denote the solutions of \eqref{eq:projection_box_constraints} produced by $r_1,r_2$ respectively. The above inequality shows that $\sum_{m=1}^Ms_{1m}<\sum_{m=1}^Ms_{2m}$. This is in contradiction with the assumption that $\sum_{m=1}^Ms_{1m}=\sum_{m=1}^Ms_{2m}=N$. Hence $r_1=r_2$, there is only one $r\in\set{1,\cdots,N}$ that produces the $(\rho,\kappa)$ that satisfy the constraints (\ref{eq:cst_r}) and (\ref{eq:cst_rm1}).
\end{proof}

\section{Proofs for Convergence Analysis}

\subsection{Lemma ~\ref{LEMMA:E_MIN}}
\label{proof:lemma:e_min}

\begin{lemma}
\label{LEMMA:E_MIN}
Let $\mB_y=\mA_y+\mA_y^T$, $\mE=\sum_{y=0}^{M-1}\mB_y\vx{\vx}^T\mB_y^T$, $\mathcal{S}$ be the convex set defined in \eqref{eq:convex_set}. The following problem is convex $\forall\ \vz\in\mathcal{S}$, $\vz\neq\vx$:
\begin{align}
\label{eq:lambda_min_val}
\begin{split}
\mu_\mE = \min_{\vz\in\mathcal{S}, \vz\neq\vx}\,\frac{(\vz-\vx)^T\mE(\vz-\vx)}{\|\vz-\vx\|_1^2}=\min_{\overline{\vh}\in\mathcal{G}}\,\overline{\vh}^T\mE\overline{\vh}\,,
\end{split}
\end{align}
where $\mu_\mE>0$ and $\overline{\vh}=\frac{\vz-\vx}{\|\vz-\vx\|_1}$, $\mathcal{G}$ is a convex set defined by the following constraints:
\begin{align}
\label{eq:hh_con_2}
&\sum_{i=1}^M\overline{h}_i=0\\
\label{eq:hh_con_3}
&\overline{h}_i\in[0,\,0.5]\quad\textnormal{if $x_i=0$}\\
\label{eq:hh_con_4}
&\overline{h}_i\in[-0.5,\,0]\quad\textnormal{if $x_i=1$}\\
\label{eq:hh_con_1}
&\|\overline{\vh}\|_1=\vr^T\overline{\vh}=1\,,
\end{align}
where $\vr\in\{-1,1\}^M$ depends on $\vx$ and is defined as follows:
\begin{align}
    r_i=1\quad\textnormal{if $x_i=0$};\quad r_i=-1\quad\textnormal{if $x_i=1$}\,.
\end{align}
\end{lemma}
\begin{proof}
The proof is given in the Supplementary Material.
\end{proof}

\subsection{Proof of Theorem ~\ref{THM:CVG_SED}}
\label{proof:thm:cvg_sed}
\begin{proof}
Let $\mB_y=\mA_y+\mA_y^T$. The gradient $\nabla f(\vz)$ then becomes
\begin{align}
\begin{split}
    \nabla f(\vz)= \frac{1}{MK^2}\sum_{y=0}^{M-1}\mB_y\vz\cdot(\vz-\vx)^T\mB_y(\vz+\vx)\,.
\end{split}
\end{align}
We first establish the regularity condition $RC(\alpha,\beta,\tau)$ for the gradient descent update, and then use it to prove the linear convergence of the projected gradient descent update.

\emph{\bfseries Step 1)} Our goal is then to find the radius $\tau$ so that $RC(\alpha,\beta,\tau)$ holds. 

We first try to find an upper bound on $\|\nabla f(\vz)\|_2^2$. Using $\sigma_{\max}^2\left(\mB_y\right)\leq 4$ and $\|\vz\|_2^2\leq\sum_mz_m= N$, we then have:
\begin{align}
\label{eq:first_term_lb}
\begin{split}
 \|\nabla f(\vz) \|_2^2& \leq \frac{16N}{K^2}\cdot f(\vz)\,.
\end{split}
\end{align}

We then try to find a lower bound on $\langle\vz-\vx,\, \nabla f(\vz)\rangle$. Using the Cauchy-Schwarz inequality, we also have that
\begin{align}
\label{eq:second_term}
\begin{split}
&\langle\vz-\vx,\, \nabla f(\vz)\rangle\\
&\geq 4f(\vz)-\sqrt{4f(\vz)}\sqrt{\frac{1}{MK^2}\sum_y\left((\vz-\vx)^\textrm{T}\mB_y\vx\right)^2}\,.
\end{split}
\end{align}
We proceed by further lower-bounding the above \eqref{eq:second_term}. Let $\vh=\vz-\vx$. For some $\theta\in(\frac{1}{2},1)$, we have
\begin{align}
\label{eq:second_term_lb1}
\begin{split}
&\theta^2\sum_y\left((\vz-\vx)^\textrm{T}\mB_y(\vz+\vx)\right)^2-\sum_y\left((\vz-\vx)^\textrm{T}\mB_y\vx\right)^2\\
&\geq \left(\theta\sqrt{\sum_y\left(\vh^\textrm{T}\mB_y\vh\right)^2}-(2\theta-1)\sqrt{\sum_y\left(\vh^T\mB_y\vx\right)^2}\right)\\
&\quad\quad\times\left(\theta\sqrt{\sum_y\left(\vh^\textrm{T}\mB_y\vh\right)^2}-(2\theta+1)\sqrt{\sum_y\left(\vh^T\mB_y\vx\right)^2}\right).
\end{split}
\end{align}
To make the left-hand side of \eqref{eq:second_term_lb1} greater than $0$, one of the conditions is that the following inequality should hold:
\begin{align}
\label{eq:second_term_lb1_two}
\sqrt{\sum_y\left(\vh^\textrm{T}\mB_y\vh\right)^2}&<(2-\frac{1}{\theta})\sqrt{\sum_y\left(\vh^T\mB_y\vx\right)^2}\,.
\end{align}

We can obtain an upper bound on $\|\vh\|_2$ to make \eqref{eq:second_term_lb1_two} hold. The left-hand side of \eqref{eq:second_term_lb1_two} can be upper bounded via:
\begin{align}
\label{eq:ub_hDh}
\begin{split}
\sum_y\left(\vh^\textrm{T}\mB_y\vh\right)^2\leq4\|\vh\|_2^2\cdot\|\vh\|_1^2\,.
\end{split}
\end{align}

The right-hand side of \eqref{eq:second_term_lb1_two} can be low-bounded as:
\begin{align}
\label{eq:lb_hdx}
\begin{split}
\sum_y\left(\vh^T\mB_y\vx\right)^2=\|\vh\|_1^2\cdot\overline{\vh}^\textrm{T}\mE\overline{\vh}\geq\|\vh\|_1^2\cdot\mu_\mE\,,
\end{split}
\end{align}
where $\overline{\vh}=\vh/\|\vh\|_1$, $\mE=\sum_{y=0}^{M-1}\mB_y\vx{\vx}^T\mB_y^T$ and $\mu_\mE>0$ can be computed using Lemma \ref{LEMMA:E_MIN}. Combining \eqref{eq:second_term_lb1_two}, \eqref{eq:ub_hDh} and \eqref{eq:lb_hdx}, we get that as long as \eqref{eq:h_bd} holds, \eqref{eq:second_term_lb1_two} will also hold.
\begin{align}
\label{eq:h_bd}
\|\vh\|_2<\tau=\left(2-\frac{1}{\theta}\right)\cdot\sqrt{\frac{\mu_\mE}{4}}\,.
\end{align}
Combining \eqref{eq:second_term_lb1},\eqref{eq:second_term_lb1_two}, we have
\begin{align}
\label{eq:lb_gradient_align}
-\sqrt{\frac{1}{MK^2}\sum_y\left((\vz-\vx)^\textrm{T}\mB_y\vx\right)^2}>-\theta\sqrt{4f(\vz)}\,.
\end{align}
Plug the above \eqref{eq:lb_gradient_align} into \eqref{eq:second_term}. We have:
\begin{align}
\label{eq:second_term_lb}
\langle\vz-\vx,\, \nabla f(\vz)\rangle > 4(1-\theta)f(\vz)\,.
\end{align}

We finally show that there exist some $\{\alpha,\beta\}$ to make the regularity condition $RC(\alpha,\beta,\tau)$ hold. Plugging \eqref{eq:second_term_lb}, \eqref{eq:first_term_lb} into \eqref{eq:rc_condition}, we need the following inequality to hold:
\begin{align}
\label{eq:rc_condtion_c1}
    4(1-\theta)f(\vz)\geq\frac{1}{\alpha}\|\vh\|_2^2+\frac{1}{\beta}\frac{16N}{K^2}f(\vz)\,.
\end{align}
Combining \eqref{eq:lb_gradient_align} and \eqref{eq:lb_hdx}, we further have
\begin{align}
\label{eq:lb_obj_fun}
f(\vz)>\frac{1}{4\theta^2}\frac{1}{MK^2}\sum_y\left(\vh^\textrm{T}\mB_y\vx\right)^2\geq\frac{1}{4\theta^2}\frac{1}{MK^2}\|\vh\|_2^2\mu_\mE\,.
\end{align}
Plugging \eqref{eq:lb_obj_fun} into \eqref{eq:rc_condtion_c1}, we need the following to hold:
\begin{align}
\label{eq:rc_condition_c2}
    \left((1-\theta)-\frac{1}{\beta}\frac{4N}{K^2}\right)\frac{1}{\theta^2}\frac{1}{MK^2}\mu_\mE\geq\frac{1}{\alpha}\,.
\end{align}
The constants $\theta$ and $\{\alpha,\beta\}$ that satisfy \eqref{eq:rc_condition_c2} can be chosen in the following order:
\begin{enumerate}
    \item Choose some $\theta\in(\frac{1}{2}, 1)$.
    \item Fix $\theta$, choose some $\beta>\frac{4N}{(1-\theta)K^2}$.
    \item Fix $\theta,\beta$, choose some $\alpha\geq\left((1-\theta)-\frac{1}{\beta}\frac{4N}{K^2}\right)^{-1}\frac{\theta^2MK^2}{\mu_\mE}$.
\end{enumerate}
We can get that the regularity condition $RC(\alpha,\beta,\tau)$ holds for some $\{\alpha>0,\beta>0\}$ and $\tau=(2-\frac{1}{\theta})\sqrt{\frac{\mu_\mE}{4}}$.

\emph{\bfseries Step 2)} We use $\overline{\vz}_{+1}=\vz-\eta\nabla f(\vz)$ to denote one gradient descent update, and $\vz_{+1}=\mathscr{P}_\mathcal{S}(\overline{\vz}_{+1})\in\mathcal{S}$ to denote one projected gradient descent update. As detailed in the Supplementary Material, we can prove that $\|\vz_{+1}-\vx\|_2^2\leq\|\overline{\vz}_{+1}-\vx\|_2^2$.

We can use the regularity condition $RC(\alpha,\beta,\tau)$ and get that $\|\vz_{t+k}-\vx\|_2^2<(1-\frac{2\eta}{\alpha})^k\cdot\|\vz_t-\vx\|_2^2$.
\end{proof}

\subsection{Monte Carlo Simulations}
\label{app:monte_carlo}
In order to find the distributions of $\sum_{y=0}^{M-1}\big(\widehat{\vh}^T\mB_y\widehat{\vh}\big)^2$ and $\widehat{\vh}^T\mE\widehat{\vh}$, we need to sample uniformly with respect to $\widehat{\vh}\in\widehat{\mathcal{H}}$:
\begin{align}
\label{eq:H_hat_set_def}
\widehat{\mathcal{H}}=\left\{\widehat{\vh}\,\left|\,\widehat{\vh}=\frac{\vh}{\|\vh\|_2}=\frac{\vz-\vx}{\|\vz-\vx\|_2},\quad  \vz\in\mathcal{S}\,,\vz\neq\vx\right.\right\}\,.
\end{align}
The nonconvex set $\widehat{\mathcal{H}}$ is a constrained region on the unit sphere. We can verify that $\widehat{\mathcal{H}}$ is the same as the new set $\mathcal{F}$ defined by the following constraints:
\begin{align}
    \label{eq:F_set_1}
    &\sum_{i=1}^M\hat{h}_i=0\\
    \label{eq:F_set_2}
    &\hat{h}_i\geq 0\quad\textnormal{if }x_i=0\\
    \label{eq:F_set_3}
    &\hat{h}_i\leq 0\quad\textnormal{if }x_i=1\\
    \label{eq:F_set_4}
    &\|\widehat{\vh}\|_2=1\,.
\end{align}
\begin{itemize}
\item If $\widehat{\vh}\in\widehat{\mathcal{H}}$, it is easy to see that $\widehat{\vh}\in\mathcal{F}$.
\item Since $\|\widehat{\vh}\|_2=1$, we can get that $|\hat{h}_i|<1$. If $\widehat{\vh}\in\mathcal{F}$, we can construct such a $\widetilde{\vz}=\vx+\widehat{\vh}$. It is easy to verify that $\widetilde{\vz}\in\mathcal{S}$ and $\|\widetilde{\vz}-\vx\|_2=\|\widehat{\vh}\|_2=1$. Hence $\widehat{\vh}=\frac{\widetilde{\vz}-\vx}{\|\widetilde{\vz}-\vx\|_2}\in\widehat{\mathcal{H}}$. 
\end{itemize}

Since directly sampling from the nonconvex set $\mathcal{F}$ is difficult, we do it indirectly. Let $\mathcal{J}$ denote the convex set defined by the constraints \eqref{eq:F_set_1}-\eqref{eq:F_set_3}. We first perform Gibbs sampling \cite{Geman:Gibbs:1984} from $\mathcal{J}$ according to a constrained standard multivariate Gaussian \cite{Schmidt:2009,Burkardt:Trunc:2014}, and then project the samples onto the unit sphere defined by \eqref{eq:F_set_4}. This way we can generate samples from $\mathcal{F}$ uniformly. Based on these samples, we can finally estimate the empirical upper and lower bounds $\nu_1>0,\nu_2>0$ in \eqref{eq:upper_bd_hbh} and \eqref{eq:lower_bd_heh}.

\end{appendices}

\bibliographystyle{IEEEbib}
\bibliography{refs}

\newpage
\onecolumn

\setcounter{section}{0}
\begin{center}
{\LARGE Supplementary Material for \\``Reconstructing Point Sets from Distance Distributions''}
\end{center}

\vspace{5em}

The supplementary material contains
\begin{itemize}
    \item Figure illustration and pseudo code for the agglomerative clustering algorithm introduced in Section II.B to extract point locations from the recovered solution $\vz$.
    \item Detailed step by step derivations of the proofs presented in the Appendix.
    \item Complementary analysis on convergence and difficulty of recovery for the beltway problem.
\end{itemize}

\newpage

\setcounter{lemma}{0}
\setcounter{theorem}{0}

\section{Agglomerative Clustering Algorithm and Illustration}
\label{supp:sec:intro}

\begin{figure*}[tbp]
\centering
\subfigure{
\label{supp:fig:z_vec}
\includegraphics[width=\textwidth]{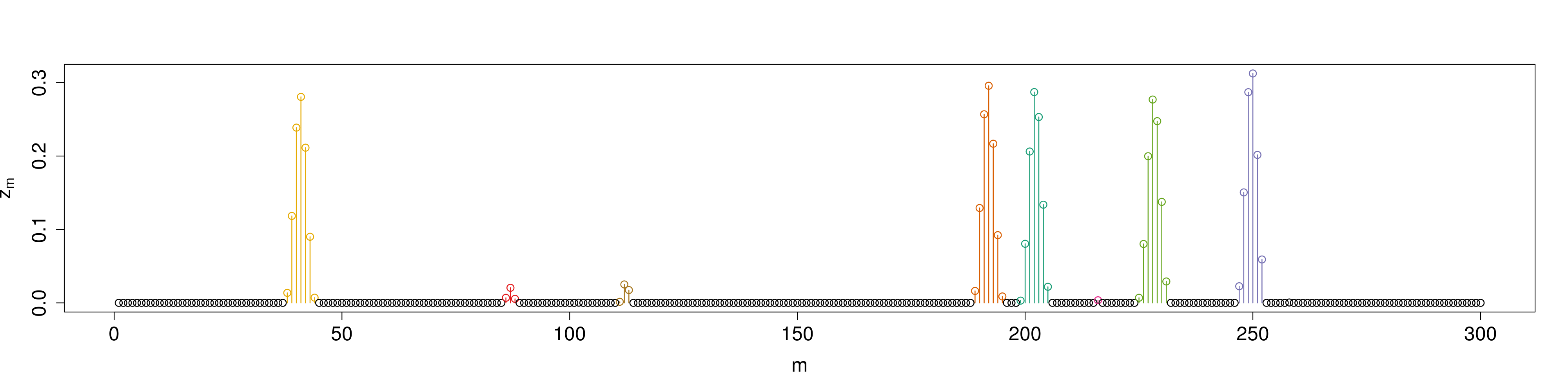}}\\
\subfigure{
\label{supp:fig:h_clustering}
\includegraphics[width=\textwidth]{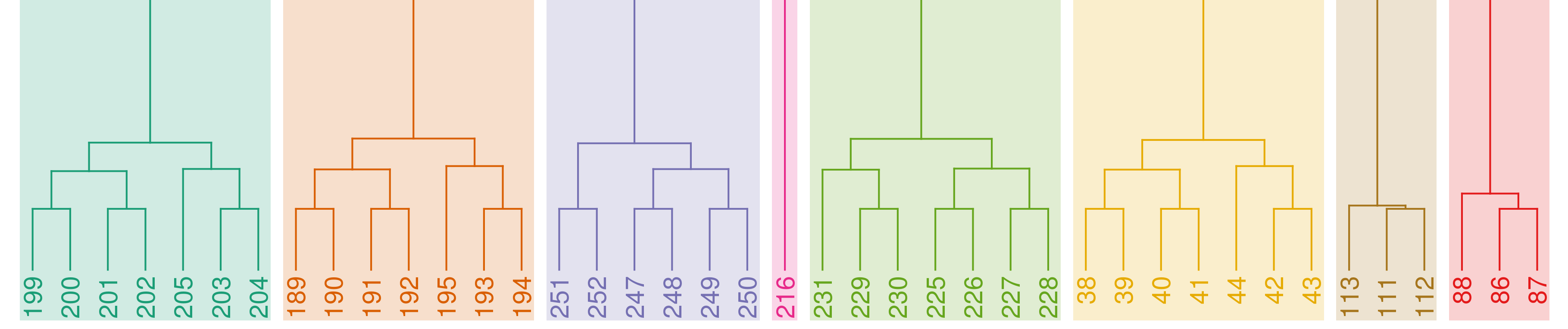}}
\caption{Illustration of agglomerative clustering for $N=5$. The agglomerative clustering produces $8$ clusters, only the centroids of the $5$ clusters with the highest weights are taken as the point locations.}
\label{supp:fig:hierarchical_clustering}
\end{figure*}

In the noisy case we have $\vx\in[0,1]^M$. The $m$-th entry $z_m$ of $\vz$ is the \emph{estimated} probability that a point is located at the $m$-th segment $l_m$. 
Extracting $N$ point locations from $\vz$ can be posed as a clustering problem. As illustrated in Fig. \ref{supp:fig:hierarchical_clustering}, each $l_m$ is viewed as a cluster with the weight $z_m$. We can cluster the $M$ segments using the agglomerative clustering approach \cite{Rokach2005} summarized in Algorithm \ref{supp:alg:agg}. The centroids of the $N$ clusters with the largest weights are taken as the estimated point locations.

\begin{algorithm}[htbp]
\caption{Extracting the point locations via agglomerative clustering }
\label{supp:alg:agg}
\begin{algorithmic}[1]
\REQUIRE The solution $\vz$, the smallest distance between two different points $d_{\min}$.
\STATE Treat each segment $l_m$ with a nonzero weight $\omega_m=z_m$ as one cluster $C_m=\set{l_m}$
\STATE Compute the centroid $c_m$ of every cluster $C_m\in\mathcal{C}=\set{C_1,C_2,\cdots}$
\WHILE{$|\mathcal{C}|>N$}
\STATE Merge the two closest clusters\footnotemark $\set{C_i,C_j}$ with weights $\set{w_i<1,w_j<1}$ and centroids $\|c_i-c_j\|<d_{\min}$ into one cluster $C_i$
\STATE Update the weight $w_i$ and the centroid $c_i$ of the new cluster $C_i$
\IF{the clusters cannot be merged further}
\STATE \textbf{break}
\ENDIF
\ENDWHILE
\STATE {\bfseries Return} the set of centroids $\set{c_1,c_2,\cdots}$
\end{algorithmic}
\end{algorithm}
\footnotetext{Randomly pick a pair of clusters in case of a draw.}

\section{Detailed Proofs of the Lemmas and the Theorem}
We recover the solution $\vz$ by solving the following constrained nonconvex optimization problem:
\begin{align}
\label{supp:eq:constrained_nonconvex}
\begin{split}
\min_{\vz}\quad &f(\vz)=\frac{1}{M}\sum_{y=0}^{M-1}\big(q_{\vz}(y)-p(y)\big)^2
\end{split}\\
\label{supp:eq:relaxed_constraint}
\textnormal{subject to}\quad &0\leq z_m\leq 1,\ \forall\ m\in\set{1,\cdots,M}\\
\label{supp:eq:l1_constraint}
&\sum_{m=1}^Mz_m = N\,.
\end{align}
Let $\mathcal{S}$ denote the convex set defined by the constraints \eqref{supp:eq:relaxed_constraint},\eqref{supp:eq:l1_constraint}. Given a proper initialization $\vz_0$, we propose to solve \eqref{supp:eq:constrained_nonconvex} via the projected gradient descent method:
\begin{align}
\label{supp:eq:pgd_update}
\vz_{t+1} = \mathscr{P}_\mathcal{S}\big(\vz_t-\eta\cdot\nabla f(\vz_t)\big)\,,
\end{align}
where $\eta>0$ is the step size, $\mathscr{P}_\mathcal{S}(\cdot)$ is the projection of the gradient descent update onto $\mathcal{S}$. Let $\overline{\vz}=\vz-\eta\cdot\nabla f(\vz)$ denote the gradient descent update. The projection is the solution to the following convex problem
\begin{align}
\label{supp:eq:projection_box_constraints}
\begin{split}
\min_\vs\quad&\frac{1}{2}\|\vs-\overline{\vz}\|_2^2\\
\textnormal{subject to}\quad&0\leq s_m\leq 1,\ \forall\ m\in\{1,\cdots,M\}\\
&\sum_{m=1}^Ms_m=N\,.
\end{split}
\end{align}

\begin{lemma}
\label{supp:lemma:lower_bound}[Lemma 2, \cite{projection06}]
Let $\vs$ be the optimal solution to the minimization problem in (\ref{supp:eq:projection_box_constraints}). Let $i$ and $j$ be two indices such that $\overline{z}_i>\overline{z}_j$. If $s_i=0$ then $s_j$ must be $0$ as well.
\end{lemma}

\subsection{Proof of Lemma \ref{supp:lemma:upper_bound}}

\begin{lemma}
\label{supp:lemma:upper_bound}
Let $\vs$ be the optimal solution to the minimization problem in (\ref{supp:eq:projection_box_constraints}). Let $i$ and $j$ be two indices such that $\overline{z}_i>\overline{z}_j$. If $s_j=1$ then $s_i$ must be $1$ as well.
\end{lemma}

\begin{proof}
Suppose that $s_i<1$. We construct a vector $\widetilde{\vs}\in\mathbb{R}^M$ out of $\vs$ by swapping the positions of $s_i$ and $s_j$ in $\vs$, i.e. $\widetilde{s}_i=s_j$ and $\widetilde{s}_j=s_i$. $\widetilde{\vs}$ also satisfies the constraints in (\ref{supp:eq:projection_box_constraints}). Since $\overline{z}_i>\overline{z}_j$ and $s_j=1$, we then have:
\begin{align}
\begin{split}
\|\vs-\overline{\vz}\|_2^2-\|\tilde{\vs}-\overline{\vz}\|_2^2&=(s_i-\overline{z}_i)^2+(s_j-\overline{z}_j)^2-(s_j-\overline{z}_i)^2-(s_i-\overline{z}_j)^2\\
&=2(1-s_i)(\overline{z}_i-\overline{z}_j)\\
&> 0\,.
\end{split}
\end{align}
This is in contradiction with the fact that $\vs$ is the minimizer of (\ref{supp:eq:projection_box_constraints}). Hence $s_i$ must be $1$.
\end{proof}

\subsection{Proof of Lemma \ref{supp:LEMMA:UNIQUE_R}}
\label{supp:proof:lemma:unique_r}

Since reordering of the entries of $\overline{\vz}$ does not change the value of (\ref{supp:eq:projection_box_constraints}), and adding some constant to $\overline{\vz}$ does not change the solution of (\ref{supp:eq:projection_box_constraints}), without loss of generality we can assume that the entries of $\overline{\vz}$ are all positive in a non-increasing order: $\overline{z}_1\geq \overline{z}_2\geq\cdots\geq \overline{z}_M\geq N$. Lemma \ref{supp:lemma:lower_bound} and \ref{supp:lemma:upper_bound} imply that for the optimal solution $\vs$:
\begin{itemize}
\item The entries of $\vs$ are in a non-increasing order.
\item The first $\rho$ entries of $\vs$ satisfy $0<s_m\leq 1$; the rest of the entries are $0$s.
\end{itemize}
Since $\exists\ s_m\in(0,1)$, we have $\rho>N$ and that at most $N-1$ entries of $\vs$ could equal $1$. Suppose the first $r-1$ entries of $\vs$ are all $1$s, the following must hold for $1\leq r\leq N<\rho$
\begin{align}
\label{supp:eq:cst_r}
&0<\overline{z}_r-\kappa< 1\\
\label{supp:eq:cst_rm1}
&1\leq \overline{z}_{r-1}-\kappa,\quad\textnormal{if }2\leq r\leq N<\rho\,.
\end{align}

\begin{lemma}
\label{supp:LEMMA:UNIQUE_R}
If the solution $\vs$ has at least one entry $s_m\in(0,1)$, there is one and only one $r\in\set{1,\ldots,N}$ that produces the $(\rho,\kappa)$ satisfying \eqref{supp:eq:cst_r} and \eqref{supp:eq:cst_rm1}.
\end{lemma}

\begin{proof}
Let $\mathcal{S}$ denote the convex set defined by the constraints $0\leq s_m\leq 1$, $\forall\ 1\leq m\leq M$ and $\|\vs\|_1=N$. Note that the entries of $\overline{\vz}$ and $\vx$ are in a non-increasing order. We will proceed in the following two steps:
\begin{enumerate}[label={\arabic*)}]
\item Since $\mathcal{S}$ is non-empty, the projection onto it exists, i.e. there is one $r\in\set{1,\cdots,N}$ that produces the $(\rho,\kappa)$ that satisfy $1\leq \overline{z}_{r-1}-\kappa$ if $2\leq r\leq N<\rho$ and $0<\overline{z}_r-\kappa<1$. In fact, since $\mathcal{S}$ is a closed convex set, the projection is also unique.
\item Without loss of generality, suppose that there are two different $r_1<r_2\in\set{1,\cdots,N}$ that produce the two pairs $(\rho_1,\kappa_1|r_1)$ and $(\rho_2,\kappa_2|r_2)$ that satisfy the constraints (\ref{supp:eq:cst_r}) and (\ref{supp:eq:cst_rm1}). We have:
\begin{gather*}
r_1<r_2 \,\Rightarrow\, r_1\leq r_2-1 \,\Rightarrow\, \overline{z}_{r_2-1}-\kappa_1\leq \overline{z}_{r_1}-\kappa_1< 1 \,\Rightarrow\, \overline{z}_{r_2-1}-1< \kappa_1\\
1< \overline{z}_{r_2-1}-\kappa_2 \,\Rightarrow\, \kappa_2< \overline{z}_{r_2-1}-1\,.
\end{gather*}
Hence $\kappa_2<\kappa_1$. We further have:
\begin{gather*}
\overline{z}_{\rho_1}-\kappa_1>0 \,\Rightarrow\, \overline{z}_{\rho_1}>\kappa_1\\
\overline{z}_{\rho_2+1}-\kappa_2\leq 0 \,\Rightarrow\, \overline{z}_{\rho_2+1}\leq\kappa_2\\
\overline{z}_{\rho_2+1}\leq\kappa_2<\kappa_1< \overline{z}_{\rho_1} \,\Rightarrow\, \overline{z}_{\rho_2+1} < \overline{z}_{\rho_1}\,.
\end{gather*}
Hence $\rho_2+1>\rho_1\,\Rightarrow\,\rho_2\geq\rho_1$.
\begin{enumerate}
\item If $r_2\leq\rho_1$, we can find the upper bound for the sum of the first $\rho_1$ entries of $\vs_1$:
\begin{align}
\label{supp:eq:r_rho_neq_1}
\begin{split}
r_1-1+\sum_{m=r_1}^{\rho_1}(\overline{z}_m-\kappa_1)=\,&r_1-1+\sum_{m=r_1}^{r_2-1}(\overline{z}_m-\kappa_1)+\sum_{m=r_2}^{\rho_1}(\overline{z}_m-\kappa_1)\\
<\,&r_1-1+\sum_{m=r_1}^{r_2-1}1+\sum_{m=r_2}^{\rho_1}(\overline{z}_m-\kappa_1)\\
<\,&r_2-1+\sum_{m=r_2}^{\rho_1}(\overline{z}_m-\kappa_2)\\
\leq\,&r_2-1+\sum_{m=r_2}^{\rho_2}(\overline{z}_m-\kappa_2)\,.
\end{split}
\end{align}

\item If $r_2>\rho_1$, we can compute:
\begin{align}
\label{supp:eq:r_rho_neq_2}
\begin{split}
r_1-1+\sum_{m=r_1}^{\rho_1}(\overline{z}_m-\kappa_1)\,\leq\, &r_1-1+\sum_{m=r_1}^{\rho_1}1\\
=\,&\rho_1\\
\leq\,&r_2-1\\
<\,&r_2-1+\sum_{m=r_2}^{\rho_2}(\overline{z}_m-\kappa_2)\,.
\end{split}
\end{align}
\end{enumerate}

Let $\vs_1,\vs_2$ denote the solutions of \eqref{supp:eq:projection_box_constraints} produced by $r_1,r_2$ respectively. Both (\ref{supp:eq:r_rho_neq_1}) and (\ref{supp:eq:r_rho_neq_2}) show that $\sum_{m=1}^Ms_{1m}<\sum_{m=1}^Ms_{2m}$. This is in contradiction with the assumption that $\sum_{m=1}^Ms_{1m}=\sum_{m=1}^Ms_{2m}=N$. Hence $r_1=r_2$, there is only one $r\in\set{1,\cdots,N}$ that produces the $(\rho,\kappa)$ that satisfy the constraints (\ref{supp:eq:cst_r}) and (\ref{supp:eq:cst_rm1}).

\end{enumerate}
\end{proof}

\subsection{Proof of Lemma \ref{supp:LEMMA:E_MIN}}
\label{supp:proof:lemma:e_min}

\begin{lemma}
\label{supp:LEMMA:E_MIN}
Let $\mB_y=\mA_y+\mA_y^T$, $\mE=\sum_{y=0}^{M-1}\mB_y\vx{\vx}^T\mB_y^T$, $\mathcal{S}$ be the convex set defined by the constraints \eqref{supp:eq:relaxed_constraint},\eqref{supp:eq:l1_constraint}. The following problem is convex $\forall\ \vz\in\mathcal{S}$, $\vz\neq\vx$:
\begin{align}
\label{supp:eq:lambda_min_val}
\begin{split}
\mu_\mE &= \min_{\vz\in\mathcal{S}, \vz\neq\vx}\,\frac{1}{\|\vz-\vx\|_1^2}(\vz-\vx)^T\mE(\vz-\vx)\\
&=\min_{\overline{\vh}\in\mathcal{G}}\,\overline{\vh}^T\mE\overline{\vh}\,,
\end{split}
\end{align}
where $\mu_\mE>0$ and $\overline{\vh}=\frac{\vz-\vx}{\|\vz-\vx\|_1}$, $\mathcal{G}$ is a convex set defined by the following constraints:
\begin{align}
\label{supp:eq:hh_con_2}
&\sum_{i=1}^M\overline{h}_i=0\\
\label{supp:eq:hh_con_3}
&\overline{h}_i\in[0,\,0.5]\quad\textnormal{if $x_i=0$}\\
\label{supp:eq:hh_con_4}
&\overline{h}_i\in[-0.5,\,0]\quad\textnormal{if $x_i=1$}\\
\label{supp:eq:hh_con_1}
&\|\overline{\vh}\|_1=\vr^T\overline{\vh}=1\,,
\end{align}
where $\vr\in\{-1,1\}^M$ depends on $\vx$ and is defined as follows:
\begin{align}
r_i=\left\{
\begin{array}{l}
1\\
-1
\end{array}
\quad
\begin{array}{l}
\textnormal{if $x_i=0$}\\
\textnormal{if $x_i=1$}\,.
\end{array}
\right.
\end{align}
\end{lemma}

\begin{proof}
Since $\mE=\sum_{y=0}^{M-1}\mB_y\vx{\vx}^\textrm{T}\mB_y^\textrm{T}$, we can see that $\overline{\vh}^T\mE\overline{\vh}=\sum_y\left(\overline{\vh}^T\mB_y\vx\right)^2\geq 0$, $\forall\ \overline{\vh}\in\mathbb{R}^M$. Hence $\mE$ is positive-semidefinite. We define the following set $\overline{\mathcal{H}}$:
\begin{align}
\label{supp:eq:H_bar_set_def}
\overline{\mathcal{H}}=\left\{\overline{\vh}\,\left|\,\overline{\vh}=\frac{1}{\|\vz-\vx\|_1}(\vz-\vx),\quad\forall \vz\in\mathcal{S}\,,\vz\neq\vx\right.\right\}\,.
\end{align}
where $\mathcal{S}$ is the convex set defined by \eqref{supp:eq:relaxed_constraint} and \eqref{supp:eq:l1_constraint}. We then have
\begin{align}
\begin{split}
    \mu_\mE &= \min_{\vz\in\mathcal{S}, \vz\neq\vx}\,\frac{1}{\|\vz-\vx\|_1^2}(\vz-\vx)^T\mE(\vz-\vx)\\
    & = \min_{\overline{h}\in\overline{\mathcal{H}}}\,\overline{\vh}^T\mE\overline{\vh}\,.
\end{split}
\end{align}

\begin{enumerate}[label={\arabic*)}]
\item We first prove that $\overline{\mathcal{H}}$ is a convex set. Let $\overline{\vh}^{(1)},\overline{\vh}^{(2)}\in\overline{\mathcal{H}}$. We have $\textnormal{sign}(\overline{h}_i^{(1)})=\textnormal{sign}(\overline{h}_i^{(2)})$ if $\overline{h}_i^{(1)}\neq0,\ \overline{h}_i^{(2)}\neq0$. Let $\overline{\vh}^{(3)}=(1-\rho)\overline{\vh}^{(1)}+\rho\overline{\vh}^{(2)}$, where $\rho\in(0,1)$. We have:
\begin{align}
\label{supp:eq:l1_norm_z3}
\begin{split}
\|\overline{\vh}^{(3)}\|_1 & = \left\|(1-\rho)\overline{\vh}^{(1)}+\rho\overline{\vh}^{(2)}\right\|_1\\
&=\sum_{i=1}^M\left|(1-\rho)\overline{h}_i^{(1)}+\rho \overline{h}_i^{(2)}\right|\\
&=\sum_{i=1}^M\left|(1-\rho)\overline{h}_i^{(1)}\right|+\left|\rho \overline{h}_i^{(2)}\right|\\
&=(1-\rho)\left\|\overline{\vh}^{(1)}\right\|_1+\rho\left\|\overline{\vh}^{(2)}\right\|_1=1\,.
\end{split}
\end{align}
Let $\iota_1=\frac{1-\rho}{\left\|\vz^{(1)}-\vx\right\|_1}$, $\iota_2=\frac{\rho}{\left\|\vz^{(2)}-\vx\right\|_1}$. We have
\begin{align}
\begin{split}
\overline{\vh}^{(3)}&=(1-\rho)\overline{\vh}^{(1)}+\rho \overline{\vh}^{(2)}\\
&=\frac{1-\rho}{\left\|\vz^{(1)}-\vx\right\|_1}\left(\vz^{(1)}-\vx\right)+\frac{\rho}{\left\|\vz^{(2)}-\vx\right\|_1}\left(\vz^{(2)}-\vx\right)\\
&=\left(\iota_1+\iota_2\right)\left(\frac{\iota_1}{\iota_1+\iota_2}\vz^{(1)}+\frac{\iota_2}{\iota_1+\iota_2}\vz^{(2)}-\vx\right)\\
&=(\iota_1+\iota_2)(\vz^{(3)}-\vx)\,.
\end{split}
\end{align}
Using \eqref{supp:eq:l1_norm_z3}, we can see that $\iota_1+\iota_2=\frac{1}{\|\vz^{(3)}-\vx\|_1}$. Since $\vz^{(1)},\vz^{(2)}\in\mathcal{S}$, we have $\vz^{(3)}\in\mathcal{S}$. We have shown that $\overline{\vh}^{(3)}$ can be written in the same form given in \eqref{supp:eq:H_bar_set_def} and thus belongs to $\overline{\mathcal{H}}$.
\begin{align}
    \overline{\vh}^{(3)}=\frac{1}{\|\vz^{(3)}-\vx\|_1}(\vz^{(3)}-\vx)\,.
\end{align}
Hence $\overline{\vh}^{(3)}\in\overline{\mathcal{H}}$, and $\overline{\mathcal{H}}\subset\mathbb{R}^M$ is a convex set. Minimizing $\overline{\vh}^\textrm{T}\mE\overline{\vh}$ with respect to $\overline{\vh}\in\overline{\mathcal{H}}$ is a convex problem.

\item We next prove that $\mu_\mE$ in \eqref{supp:eq:lambda_min_val} is strictly positive. If $\overline{\vh}^T\mE\overline{\vh}=\sum_y\left(\overline{\vh}^T\mB_y\vx\right)^2=0$, we have $(\vz-\vx)^T\mB_y\vx=0$, $\forall\ y\in\set{0,\cdots,M-1}$. When $y=0$, $\mB_0=2\mI$, where $\mI$ is the identity matrix, we get $\vz^T\vx={\vx}^\textrm{T}\vx=N$. Since $\vz\in\mathcal{S}$ and $\vx\in\set{0,1}^M$ in the noiseless case, we have $\vz=\vx$. This is in contradiction with the assumption $\vz\neq\vx$, hence $\overline{\vh}^T\mE\overline{\vh}>0$, $\forall\ \overline{\vh}\in\overline{\mathcal{H}}$.

\item We finally prove that $\overline{\mathcal{H}}$ and a new set $\mathcal{G}$ defined by the following constraints are the same:
\begin{align}
\label{supp:eq:hh_con_2}
&\sum_{i=1}^M\overline{h}_i=0\\
\label{supp:eq:hh_con_3}
&\overline{h}_i\in[0,\,0.5]\quad\textnormal{if $x_i=0$}\\
\label{supp:eq:hh_con_4}
&\overline{h}_i\in[-0.5,\,0]\quad\textnormal{if $x_i=1$}\\
\label{supp:eq:hh_con_1}
&\|\overline{\vh}\|_1=\vr^T\overline{\vh}=1\,,
\end{align}
where $\vr\in\{-1,1\}^M$ is defined as follows:
\begin{align}
r_i=\left\{
\begin{array}{l}
1\\
-1
\end{array}
\quad
\begin{array}{l}
\textnormal{if $x_i=0$}\\
\textnormal{if $x_i=1$}\,.
\end{array}
\right.
\end{align}
\begin{itemize}
\item It is easy to verify that if $\overline{\vh}\in\overline{\mathcal{H}}$, \eqref{supp:eq:hh_con_2} and \eqref{supp:eq:hh_con_1} hold. Since $z_i\in[0,1]$ and $x_i\in\{0,1\}$, if $x_i=0$, $\overline{h}_i\geq 0$; if $x_i=1$, $\overline{h}_i\leq 0$. On the other hand, if $|\overline{h}_i|>0.5$, from \eqref{supp:eq:hh_con_2} we have $\sum_{j\neq i}|\overline{h}_j|\geq|\sum_{j\neq i}\overline{h}_j|=|-\overline{h}_i|>0.5$. This means that $\|\overline{\vh}\|_1=|\overline{h}_i|+\sum_{j\neq i}|\overline{h}_j|>1$, which contradicts \eqref{supp:eq:hh_con_1}. Hence $|\overline{h}_i|\leq 0.5$, \eqref{supp:eq:hh_con_3} and \eqref{supp:eq:hh_con_4} hold. This proves that $\overline{\vh}\in\mathcal{G}$. 
\item If $\overline{\vh}\in\mathcal{G}$, we can construct such a $\widetilde{\vz}=\vx+\overline{\vh}$. It is easy to verify that $\widetilde{\vz}\in\mathcal{S}$ and $\|\widetilde{\vz}-\vx\|_1=\|\overline{\vh}\|_1=1$. Hence $\overline{\vh} = \frac{1}{\|\widetilde{\vz}-\vx\|_1}(\widetilde{\vz}-\vx)\in\overline{\mathcal{H}}$.
\end{itemize}
\end{enumerate}
Computing $\mu_\mE=\min_{\overline{\vh}\in\mathcal{G}}\overline{\vh}^T\mE\overline{\vh}\,>0$ is thus a convex problem, and can be efficiently solved via quadratic programming.
\end{proof}

\subsection{Proof of Theorem ~\ref{supp:THM:CVG_SED}}
\label{supp:proof:thm:cvg_sed}

When the distance between the solution $\vz_t$ and a global optimum $\vx$ is less than some $\tau>0$, i.e. $\|\vz_t-\vx\|_2<\tau$, we would like to show that the projected gradient descent update in (15) converges linearly to a global optimizer $\vx$. The convergence neighbourhood $\mathcal{E}(\tau)$ is characterized by the regularity condition $RC(\alpha,\beta,\tau)$ of the objective function $f(\vz)$ \cite{WF:2015}: For all $\vz\in\mathcal{E}(\tau)$,
\begin{align}
\label{supp:eq:rc_condition}
    \langle\nabla f(\vz), \vz-\vx\rangle\geq\frac{1}{\alpha}\|\vz-\vx\|_2^2+\frac{1}{\beta}\|\nabla f(\vz)\|_2^2\,,
\end{align}
where $\alpha>0,\beta>0$ are some chosen constants. Let $\overline{\vz}_{t+k}$ denote the gradient descent update. The $RC(\alpha,\beta,\tau)$ in \eqref{supp:eq:rc_condition} ensures $\overline{\vz}_{t+k}$ with a step size $\eta\in(0,\frac{2}{\beta}]$ converges linearly to $\vx$ once $\vz_t$ reaches $\mathcal{E}(\tau)$ \cite[Lemma 7.10]{WF:2015}:
\begin{align}
\label{supp:eq:gd_closer}
    \left\|\overline{\vz}_{t+k}-\vx \right\|^2_2\leq(1-\frac{2\eta}{\alpha})^k\cdot\|\vz_t-\vx\|^2_2.
\end{align}
We shall further extend the above \eqref{supp:eq:gd_closer} to the projected gradient descent update $\vz_{t+k}$. For the turnpike problem, we make use of the following theorem:
\begin{theorem}
\label{supp:THM:CVG_SED}
In the noiseless case, let $\vh=\vz-\vx$ and $\mB_y=\mA_y+\mA_y^T$. If $\vz$ satisfies
\begin{align}
\label{supp:eq:absolute_radius}
\|\vh\|_2=\|\vz-\vx\|_2<\tau=(2-\frac{1}{\theta})\cdot\sqrt{\frac{\mu_\mE}{4}}\,,
\end{align}
where $\theta\in\big(\frac{1}{2},1\big)$ is some fixed constant and $\mu_\mE>0$ depends on the matrix $\mE=\sum_{y=0}^{M-1}\mB_y\vx\vx^T\mB_y^T$, 
\begin{enumerate}
\item There exists a choice of $\{\alpha>0,\beta>0\}$ such that the regularity condition $RC(\alpha,\beta,\tau)$ holds.
\item Under this choice of parameters $\{\alpha,\beta,\tau\}$, if $\|\vh_t\|_2=\|\vz_t-\vx\|_2<\tau$ and the step size $\eta\in(0,\frac{2}{\beta}]$, the projected gradient descent update in \eqref{supp:eq:pgd_update} converges linearly to $\vx$:
\begin{equation}
\label{supp:eq:converge}
\|\vz_{t+k}-\vx\|_2^2<(1-\frac{2\eta}{\alpha})^k\cdot\|\vz_t-\vx\|_2^2\,.
\end{equation}
\end{enumerate}
\end{theorem}

\begin{proof}
Let $\mB_y=\mA_y+\mA_y^T$. The objective function $f(\vz)$ in \eqref{supp:eq:constrained_nonconvex} can be written as
\begin{align}
    f(\vz) = \frac{1}{4MK^2}\sum_{y=0}^{M-1}\left(\vz^T\mB_y\vz-\vx^T\mB_y\vx\right)^2\,.
\end{align}
The gradient $\nabla f(\vz)$ is 
\begin{align}
\begin{split}
    \nabla f(\vz) &= \frac{1}{MK^2}\sum_{y=0}^{M-1}\mB_y\vz\cdot\left(\vz^T\mB_y\vz-\vx^T\mB_y\vx\right)\\
    &= \frac{1}{MK^2}\sum_{y=0}^{M-1}\mB_y\vz\cdot(\vz-\vx)^T\mB_y(\vz+\vx)\,.
\end{split}
\end{align}
In the following we first establish the regularity condition $RC(\alpha,\beta,\tau)$ for the gradient descent update, and then use it to prove the linear convergence of the projected gradient descent update.

\paragraph{Step 1:} 
Our goal is then to find the radius $\tau$ so that $RC(\alpha,\beta,\tau)$ holds.
\begin{itemize}
\item We first try to find an upper bound on $\|\nabla f(\vz)\|_2^2$:
\begin{align}
\label{supp:eq:first_term_lb}
\begin{split}
\|\nabla f(\vz) \|_2^2& = \frac{1}{K^4}\left\|\frac{1}{M}\sum_y\mB_y\vz\cdot(\vz-\vx)^\textrm{T}\mB_y(\vz+\vx)\right\|_2^2\\
&\leq \frac{1}{MK^4}\sum_y\left\|\mB_y\vz\cdot(\vz-\vx)^\textrm{T}\mB_y(\vz+\vx)\right\|_2^2\\
&= \frac{1}{MK^4}\sum_y\left\|\mB_y\vz\right\|_2^2\cdot\left((\vz-\vx)^\textrm{T}\mB_y(\vz+\vx)\right)^2\\
&\leq \frac{1}{MK^4}\sum_y\sigma_{\max}^2\left(\mB_y\right)\|\vz\|_2^2\cdot\left((\vz-\vx)^\textrm{T}\mB_y(\vz+\vx)\right)^2\\
&\leq \frac{4}{MK^4}\|\vz\|_2^2\sum_y\left((\vz-\vx)^\textrm{T}\mB_y(\vz+\vx)\right)^2\\
&= \frac{16}{K^2}\|\vz\|_2^2\cdot f(\vz)\\
&\leq \frac{16N}{K^2}f(\vz)\,,
\end{split}
\end{align}
where $\sigma_{\max}^2\left(\mB_y\right)\leq 4$, $\forall\ y=\{0,1,\cdots,M-1\}$ according to the Schur's bound \cite{Schur1911}, and $\|\vz\|_2^2\leq\sum_{m=1}^Mz_m =N$ is due to the constraints \eqref{supp:eq:relaxed_constraint},\eqref{supp:eq:l1_constraint}.

\item We then try to find a lower bound on $\langle\vz-\vx,\, \nabla f(\vz)\rangle$. Using the Cauchy-Schwarz inequality, we also have that
\begin{align}
\label{supp:eq:second_term}
\begin{split}
\langle\vz-\vx,\, \nabla f(\vz)\rangle&=\frac{1}{MK^2}\sum_y(\vz-\vx)^\textrm{T}\mB_y\vz\cdot(\vz-\vx)^\textrm{T}\mB_y(\vz+\vx)\\
&=4f(\vz)-\frac{1}{MK^2}\sum_y(\vz-\vx)^\textrm{T}\mB_y(\vz+\vx)\cdot(\vz-\vx)^\textrm{T}\mB_y\vx\\
&\geq 4f(\vz)-\sqrt{4f(\vz)}\sqrt{\frac{1}{MK^2}\sum_y\left((\vz-\vx)^\textrm{T}\mB_y\vx\right)^2}\,.
\end{split}
\end{align}
We proceed by further lower-bounding the above \eqref{supp:eq:second_term}. Let $\vh=\vz-\vx$. For some $\theta\in(\frac{1}{2},1)$, we have
\begin{align}
\label{supp:eq:second_term_lb1}
\begin{split}
&\theta^2\sum_y\left((\vz-\vx)^\textrm{T}\mB_y(\vz+\vx)\right)^2-\sum_y\left((\vz-\vx)^\textrm{T}\mB_y\vx\right)^2\\
=\ &\theta^2\sum_y\left(\vh^\textrm{T}\mB_y(\vh+2\vx)\right)^2-\sum_y\left(\vh^\textrm{T}\mB_y\vx\right)^2\\
=\ &\theta^2\sum_y\left(\vh^\textrm{T}\mB_y\vh\right)^2+4\theta^2\sum_y\vh^\textrm{T}\mB_y\vh\cdot\vh^\textrm{T}\mB_y\vx+(4\theta^2-1)\sum_y\left(\vh^\textrm{T}\mB_y\vx\right)^2\\
\geq \ &\theta^2\sum_y\left(\vh^\textrm{T}\mB_y\vh\right)^2-4\theta^2\sqrt{\sum_y\left(\vh^\textrm{T}\mB_y\vh\right)^2}\sqrt{\sum_y\left(\vh^T\mB_y\vx\right)^2}+(4\theta^2-1)\sum_y\left(\vh^\textrm{T}\mB_y\vx\right)^2\\
=\ &\left(\theta\sqrt{\sum_y\left(\vh^\textrm{T}\mB_y\vh\right)^2}-2q\sqrt{\sum_y\left(\vh^T\mB_y\vx\right)^2}\right)^2-\left(\sqrt{\sum_y\left(\vh^\textrm{T}\mB_y\vx\right)^2}\right)^2\\
=\ &\left(\theta\sqrt{\sum_y\left(\vh^\textrm{T}\mB_y\vh\right)^2}-(2\theta-1)\sqrt{\sum_y\left(\vh^T\mB_y\vx\right)^2}\right)\left(\theta\sqrt{\sum_y\left(\vh^\textrm{T}\mB_y\vh\right)^2}-(2\theta+1)\sqrt{\sum_y\left(\vh^T\mB_y\vx\right)^2}\right)\,.
\end{split}
\end{align}
To make \eqref{supp:eq:second_term_lb1} greater than $0$, either of the following two inequalities should hold:
\begin{align}
\label{supp:eq:second_term_lb1_one}
\sqrt{\sum_y\left(\vh^\textrm{T}\mB_y\vh\right)^2}&>(2+\frac{1}{\theta})\sqrt{\sum_y\left(\vh^T\mB_y\vx\right)^2}\\
\label{supp:eq:second_term_lb1_two}
\sqrt{\sum_y\left(\vh^\textrm{T}\mB_y\vh\right)^2}&<(2-\frac{1}{\theta})\sqrt{\sum_y\left(\vh^T\mB_y\vx\right)^2}\,.
\end{align}

We can obtain an upper bound on $\|\vh\|_2$ to make \eqref{supp:eq:second_term_lb1_two} hold. Specifically, the left-hand side of \eqref{supp:eq:second_term_lb1_two} can be upper bounded via:
\begin{align}
\label{supp:eq:ub_hDh}
\begin{split}
\sum_y\left(\vh^\textrm{T}\mB_y\vh\right)^2&=\|\vh\|_2^4\cdot\sum_y\left(\widehat{\vh}^\textrm{T}\mB_y\widehat{\vh}\right)^2\\
&=\|\vh\|_2^4\cdot\sum_y\|\mB_y\|_{op}^2\cdot\left(\frac{|\widehat{\vh}^T\mB_y\widehat{\vh}|}{\|\mB_y\|_{op}}\right)^2\\
&\leq \|\vh\|_2^4\cdot\sum_y\|\mB_y\|_{op}^2\cdot\left(\frac{|\widehat{\vh}^T\mB_y\widehat{\vh}|}{\|\mB_y\|_{op}}\right)\\
&=\|\vh\|_2^4\cdot\sum_y\|\mB_y\|_{op}\cdot|\widehat{\vh}^T\mB_y\widehat{\vh}|\\
&\leq \|\vh\|_2^4\cdot\sum_y\|\mB_y\|_{op}\cdot |\widehat{\vh}|^T\mB_y|\widehat{\vh}|\\
&=\|\vh\|_2^4\cdot\sum_y\sigma_{\max}(\mB_y)\cdot |\widehat{\vh}|^T\mB_y|\widehat{\vh}|\\
&\leq 2\|\vh\|_2^4\cdot\sum_y|\widehat{\vh}|^T\mB_y|\widehat{\vh}|\\
&=2\|\vh\|_2^2\cdot|\vh|^T\sum_y\mB_y|\vh|\\
&=2\|\vh\|_2^2\cdot|\vh|^T(\vone_\textnormal{mat}+\mI)|\vh|\\
&=2\|\vh\|_2^2\cdot(\|\vh\|_1^2+\|\vh\|_2^2)\\
&\leq 4\|\vh\|_2^2\cdot\|\vh\|_1^2\,,
\end{split}
\end{align}
where $\widehat{\vh}=\frac{1}{\|\vh\|_2}\vh$, $\vone_\textnormal{mat}$ is a matrix of all $1$s and $\mI$ is the identity matrix. The first inequality in \eqref{supp:eq:ub_hDh} is obtained by $|\widehat{\vh}^T\mB_y\widehat{\vh}|=|\langle\widehat{\vh},\mB_y\widehat{\vh}\rangle|\leq\|\widehat{\vh}\|_2\|\mB_y\widehat{\vh}\|_2\leq\|\mB_y\|_{op}$ and hence $\frac{|\widehat{\vh}^T\mB_y\widehat{\vh}|}{\|\mB_y\|_{op}}\leq 1$; the second inequality is obtained by $|\widehat{\vh}^\textrm{T}\mB_y\widehat{\vh}|=|\sum_{ij}A_y(i,j)\hat{h}_i\hat{h}_j|\leq\sum_{ij}A_y(i,j)|\hat{h}_i||\hat{h}_j|=|\widehat{\vh}|^\textrm{T}\mB_y|\widehat{\vh}|$. If we choose the operator norm $\|\cdot\|_{op}$ to be the Euclidean norm, then $\|\mB_y\|_{op}=\sigma_{\max}(\mB_y)\leq 2$; the last inequality is obtained via $\|\vh\|_2\leq\|\vh\|_1$.

The right-hand side of \eqref{supp:eq:second_term_lb1_two} can be low-bounded as:
\begin{align}
\label{supp:eq:lb_hdx}
\begin{split}
\sum_y\left(\vh^T\mB_y\vx\right)^2&=\|\vh\|_1^2\cdot\overline{\vh}^\textrm{T}\left(\sum_y\mB_y\vx{\vx}^\textrm{T}\mB_y^\textrm{T}\right)\overline{\vh}=\|\vh\|_1^2\cdot\overline{\vh}^\textrm{T}\mE\overline{\vh}\\
&\geq\|\vh\|_1^2\cdot\mu_\mE\,,
\end{split}
\end{align}
where $\overline{\vh}=\frac{1}{\|\vh\|_1}\vh$, $\mE=\sum_{y=0}^{M-1}\mB_y\vx{\vx}^T\mB_y^T$ and $\mu_\mE>0$ can be computed using Lemma \ref{supp:LEMMA:E_MIN}. Combining \eqref{supp:eq:second_term_lb1_two}, \eqref{supp:eq:ub_hDh} and \eqref{supp:eq:lb_hdx}, we can see that as long as the following \eqref{supp:eq:h_bd} holds, \eqref{supp:eq:second_term_lb1_two} will also hold.
\begin{align}
\label{supp:eq:h_bd}
\|\vh\|_2<\tau=\left(2-\frac{1}{\theta}\right)\cdot\sqrt{\frac{\mu_\mE}{4}}\,.
\end{align}
The above \eqref{supp:eq:h_bd} guarantees that \eqref{supp:eq:second_term_lb1} is always greater than $0$. 

Combining \eqref{supp:eq:second_term_lb1},\eqref{supp:eq:second_term_lb1_two}, we have
\begin{align}
\label{supp:eq:lb_gradient_align}
-\sqrt{\frac{1}{MK^2}\sum_y\left((\vz-\vx)^\textrm{T}\mB_y\vx\right)^2}>-\theta\sqrt{4f(\vz)}\,.
\end{align}
Plug the above \eqref{supp:eq:lb_gradient_align} into \eqref{supp:eq:second_term}. We have:
\begin{align}
\label{supp:eq:second_term_lb}
\langle\vz-\vx,\, \nabla f(\vz)\rangle > 4(1-\theta)f(\vz)\,.
\end{align}

\item We finally show that there exist some $\{\alpha,\beta\}$ to make the regularity condition $RC(\alpha,\beta,\tau)$ hold. 

Plugging \eqref{supp:eq:second_term_lb}, \eqref{supp:eq:first_term_lb} into \eqref{supp:eq:rc_condition}, we need the following inequality to hold:
\begin{align}
\label{supp:eq:rc_condtion_c1}
    4(1-\theta)f(\vz)\geq\frac{1}{\alpha}\|\vh\|_2^2+\frac{1}{\beta}\frac{16N}{K^2}f(\vz)\,.
\end{align}
Combining \eqref{supp:eq:lb_gradient_align} and \eqref{supp:eq:lb_hdx}, we further have
\begin{align}
\label{supp:eq:lb_obj_fun}
f(\vz)>\frac{1}{4\theta^2}\frac{1}{MK^2}\sum_y\left(\vh^\textrm{T}\mB_y\vx\right)^2\geq\frac{1}{4\theta^2}\frac{1}{MK^2}\|\vh\|_1^2\mu_\mE\geq\frac{1}{4\theta^2}\frac{1}{MK^2}\|\vh\|_2^2\mu_\mE\,.
\end{align}
Plugging \eqref{supp:eq:lb_obj_fun} into \eqref{supp:eq:rc_condtion_c1}, we then need the following inequality to hold:
\begin{align}
\label{supp:eq:rc_condition_c2}
    \left((1-\theta)-\frac{1}{\beta}\frac{4N}{K^2}\right)\frac{1}{\theta^2}\frac{1}{MK^2}\mu_\mE\geq\frac{1}{\alpha}\,.
\end{align}
The constants $\theta$ and $\{\alpha,\beta\}$ that satisfy \eqref{supp:eq:rc_condition_c2} can be chosen in the following order:
\begin{enumerate}
    \item Choose some $\theta\in(\frac{1}{2}, 1)$.
    \item Fix $\theta$, choose some $\beta>\frac{4N}{(1-\theta)K^2}$.
    \item Fix $\theta,\beta$, choose some $\alpha\geq\left((1-\theta)-\frac{1}{\beta}\frac{4N}{K^2}\right)^{-1}\frac{\theta^2MK^2}{\mu_\mE}$.
\end{enumerate}

\end{itemize}
We can get that the regularity condition $RC(\alpha,\beta,\tau)$ holds and $\tau=(2-\frac{1}{\theta})\sqrt{\frac{\mu_\mE}{4}}$.

\paragraph{Step 2:} We use $\overline{\vz}_{+1}=\vz-\eta\nabla f(\vz)$ to denote one gradient descent update, and $\vz_{+1}=\mathscr{P}_\mathcal{S}(\overline{\vz}_{+1})\in\mathcal{S}$ to denote one projected gradient descent update.

Let $\vs$ be a linear combination of $\vz_{+1}$ and a global optimizer $\vx$ such that
\begin{align}
\label{supp:eq:linear_comb}
\vx-\vz_{+1} = a(\vz_{+1}-\vs)\,,
\end{align}
where $a\in\mathbb{R}$, $a\neq 0$ is some constant. We can always find an $\vs$ such that the following holds,
\begin{align}
\label{supp:eq:perpendicular}
(\vs-\vz)^T(\vs-\vz_{+1})=0\,.
\end{align}
\begin{enumerate}[label={\arabic*.}]
\item If $\overline{\vz}_{+1}=\vz_{+1}$, then $\vz\in\mathcal{S}$ and \eqref{supp:eq:converge} holds.
\item If $\vx=\vz_{+1}$, \eqref{supp:eq:converge} naturally holds.
\item Otherwise, we can choose $a=\frac{\|\vx-\vz_{+1}\|_2^2}{\left(\vz_{+1}-\overline{\vz}_{+1}\right)^T\left(\vx-\vz_{+1}\right)}$. From \eqref{supp:eq:perpendicular}, we can get:
\begin{align}
\label{supp:eq:sum_four}
\|\vs\|_2^2-\overline{\vz}_{+1}^T\vs-\vs^T\vz_{+1} = -\overline{\vz}_{+1}^T\vz_{+1}\,.
\end{align}
We also have
\begin{align}
\label{supp:eq:sum_one}
\|\overline{\vz}_{+1}-\vz_{+1}\|_2^2 &= \|\overline{\vz}_{+1}\|_2^2 + \|\vz_{+1}\|_2^2-2\overline{\vz}_{+1}^T\vz_{+1}\\
\label{supp:eq:sum_two}
\|\overline{\vz}_{+1}-\vs\|_2^2&=\|\overline{\vz}_{+1}\|_2^2+\|\vs\|_2^2-2\overline{\vz}_{+1}^T\vs\\
\label{supp:eq:sum_three}
\|\vs-\vz_{+1}\|_2^2 &= \|\vs\|_2^2 + \|\vz_{+1}\|_2^2 - 2\vs^T\vz_{+1}\,.
\end{align}
Combining \eqref{supp:eq:sum_four}-\eqref{supp:eq:sum_three}, we get that
\begin{align}
\label{supp:eq:sum_five}
\|\overline{\vz}_{+1}-\vz_{+1}\|_2^2=\|\overline{\vz}_{+1}-\vs\|_2^2+\|\vs-\vz_{+1}\|_2^2\,.
\end{align}
Using \eqref{supp:eq:linear_comb}, we have
\begin{align}
\label{supp:eq:sz_xs}
    \vs-\vz_{+1} = \frac{1}{a+1}(\vs-\vx)\,.
\end{align}
Plug \eqref{supp:eq:sz_xs} into \eqref{supp:eq:perpendicular}. We have
\begin{align}
    (\vs-\overline{\vz}_{+1})^T(\vs-\vx)=0\,.
\end{align}
Similarly, we can get that
\begin{align}
\label{supp:eq:sum_six}
\|\overline{\vz}_{+1}-\vx\|_2^2=\|\overline{\vz}_{+1}-\vs\|_2^2+\|\vs-\vx\|_2^2\,.
\end{align}

\begin{enumerate}[label*={\arabic*)}]
\item If $\vs\in\mathcal{S}$, since $\vz_{+1}$ is the projection of $\overline{\vz}_{+1}$ in $\mathcal{S}$, we have $\|\overline{\vz}_{+1}-\vz_{+1}\|_2^2\leq\|\overline{\vz}_{+1}-\vs\|_2^2$. Using \eqref{supp:eq:sum_five}, we have:
\begin{align}
\|\vs-\vz_{+1}\|_2^2=0\,.
\end{align}
Hence $\vs$ and $\vz_{+1}$ is the same point. From \eqref{supp:eq:sum_six}, we can get:
\begin{align}
\|\overline{\vz}_{+1}-\vx\|_2^2=\|\overline{\vz}_{+1}-\vz_{+1}\|_2^2+\|\vz_{+1}-\vx\|_2^2\,.
\end{align}
Since $\overline{\vz}_{+1}\notin\mathcal{S}$, we have $\|\overline{\vz}_{+1}-\vz_{+1}\|_2^2>0$. Hence $\|\overline{\vz}_{+1}-\vx\|_2^2>\|\vz_{+1}-\vx\|_2^2$.
\item If $\vs\notin\mathcal{S}$, we have:
\begin{align}
\label{supp:eq:sum_seven}
\begin{split}
\|\vs-\vx\|_2^2&=\|\vs-\vz_{+1}+\vz_{+1}-\vx\|_2^2\\
&=\|\vs-\vz_{+1}\|_2^2+\|\vz_{+1}-\vx\|_2^2+2\left(\vs-\vz_{+1}\right)^\textrm{T}\left(\vz_{+1}-\vx\right)\,.
\end{split}
\end{align}
\begin{itemize}
\item If $a\in[0, \infty)$, from \eqref{supp:eq:linear_comb}, we have $\left(\vs-\vz_{+1}\right)^\textrm{T}\left(\vz_{+1}-\vx\right)\geq 0$. From \eqref{supp:eq:sum_seven}, we have $\|\vs-\vx\|_2^2\geq\|\vz_{+1}-\vx\|_2^2$. Using \eqref{supp:eq:sum_six}, we have $\|\overline{\vz}_{+1}-\vx\|_2^2\geq\|\vz_{+1}-\vx\|_2^2$.
\item If $a\in(-1,0)$, from \eqref{supp:eq:linear_comb}, we have $\vz_{+1}-\vx=\frac{a}{-1-a}(\vx-\vs)$. Since $\frac{a}{-1-a}>0$, $(\vz_{+1}-\vx)^T(\vx-\vs)>0$. We then have:
\begin{align}
\begin{split}
\|\vz_{+1}-\vs\|_2^2&=\|\vz_{+1}-\vx+\vx-\vs\|_2^2\\
&=\|\vz_{+1}-\vx\|_2^2+\|\vx-\vs\|_2^2+2(\vz_{+1}-\vx)^T(\vx-\vs)\\
&>\|\vx-\vs\|_2^2\,.
\end{split}
\end{align}
Using \eqref{supp:eq:sum_five} and \eqref{supp:eq:sum_six}, we have:
\begin{align}
\|\overline{\vz}_{+1}-\vz_{+1}\|_2^2>\|\overline{\vz}_{+1}-\vx\|_2^2\,.
\end{align}
This is in contradiction with the assumption that $\vz_{+1}$ is the projection of $\vz$ in $\mathcal{S}$ so that $\vz_{+1}$ is closest point in $\mathcal{S}$ to $\vz$ in terms of $l_2$ norm: $\|\overline{\vz}_{+1}-\vz_{+1}\|_2^2\leq\|\overline{\vz}_{+1}-\vx\|_2^2$, hence $a\notin(-1,0)$.
\item If $a\in(-\infty, -1]$, from \eqref{supp:eq:linear_comb}, we have $\vs=-\frac{1}{a}\vx+(1+\frac{1}{a})\vz_{+1}$. Since $-\frac{1}{a}\in(0,1]$, $\vs\in\mathcal{S}$. This is in contradiction with the assumption $\vs\notin\mathcal{S}$, hence $a\notin(-\infty,1]$.
\end{itemize}
\end{enumerate}
In summary, we have
\begin{align}
\|\vz_{+1}-\vx\|_2^2\leq\|\overline{\vz}_{+1}-\vx\|_2^2\,.
\end{align}
We can use the regularity condition $RC(\alpha,\beta,\tau)$ and get the following 
\begin{align}
    \|\vz_{t+k}-\vx\|_2^2<(1-\frac{2\eta}{\alpha})^k\cdot\|\vz_t-\vx\|_2^2\,.
\end{align}
\end{enumerate}
\end{proof}

\section{Beltway: Convergence and Difficulty of Recovery}
\subsection{Convergence Theorem ~\ref{supp:THM:CVG_SED_BW}}
\label{supp:proof:thm:cvg_sed_bw}

For the beltway problem we have a similar theorem for the convergence analysis.
\begin{theorem}
\label{supp:THM:CVG_SED_BW}
In the noiseless case, let $\vh=\vz-\vx$ and $\mS_y=\mR_y+\mR_y^T$. If $\vz$ satisfies
\begin{align}
\label{supp:eq:absolute_radius}
\|\vh\|_2=\|\vz-\vx\|_2<\tau=(2-\frac{1}{\theta})\cdot\sqrt{\frac{\mu_\mF}{4}}\,,
\end{align}
where $\theta\in\big(\frac{1}{2},1\big)$ is some fixed constant and $\mu_\mF>0$ depends on the matrix $\mF=\sum_{y=0}^{M-1}\mS_y\vx\vx^T\mS_y^T$, 
\begin{enumerate}
\item There exists a choice of $\{\alpha>0,\beta>0\}$ such that the regularity condition $RC(\alpha,\beta,\tau)$ holds.
\item Under this choice of parameters $\{\alpha,\beta,\tau\}$, if $\|\vh_t\|_2=\|\vz_t-\vx\|_2<\tau$ and the step size $\eta\in(0,\frac{2}{\beta}]$, the projected gradient descent update in \eqref{supp:eq:pgd_update} converges linearly to $\vx$:
\begin{equation}
\label{supp:eq:converge}
\|\vz_{t+k}-\vx\|_2^2<(1-\frac{2\eta}{\alpha})^k\cdot\|\vz_t-\vx\|_2^2\,.
\end{equation}
\end{enumerate}
\end{theorem}
The proof can be derived in a similar fashion to the one in the turnpike case by simply replacing $\{\mA_y,\mB_y,\mE\}$ with $\{\mR_y,\mS_y,\mF\}$.

\subsection{Analysis on Difficulty of Recovery}
The mutual information $I(X;Y)$ between the point density $X$ and the distance $Y$ can be derived in a similar way. Note that the probability distribution of the distance $Y$ in the beltway problem is different from that of the turnpike problem. We define the following two random variables:
\begin{definition}
The point density $X\in\{1,\cdots,M\}$ is a random variable with distribution $P(X=m)=\frac{1}{N}x_m$, where $x_m$ is the $m$-th entry of the ground truth signal $\vx$, and $P(X=m)$ corresponds to the normalized point density at the $m$-th segment $l_m$ in the 1D discrete domain.
\end{definition}

\begin{definition}
The distance $Y\in\{0,\cdots,M-1\}$ is a random variable with conditional distribution
\begin{align}
\begin{split}
P(Y=y|X=m)=&\frac{1}{2}\sum_{k=1}^MP(X=k)\cdot\Big(\delta(y=k-m)+\delta(M-y=m-k)\Big)\\
&+\frac{1}{2}\sum_{k=1}^MP(X=k)\cdot\Big(\delta(y=M-k+m)+\delta(y=m-k)\Big)\,,
\end{split}
\end{align}
where ``$\frac{1}{2}$'' is due to the observation that a distance $Y$ could be either from $x_m$ to $x_k$ or from $x_k$ to $x_m$ with probability $\frac{1}{2}$.
Its marginal distribution is then
\begin{align}
\label{supp:eq:it_analysis_distance_dist}
\begin{split}
P(Y=y)=&\sum_{m=1}^MP(X=m)P(Y=y|X=m)\\
=&\frac{1}{2N^2}\sum_{m=1}^M\sum_{k=1}^M x_mx_k\cdot\Big(\delta(y=k-m)+\delta(M-y=m-k)\Big)\\
&+\frac{1}{2N^2}\sum_{m=1}^M\sum_{k=1}^M x_mx_k\cdot\Big(\delta(y=M-k+m)+\delta(y=m-k)\Big)\\
=&\frac{1}{N^2}\sum_{m=1}^M\sum_{k=1}^M x_mx_k\cdot\Big(\delta(y=k-m)+\delta(M-y=m-k)\Big)\,.
\end{split}
\end{align}
\end{definition}
The mutual information $I(X;Y)$ can then be computed as follows:
\begin{align}
\begin{split}
    I(X;Y)&=H(Y)-H(Y|X)\\
    &=H(Y)-\sum_{m=1}^MP(X=m)\cdot H(Y|X=m)\,.
\end{split}    
\end{align}

\end{document}